\def\paperversion{tr}
\documentclass[a4paper, cleveref, thm-restate,  USenglish]{authors-lipics-v2019/lipics-v2019}
\def\anonymoussubmission{1} 
\newif\ifsinglecolumn\singlecolumnfalse
\newif\ifwidemargins\widemarginsfalse
\newif\ifwarning\warningfalse
\newif\ifshowcomments\showcommentsfalse
\newif\ifblinded\blindedfalse
\newif\ifshowlinenums\showlinenumsfalse
\newif\ifreport\reportfalse
\newif\ifcopyrightspace\copyrightspacefalse
\newif\ifacknowledgments\acknowledgmentsfalse
\newif\ifshowpagenumbers\showpagenumberstrue
\newif\iffinalformat\finalformatfalse
\newif\ifweb\webfalse
\newif\ifexternalize\externalizetrue

\ifx\anonymoussubmission\undefined
\else
  \if\anonymoussubmission 1
    \blindedtrue
  \else
    \blindedfalse
  \fi
\fi
\IfFileExists{.blinded}{\blindedtrue}

\ifx\paperversion\xxxxundefined
\PackageError{paperversions}{*** No valid document version was specified.
Macro paperversions must be defined as one of (markup, draft,
local, submission, final, web, tr, trdraft)}
\fi

\def\xxversion{\csname xx\paperversion\endcsname}
\newif\ifsawversion\sawversionfalse

\ifcase\xxversion\relax
    \widemarginstrue
    \singlecolumntrue
    \warningtrue
    \showlinenumstrue
    \sawversiontrue
    \blindedfalse
\or 
    \warningtrue
    \showlinenumsfalse
    \showcommentstrue
    \sawversiontrue
\or 
    \warningtrue
    \showcommentsfalse
    \blindedfalse
    \sawversiontrue
    \acknowledgmentstrue
    \showlinenumstrue
\or 
    \sawversiontrue
    \showlinenumstrue
\or 
    \blindedfalse
    \sawversiontrue
    \copyrightspacetrue
    \acknowledgmentstrue
    \showpagenumbersfalse
    \finalformattrue
\or 
    \singlecolumntrue
    \blindedfalse
    \sawversiontrue
    \reporttrue
    \acknowledgmentstrue
    \webtrue
\or 
    \blindedfalse
    \singlecolumntrue
    \showcommentstrue
    \sawversiontrue
    \reporttrue
    \showlinenumstrue
    \warningtrue
\or 
    \blindedfalse
    \showcommentstrue
    \copyrightspacetrue
    \acknowledgmentstrue
    \sawversiontrue
    \finalformattrue
    \showlinenumstrue
    \warningtrue
\or 
    \blindedfalse
    \sawversiontrue
    \copyrightspacetrue
    \acknowledgmentstrue
    \finalformattrue
    \webtrue
\or 
    \singlecolumntrue
    \blindedtrue
    \acknowledgmentsfalse
    \sawversiontrue
    \reporttrue
    \webtrue
\or 
    \blindedfalse
    \sawversiontrue
    \copyrightspacetrue
    \acknowledgmentstrue
    \finalformattrue
    \externalizefalse
    \webtrue
\fi

\ifsawversion
\else
\fi

\let\xxversion=\undefined

\ifreport
\else
  \documentclass[a4paper, cleveref, \ifblinded anonymous,\fi USenglish]{authors-lipics-v2019/lipics-v2019}
\fi
\iffinalformat
    \ifweb
    \else
    \fi
\fi

\ifshowlinenums
\else
  \nolinenumbers
\fi
 
\usepackage{amsmath,
            amssymb,
            color,
            colortbl,
            eso-pic,
            graphicx,
            xr,
            xr-hyper,
            hyperref,
            ifthen,
            mfirstuc,
            stmaryrd,
            tikz,
            thmtools,
            tex-macros/utf8math,
            xcolor,
            xspace}
\usetikzlibrary{patterns}
\hypersetup{
	colorlinks,
	urlcolor=black,
	citecolor=black,
	filecolor=black,
	linkcolor=black
}

\ifreport\else
\crefname{section}{\S}{\S}
\Crefname{section}{\S}{\S}
\crefname{subsection}{\S}{\S}
\Crefname{subsection}{\S}{\S}
\crefname{subsubsection}{\S}{\S}
\Crefname{subsubsection}{\S}{\S}
\fi

\usepackage{inconsolata}
\usepackage{listings}
  \lstdefinestyle{pseudocode}{
      language        = Python,
      basicstyle      = \footnotesize\ttfamily,
      backgroundcolor = \color{white},
      keywordstyle    = \color{blue},
      stringstyle     = \color{green},
      commentstyle    = \color{red}\ttfamily,
      morekeywords    = {new, atomic, has},
      escapeinside    = {@}{@},
      numbers=left,
      numbersep=5pt, 
      numberstyle=\tiny\color{gray},
  }

\renewcommand{\floatpagefraction}{0.75}

\ifreport
\title{Heterogeneous Paxos: Technical Report}
\else
\title{Heterogeneous Paxos}
\fi

\ifblinded
\else
\author{\href{https://isaacsheff.com}{Isaac Sheff}}
       {  Max Planck Institute for Software Systems,
           Campus E1 5,
           Room 531,
           66121 Saarbr{\"u}cken,
           Germany 
         \ifreport\and \url{https://IsaacSheff.com}\fi}
       {isheff@mpi-sws.org}
       {https://orcid.org/0000-0002-7822-1503}
       {} 

\author{\href{https://www.cs.cornell.edu/~xinwen/}{Xinwen Wang}}
       {  Cornell University,
            Gates Hall,
            107 Hoy Road,
            Ithaca, New York, 14853,
            USA
          \ifreport\and \url{https://www.cs.cornell.edu/~xinwen/}\fi}
       {xinwen@cs.cornell.edu}
       {https://orcid.org/0000-0003-2958-6589}
       {} 

\author{\href{https://www.cs.cornell.edu/home/rvr/}{Robbert van Renesse}}
       {  Cornell University,
            433 Gates Hall,
            107 Hoy Road,
            Ithaca, New York, 14853,
            USA
          \ifreport\and \url{https://www.cs.cornell.edu/home/rvr/}\fi}
       {rvr@cs.cornell.edu}
       {https://orcid.org/0000-0003-3598-0283}
       {} 

\author{\href{https://www.cs.cornell.edu/andru/}{Andrew C. Myers}}
       {  Cornell University,
            428 Gates Hall,
            107 Hoy Road,
            Ithaca, New York, 14853,
            USA
          \ifreport\and \url{https://www.cs.cornell.edu/andru/}\fi}
       {andru@cs.cornell.edu}
       {https://orcid.org/0000-0001-5819-7588}
       {} 
\fi

\authorrunning{Isaac Sheff, Xinwen Wang, Robbert van Renesse, and Andrew C. Myers } 

\Copyright{Isaac Sheff, Xinwen Wang, Robbert van Renesse, and Andrew C. Myers} 

\ifwarning
\AtBeginDocument{
 \AddToShipoutPicture*{\put(306,690){\it\color{red}{\fbox{\large Draft---please do not distribute}}}}
}
\fi

\ifreport
\else
\fi

\ifreport\else
\relatedversion{
  \ifblinded
     Anonymized Technical Report in Appendix
  \else
     Technical Report at~\url{https://arxiv.org/abs/2011.08253}~\cite{hetconsTR}.
  \fi
  }
\fi



\ifreport
\hideLIPIcs
\EventShortTitle{Technical Report}
\else
\EventLongTitle{Conference on Principles of Distributed Systems (OPODIS 2020)}
\EventShortTitle{OPODIS 2020}
\EventAcronym{OPODIS}
\EventYear{2020}
\EventDate{14 -- 16 December 2020}
\EventLocation{Online}
\ArticleNo{81}
\fi

\newcommand{\p}[1]{{\ensuremath{\left({{#1}}\right)}}}
\newcommand{\cb}[1]{{\left\{{{#1}}\right\}}}

\newcommand{\abs}[1]{{\left|{{#1}}\right|}}

\newcommand{\ceil}[1]{{\ensuremath{\left\lceil{{#1}}\right\rceil}}}

\newcommand{\tb}[1]{{\textrm{\textbf{{#1}}}}}

\newcommand{\tallpipe}[2]{{\ensuremath{\begin{array}{r|l}{{#1}}&{{#2}}\end{array}}}}

\newcommand{\crash}{crash\xspace}
\newcommand{\byzantine}{Byzantine\xspace}

\newcommand{\colort}[2]{{\color{#1}{#2}}}
\newcommand{\red}[1]{\colort{red}{#1}}
\newcommand{\blue}[1]{\colort{blue}{#1}}
\definecolor{DarkGreen}{rgb}{0,0.5,0}
\newcommand{\green}[1]{{\colort{DarkGreen}{#1}}}
\newcommand{\purple}[1]{{\colort{purple}{#1}}}
\newcommand{\orange}[1]{{\colort{orange}{#1}}}

\newcommand{\clg}[0]{{\textrm{CLG}}\xspace}
\newcommand{\edge}[2]{{{{#1}}\!\!-\!\!{{#2}}}}

\newcommand{\reallysafe}[0]{\ensuremath{\mathcal{S}}\xspace}
\newcommand{\reallylive}[0]{\ensuremath{\mathcal{L}}\xspace}
\newcommand{\entangled}[2]{{\textrm{Entangled}\p{{{#1}}, {{#2}}}}}

\newcommand{\accurate}[1]{{\textrm{Accurate}\p{{{#1}}}}}
\newcommand{\terminating}[1]{{\textrm{Terminating}\p{{{#1}}}}}

\newcommand{\sig}[1]{{\ensuremath{\textrm{Sig}\p{#1}}}}
\newcommand{\tran}[1]{{\ensuremath{\textrm{Tran}\p{#1}}}}
\newcommand{\geta}[1]{{\ensuremath{\textrm{Get1a}\p{#1}}}}
\newcommand{\ba}[1]{{\ensuremath{\textrm{b}\p{#1}}}}
\newcommand{\caught}[1]{{\ensuremath{\textrm{Caught}\p{#1}}}}
\newcommand{\con}[2]{{\ensuremath{\textrm{Con}_{{#1}}\p{{#2}}}}}
\newcommand{\qa}[1]{{\ensuremath{\textrm{q}\p{{#1}}}}}
\newcommand{\fresh}[2]{{\ensuremath{\textrm{fresh}_{{#1}}\p{{#2}}}}}
\newcommand{\buried}[2]{{\ensuremath{\textrm{Buried}\p{{{#1}}, {{#2}}}}}}
\newcommand{\cona}[2]{{\ensuremath{\textrm{Con2as}_{{#1}}\p{{#2}}}}}

\newcommand{\va}[1]{{\ensuremath{\textrm{V}\p{{#1}}}}}

\newcommand{\decision}[2]{{\ensuremath{\textrm{Decision}_{{#1}}\p{{#2}}}}}

\DeclareMathOperator*{\argmax}{argmax}

\newcommand{\andlinesTwo}[2]{{  \ensuremath{\begin{array}{r l} & {{#1}} \\ \land & {{#2}} \end{array}}}}
\newcommand{\andlinesThree}[3]{{\ensuremath{\begin{array}{r l} & {{#1}} \\ \land & {{#2}} \\ \land & {{#3}} \end{array}}}}
\newcommand{\andlinesFour}[4]{{ \ensuremath{\begin{array}{r l} & {{#1}} \\ \land & {{#2}} \\ \land & {{#3}} \\ \land & {{#4}} \end{array}}}}
\newcommand{\andlinesFive}[5]{{ \ensuremath{\begin{array}{r l} & {{#1}} \\ \land & {{#2}} \\ \land & {{#3}} \\ \land & {{#4}} \\ \land & {{#5}} \end{array}}}}
\newcommand{\andlinesSix}[6]{{  \ensuremath{\begin{array}{r l} & {{#1}} \\ \land & {{#2}} \\ \land & {{#3}} \\ \land & {{#4}} \\ \land & {{#5}} \\ \land & {{#6}} \end{array}}}}

\newcommand\membershipdisagreement[5]{
\node[circle, fill=blue,  inner sep=#4 * 3mm, draw=blue,  line width=0mm](A#3) at (0 + #1,#5 + #2) {\ };
\node[circle, fill=black, inner sep=#4 * 3mm, draw=black, line width=0mm](B#3) at (#5 + #1,0 + #2) {\ };
\node[circle, fill=black, inner sep=#4 * 3mm, draw=black, line width=0mm](C#3) at (#5 + #1,#5 + #2) {\ };
\node[circle, fill=black, inner sep=#4 * 3mm, draw=black, line width=0mm](D#3) at (#5 + #1,2 * #5 + #2) {\ };
\node[circle, fill=red!30, inner sep=#4 * 2.57mm, draw=red, line width=#4 * 1.3mm](E#3) at (2 * #5 + #1,#5 + #2) {\ };
}

\newcommand\membershipdisagreementfailure[8]{
\pgfdeclareimage[width=#6]{devil}{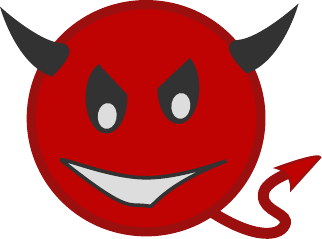}

\ifthenelse{\equal{#7}{A}}
           {\ifthenelse{\equal{#8}{byz}}{\node[] (Abyz#3) at (0 + #1,#5 + #2) {\pgfuseimage{devil}};}{}}
           {}
          
\ifthenelse{\equal{#7}{B}}
           {\ifthenelse{\equal{#8}{byz}}{\node[] (Bbyz#3) at (#5 + #1,0 + #2) {\pgfuseimage{devil}};}{}}
           {}

\ifthenelse{\equal{#7}{C}}
           {\ifthenelse{\equal{#8}{byz}}{\node[] (Cbyz#3) at (#5 + #1,#5 + #2) {\pgfuseimage{devil}};}{}}
           {}

\ifthenelse{\equal{#7}{D}}
           {\ifthenelse{\equal{#8}{byz}}{\node[] (Dbyz#3) at (#5 + #1,2 * #5 + #2) {\pgfuseimage{devil}};}{}}
           {}

\ifthenelse{\equal{#7}{E}}
           {\ifthenelse{\equal{#8}{byz}}{\node[] (Ebyz#3) at (2 * #5 + #1,#5 + #2) {\pgfuseimage{devil}};}{}}
           {}
}

\newcommand\introexamplenodes[2]{
\node[circle, fill=blue, inner sep=1.2mm, draw=blue, line width=0mm](LA) at (#1-0.9,#2) {\ };
\node[circle, fill=blue, inner sep=1.2mm, draw=blue, line width=0mm](LB) at (#1-0.9,#2+0.45) {\ };
\node[circle, fill=blue, inner sep=1.2mm, draw=blue, line width=0mm](LC) at (#1-0.9,#2+0.9) {\ };
\node[circle, fill=black, inner sep=1.2mm, draw=black, line width=0mm](LD) at (#1-.45,#2) {\ };
\node[circle, fill=black, inner sep=1.2mm, draw=black, line width=0mm](LE) at (#1-.45,#2+0.45) {\ };
\node[circle, fill=black, inner sep=1.2mm, draw=black, line width=0mm](LF) at (#1-.45,#2+0.9) {\ };
\node[circle, fill=red!30, inner sep=1mm, draw=red, line width=.4mm](LG) at (#1,#2) {\ };
\node[circle, fill=red!30, inner sep=1mm, draw=red, line width=.4mm](LH) at (#1,#2+0.45) {\ };
\node[circle, fill=red!30, inner sep=1mm, draw=red, line width=.4mm](LI) at (#1,#2+0.9) {\ };
}
\newcommand\introexamplefault[4]{
\ifthenelse{\equal{#3}{1}}{\node[]() at (#1-0.9,#2+0.9) {\pgfuseimage{#4}}}{};
\ifthenelse{\equal{#3}{2}}{\node[]() at (#1-0.45,#2+0.9) {\pgfuseimage{#4}}}{};
\ifthenelse{\equal{#3}{3}}{\node[]() at (#1,#2+0.9) {\pgfuseimage{#4}}}{};
\ifthenelse{\equal{#3}{4}}{\node[]() at (#1-0.9,#2+0.45) {\pgfuseimage{#4}}}{};
\ifthenelse{\equal{#3}{5}}{\node[]() at (#1-0.45,#2+0.45) {\pgfuseimage{#4}}}{};
\ifthenelse{\equal{#3}{6}}{\node[]() at (#1,#2+0.45) {\pgfuseimage{#4}}}{};
\ifthenelse{\equal{#3}{7}}{\node[]() at (#1-0.9,#2) {\pgfuseimage{#4}}}{};
\ifthenelse{\equal{#3}{8}}{\node[]() at (#1-0.45,#2) {\pgfuseimage{#4}}}{};
\ifthenelse{\equal{#3}{9}}{\node[]() at (#1,#2) {\pgfuseimage{#4}}}{};
}

\newcommand\block{block\xspace}
\newcommand\blocks{blocks\xspace}

\newcommand\attestations{attestations\xspace}

\newcommand\integrityattestations{integrity \attestations}

\definecolor{LightBlue}{rgb}{.9,.9,1}

\newcommand{\reducedstrut}{%
  \vrule width 0pt height .9\ht\strutbox depth .9\dp\strutbox\relax}
\newcommand{\hetdifftext}[1]{%
  \begingroup
  \setlength{\fboxsep}{0pt}%
  \colorbox{blue!10}{\reducedstrut#1\/}%
  \endgroup
}
\newcommand{\hetdiff}[1]{\hetdifftext{$\displaystyle#1$}}

\ifblinded
\newcommand{\charlotteURL}{https://bit.ly/AnonymizedCharlotteCode}
\else
\newcommand{\charlotteURL}{https://github.com/isheff/charlotte-public}
\fi
\supplement{Implementation source: \url{\charlotteURL}}

\usepackage{checkpagelimit}
\usepackage{tightheaders}
\robustify
\robustify\addtocounter
\robustify
\robustify

\ifreport \else
\let\localcref\cref
\let\localCref\Cref

\newcommand{\techReportName}{\ifblinded{the Appendix}\else\cite{hetconsTR}\fi}

\newcommand{\tref}[1]{\localcref{#1} of~\techReportName}
\newcommand{\Tref}[1]{\localCref{#1} of~\techReportName}

\makeatletter
\renewcommand{\cref}[1]{%
  \@ifundefined{r@#1}{%
    \tref{TR-#1}%
  }{%
    \localcref{#1}%
  }%
}
\renewcommand{\Cref}[1]{%
  \@ifundefined{r@#1}{%
    \Tref{TR-#1}%
  }{%
    \localCref{#1}%
  }%
}
\makeatother
\fi

\begin{document}
\pdfoutput=1
\maketitle
\begin{abstract}
  In distributed systems, a group of \textit{learners} achieve
\textit{consensus} when, by observing the output of some
\textit{acceptors}, they all arrive
 at the same value.
Consensus is crucial for ordering transactions in failure-tolerant
 systems.
Traditional consensus algorithms are homogeneous in three ways:
\begin{itemize}
\item  all learners are treated equally,
\item  all acceptors are treated equally, and
\item  all failures are treated equally.
\end{itemize}
These assumptions, however, are unsuitable for cross-domain
 applications, including blockchains, where not all acceptors are
 equally trustworthy, and not all learners have the same assumptions
 and priorities.
We present the first consensus algorithm to be heterogeneous in all three
 respects.
Learners set their own mixed failure tolerances over
   differently trusted sets of acceptors.
We express these assumptions in a novel \textit{Learner Graph}, and
 demonstrate sufficient conditions for consensus.

We present \textit{Heterogeneous Paxos}, an extension of Byzantine
 Paxos.
Heterogeneous Paxos achieves consensus for any viable
 Learner Graph in best-case three message sends, which is optimal.
We present a proof-of-concept implementation and
 demonstrate how tailoring for heterogeneous scenarios can save
 resources and reduce latency.

\end{abstract}

\ifreport
\newpage
\fi

\section{Introduction}
\label{sec:intro}
The rise of blockchain systems has renewed interest in the classic
 problem of consensus, but traditional consensus protocols are not
 designed for the highly decentralized, heterogeneous environment
 of blockchains.
In a Consensus protocol, processes called
 \textit{learners} try to decide on the same value, based on the
 outputs of some set of processes called \textit{acceptors},
 some of whom may fail.
(In our model, learners send no messages, and so they cannot fail.)
Consensus is a vital part of any fault-tolerant system maintaining
 strongly consistent state, such as
 Datastores~\cite{corbett2013, calder2011},
 Blockchains~\cite{bitcoin,ethereum,ScalingDecentralizedBlockchains},
 or indeed anything which orders transactions.
Traditionally, consensus protocols have been \textit{homogeneous}
 along three distinct dimensions:
\begin{itemize}
\item Homogeneous acceptors.
      Traditional systems tolerate some number $f$ of
       failed acceptors, but acceptors are interchangeable.
      Prior work including
       ``failure-prone sets''~\cite{Malkhi97a,survivor-sets}
       explores \textit{heterogeneous} acceptors.
\item Homogeneous failures.
      Systems are traditionally designed to tolerate either purely
       \byzantine or purely \crash failures.
      There is no distinction between failure
       scenarios in which the same acceptors fail, but possibly in
       different ways.
      However, some projects have explored \textit{heterogeneous}, or
       ``mixed'' failures~\cite{Siu1998,clement2009upright,Liu2015XFTPF}.
\item Homogeneous learners.
      All learners make the same assumptions, so system
       guarantees apply either to all learners, or to none.
      Systems with \textit{heterogeneous} learners include
       Cobalt~\cite{cobalt} and
       Stellar~\cite{mazieresstellar,StellarDISC,GarcaPrez2018FederatedBQ}.
\end{itemize}

Blockchain systems can violate homogeneity on all three dimensions.
Permissioned blockchain systems like Hyperledger~\cite{hyperledger},
 J.P. Morgan's Quorum~\cite{quorumWP}, and R3's Corda~\cite{corda}
 exist specifically to facilitate atomic transactions between mutually
 distrusting businesses. 
A crucial part of setting up any implementation has been settling on a
 set of equally trustworthy, failure-independent acceptors. 
These setups are complicated by the reality that different parties
 make different assumptions about whom to trust, and how.

Defining heterogeneous consensus poses challenges
 not covered by homogeneous definitions,
 particularly with respect to learners.
How should learners express their failure tolerances?
When different learners expect different possible failures, when do
 they need to agree? 
If a learner's failure assumptions are wrong, does it have
 any guarantees?
No failure models developed for one or two dimensions of heterogeneity
 easily compose to describe all three.

Failure models developed for one or two dimensions of heterogeneity
do not easily compose to describe all three, but
 our new trust model, the
 \textit{Learner Graph}~(\cref{sec:learnergraph}),
 can express the precise trust assumptions of learners in terms of
 diverse acceptors and failures.
Compared to trying to find a homogeneous setup agreeable to all
 learners, finding a learner graph for which consensus is possible is
 strictly more permissive.
In fact, the learner graph is substantially more expressive than the models
 used in prior heterogeneous learner consensus work, including Stellar's
 \textit{slices}~\cite{mazieresstellar} or Cobalt's
 \textit{essential subsets}~\cite{cobalt}.
%
Building on our learner graph, we present the first fully
 \textit{heterogeneous} consensus protocol.
It generalizes Paxos to be heterogeneous along all three
 dimensions.

Heterogeneity allows acceptors to tailor a consensus protocol for
 the specific requirements of learners, rather than trying to force
 every learner to agree whenever any pair demand to agree.
This increased flexibility can save time and resources, or even make
 consensus possible where it was not before, as we now show with an
 example.

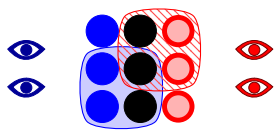
\begin{figure}
 \centering
\begin{tabular}{cc}
\begin{minipage}{.3\textwidth}
\begin{tikzpicture}[scale=.5]
\pgfdeclareimage[width=5mm]{eye-blue}{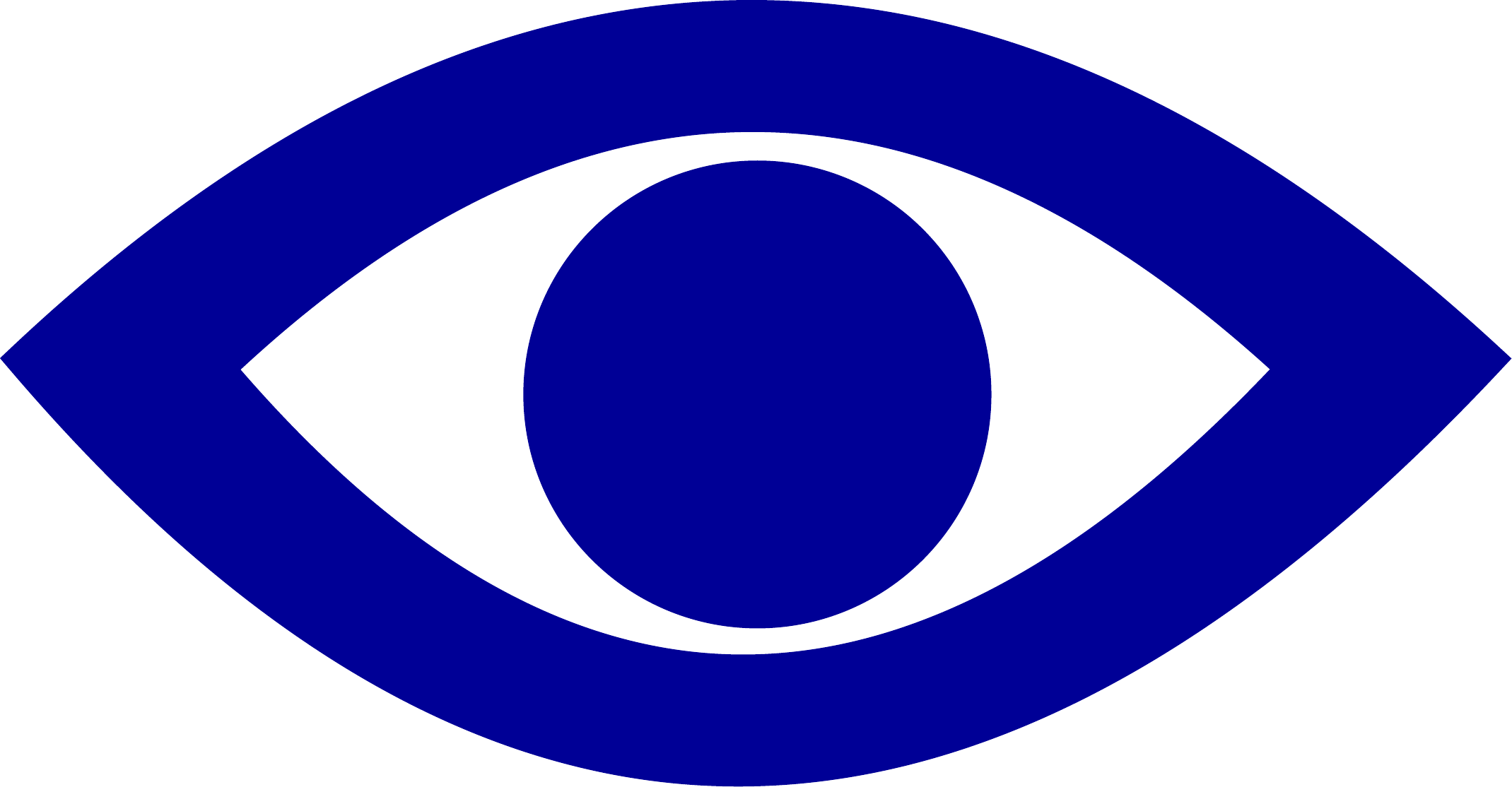}
\pgfdeclareimage[width=5mm]{eye-red}{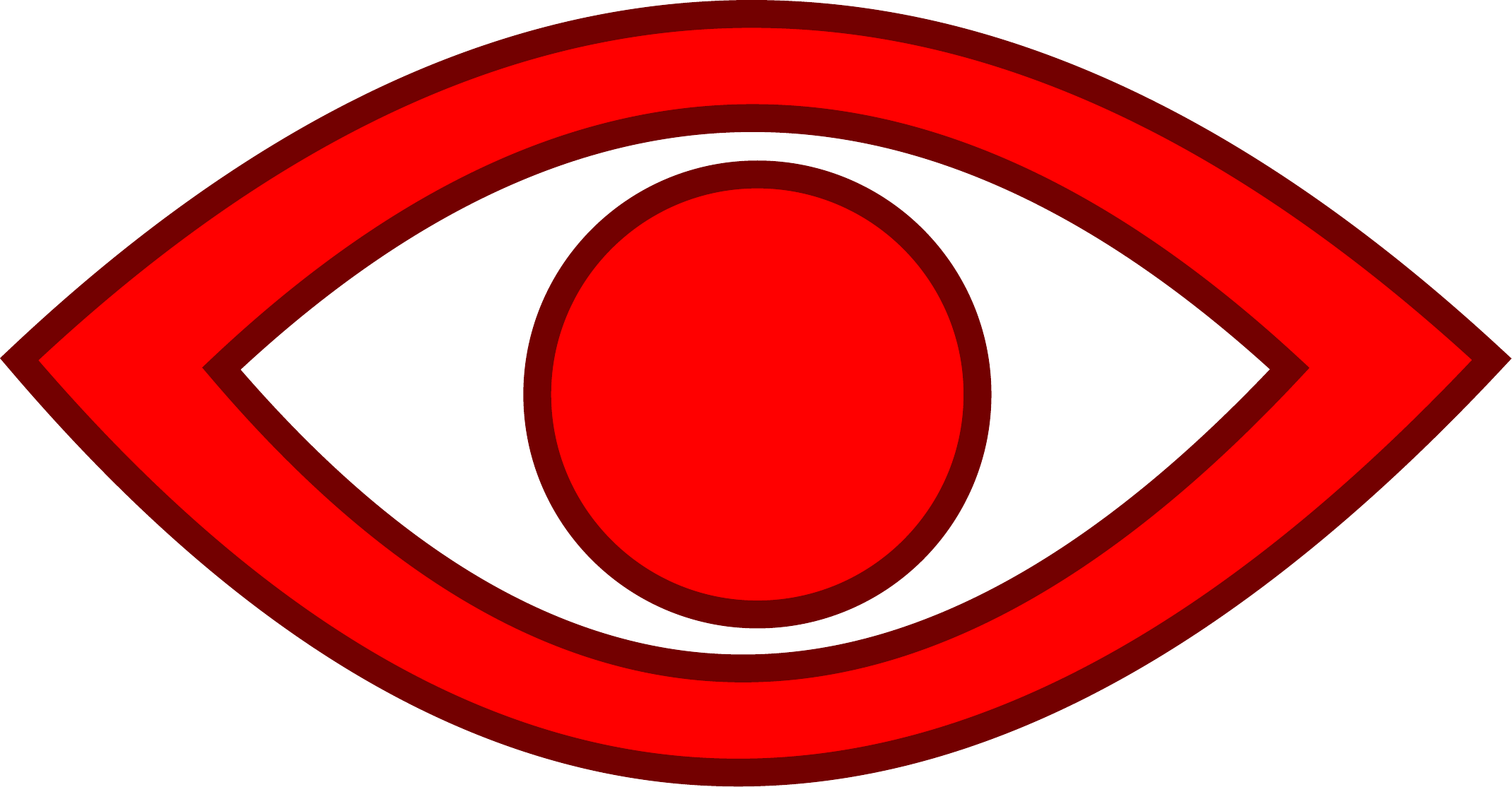}
\draw[blue,fill=blue!20] plot [smooth cycle] coordinates {
    (-.4, -.4)
    (-.4, 1.4)
    (1.4, 1.4)
    (1.4,  -.4)
  };
\draw[red, pattern=north west lines, pattern color=red!70] plot [smooth cycle] coordinates {
    (2.4, 2.4)
    ( .6, 2.4)
    ( .6,  .6)
    (2.4,  .6)
  };
\node[circle, fill=blue, inner sep=1.5mm, draw=blue, line width=0mm](A) at (0,0) {\ };
\node[circle, fill=blue, inner sep=1.5mm, draw=blue, line width=0mm](B) at (0,1) {\ };
\node[circle, fill=blue, inner sep=1.5mm, draw=blue, line width=0mm](C) at (0,2) {\ };
\node[circle, fill=black, inner sep=1.5mm, draw=black, line width=0mm](D) at (1,0) {\ };
\node[circle, fill=black, inner sep=1.5mm, draw=black, line width=0mm](E) at (1,1) {\ };
\node[circle, fill=black, inner sep=1.5mm, draw=black, line width=0mm](F) at (1,2) {\ };
\node[circle, fill=red!30, inner sep=1.28mm, draw=red, line width=.65mm](G) at (2,0) {\ };
\node[circle, fill=red!30, inner sep=1.28mm, draw=red, line width=.65mm](H) at (2,1) {\ };
\node[circle, fill=red!30, inner sep=1.28mm, draw=red, line width=.65mm](I) at (2,2) {\ };

\node[](blue0) at (-2,.5) {\pgfuseimage{eye-blue}};
\node[](blue1) at (-2,1.5) {\pgfuseimage{eye-blue}};

\node[](red0) at (4,.5) {\pgfuseimage{eye-red}};
\node[](red1) at (4,1.5) {\pgfuseimage{eye-red}};
\end{tikzpicture}
\end{minipage} & \begin{minipage}{.6\textwidth}
\caption[Heterogeneous Learners, Acceptors, and Failures]{
Illustration of the scenario in~\cref{sec:introexample}.
Blue learners are drawn as blue eyes, red learners as red,
 outlined eyes.
Blue acceptors are drawn as blue circles, red acceptors as red,
 outlined circles, and third parties as black circles. 
 The light solid blue region holds a quorum for the blue learners,
 and the striped red holds a quorum for the red learners.}
\label{fig:example}
\end{minipage}
\end{tabular}
\end{figure}

\subsection{Example}
\label{sec:introexample}
Suppose organizations \blue{Blue Org} and \red{Red Org} want to
 agree on a value, such as the order of transactions involving both of
 their databases or blockchains.
The people at \blue{Blue Org} are
 \textit{\blue{blue learners}}: they want to decide on a value
 subject to \emph{their} failure assumptions.
Likewise, the people at \red{Red Org} are
 \textit{\red{red learners}} with their own assumptions.
While neither organization's learners believe their own
 organization's acceptors (machines) are \byzantine, they do not
 trust the other organization's acceptors at all.
To help achieve consensus, they enlist three trustworthy third-party
 acceptors.
\Cref{fig:example} illustrates this situation.

All learners want to agree so long as there are no \byzantine failures.
However, no learner is willing to lose liveness
 (never decide on a value) if only one of its own acceptors has
 crashed, one third-party acceptor is \byzantine, and 
 all the other organization's learners are \byzantine.
Furthermore, learners within the same organization expect
 \textit{never} to disagree, so long as none of their own
 organization's acceptors are \byzantine.

Unfortunately, existing protocols cannot satisfy these learners.
Stellar~\cite{mazieresstellar}, for instance, has one of the most
 expressive heterogeneous models available, but it cannot express
 heterogeneous failures.
It cannot express \blue{blue} and \red{red} learners'
 desire to terminate if a third-party acceptor crashes, but not
 necessarily agree a third-party acceptor is \byzantine.
Our work enables a heterogeneous consensus protocol that
 satisfies all learners.
\ifreport
In particular, \blue{blue learners} will have quorums of any 2
 \blue{blue} and any 2 third-party acceptors, and
 \red{red learners} will have quorums of any 2 \red{red} and any
 2 third-party acceptors.
\fi

\subsection{Heterogeneous Paxos}
\label{sec:introhetcons}
Heterogeneous Paxos, our novel generalization of Byzantine
 Paxos achieves consensus in a fully
 heterogeneous setting~(\cref{sec:hetcons}), with
 precisely defined conditions under which learners in
 guaranteed safety and
 liveness\ifreport~(\cref{sec:correctness})\else\relax\fi.
Heterogeneous Paxos inherits Paxos' optimal 3-message-send best-case
 latency, making it especially good for latency-sensitive applications
 with geodistributed acceptors, including blockchains.
We have implemented this protocol and used it to construct several
 permissioned blockchains~\cite{GarcaPrez2018FederatedBQ}.
We demonstrate the savings in latency and resources that arise from
 tailoring consensus to specific learners' constraints\ifreport
 in a variety of scenarios~(\cref{sec:evaluation}).
By including the trust requirements of learners affiliated with
 multiple chains, we can even append single blocks to multiple chains
 simultaneously.
This is the equivalent of multi-shard transactions, where each shard
 has its own trust configuration.
With this implementation, we can build composable blockchains with
 heterogeneous trust, allowing application-specific chains to order
 only the blocks they need~(\cref{sec:charlotte})\fi.

\subsection{Contributions}
\label{sec:contributions}
\begin{itemize}
  \item The \textbf{Learner Graph} offers a general way to
         express heterogeneous trust assumptions in all three
         dimensions~(\cref{sec:learnergraph}).
  \item We \textbf{formally generalize the traditional consensus
         properties} (Validity, Agreement, and Termination) for the
         fully heterogeneous setting
         (\cref{sec:heterogeneousconsensus}).
  \item \textbf{Heterogeneous Paxos} is the first consensus protocol
         with heterogeneous learners, heterogeneous acceptors,
         and heterogeneous failures~(\cref{sec:hetcons}).
        It also inherits Paxos' optimal 3-message-send best-case
         latency\ifreport, and supports repeated
         consensus~(\cref{sec:repeatedconsensus})\fi.
\item \textbf{Experimental results} from our implementation of
        Heterogeneous Paxos demonstrate its use to construct
        permissioned blockchains with previously unobtainable
        security and performance properties (\cref{sec:charlotte}).
\end{itemize}

\section{System Model}
\label{sec:systemmodel}
We consider a \textit{closed-world} (or \textit{permissioned})
system consisting of a fixed set of \textit{acceptors}, a fixed set of
 \textit{proposers}, and a fixed set of \textit{learners}.
Proposers and acceptors can send messages to other acceptors and
 learners.
Some predetermined, but unknown set of acceptors are
 \textit{faulty} (we assume a non-adaptive adversary).
Faults include crash failures, which are not \textit{live} (they can
 stop at any time without detection), and
 \byzantine failures, which are neither \textit{live} nor
 \textit{safe} (they can behave arbitrarily).
\begin{definition}[Live]
\label{defn:live}
A live acceptor eventually sends every message required by the
 protocol.
\end{definition}
\begin{definition}[Safe]
  \label{defn:safe}
A safe acceptor will not send messages unless they are required by the
 protocol, and will send messages only in the order specified by the
 protocol.
\end{definition}

Learners set the conditions under which they expect to agree.
They want to decide values, and to be guaranteed agreement under
 certain conditions.
While learners can make bad assumptions, since they do not send
 messages, they cannot misbehave, and so there are no
 ``faulty learners.''

\ifreport\subsection{Network}
\else
\textbf{Network:}
\fi
\label{sec:network}
Network communication is point-to-point and reliable: if a live
 acceptor sends a message to another live acceptor, or to
 a learner, the message arrives.
We adopt a slight weakening of
 \textit{partial synchrony}~\cite{Dwork1988}: after some unknown
 global stabilization time (GST), all messages between live
 acceptors arrive within some unknown latency bound $\Delta$. 
In Heterogeneous Paxos, live acceptors send all messages to all
 acceptors and learners, but \byzantine acceptors may
 equivocate, sending messages to different recipients in different
 orders, with unbounded delays. 
We assume that messages carry effectively unbreakable cryptographic
 signatures, and that acceptors are identified by public keys. 
We also assume messages can \textbf{reference} other messages
 \textit{by collision-resistant hash}: if one message contains a
 hash of another, it uniquely identifies the message it is
 referencing~\cite{collision-resistance}. 

\ifreport\subsection{Consensus}
\else
\textbf{Consensus:}
\fi
\label{sec:consensus}
The purpose of consensus is for each learner to
 decide on exactly one value, and for all learners to decide on the
 same value.
%
Here, \textit{execution} refers to a specific instance of consensus:
 the actions of a specific set of acceptors during some 
 time frame.
A \textit{protocol} refers to the instructions that safe acceptors
 follow during an execution.

An execution of consensus begins when \textit{proposers}
 \textit{propose} candidate values, in the form of a message received
 by a correct acceptor.
(No consensus can make guarantees about proposed values only known to
 crashed or \byzantine acceptors.)
Proposers might be clients sending requests into the system.
We make no assumptions about proposer correctness for safety
 properties, but to guarantee liveness, we will assume that acceptors
 can act as proposers as well
 (i.e. proposers are a superset of acceptors).
\ifreport
We assume, however, that values are \textit{verifiable}: 
 for instance, a \textit{propose} message might need to bear relevant
 signatures.
\fi
After receiving some messages from acceptors, each learner
 eventually \textit{decides} on a single value.

Traditionally, consensus requires three properties~\cite{Fischer82b}:
\begin{itemize}
  \item \textit{Validity}: if a learner decides $p$, then $p$ was
         proposed.\footnote
           { Correia, Neves, and Ver\'issimo list several popular
              \textit{validity} conditions.  Ours corresponds to
              MCV2~\cite{Correia2006}}
  \item \textit{Agreement}: if learner $\red a$ decides value
         $\red v$, and learner $\blue b$ decides value
         $\blue{v^\prime}$, then $\red v = \blue{v^\prime}$.
  \item \textit{Termination}: all learners eventually decide.
\end{itemize}
In \cref{sec:heterogeneousconsensus}, we generalize these properties to account
 for heterogeneity.

\section{The Learner Graph}
\label{sec:learnergraph}

We characterize learners' failure assumptions with a
 novel construct called a \textit{learner graph}.
The learner graph is a general way to characterize trust
 assumptions for heterogeneous consensus.
It can encompass most existing formulations, including Stellar's
 ``slices''~\cite{mazieresstellar} and Cobalt's
 ``essential sets''~\cite{cobalt}.
We discuss other formulations
 in~\cref{sec:related}.
\begin{definition}[Learner Graph]
  \label{defn:learnergraph}
A learner graph is an undirected graph in which vertices are learners,
 each labeled with the conditions under which they must
 terminate~(\cref{sec:termination} formally defines
 \textit{termination}).
Each pair of learners is connected by an edge, labeled with the
 conditions under which those learners must
 agree~(\cref{sec:agreement} formally defines \textit{agreement}).
\end{definition}

\subsection{Quorums}
\label{sec:quorums}
A \textit{quorum} is a set of acceptors sufficient to make a learner decide: even if everything else has crashed~\cite{LAMP}, if a quorum are behaving correctly, a learner will eventually decide.
In a learner graph, each learner $\red a$ is labeled with a set of quorums $\red{Q_a}$.
The learner requires termination precisely when at least one quorum are all live. 
\ifreport
\begin{definition}[Quorum]
\label{defn:quorum}
A quorum $\red{q_a}$ is a set of acceptors.
\end{definition}
\fi

Within a specific execution, we assume some (unknown) set of
pre-determined acceptors are actually \textit{live}. We call this set
\reallylive.

\subsection{Safe Sets}
\label{sec:safesets}
To characterize the conditions under which two learners want to agree, we need to express all possible failures they anticipate. 
Surprisingly, crash failures cannot cause disagreement: any disagreement that occurs when some acceptor has crashed could also occur if the same acceptor were correct, but very slow, and did not act until after the learner decided.
Therefore, for agreement purposes, each tolerable failure scenario is
characterized by a \textit{safe set} (usually written $\red s$), the set of acceptors who are
\textit{safe}, meaning they act only according to the protocol.
\ifreport
\begin{definition}[Safe Set]
\label{defn:safeset}
A safe set $\red s$ is a set of acceptors.
\end{definition}
\fi
Between any pair of learners $\red a$ and $\blue b$ in the learner
graph, we label the edge  between them with a set of safe sets
$\edge{\red a}{ \blue b}$: so long as one of the safe sets in
$\edge{\red a}{ \blue b}$ indeed comprises only safe acceptors, the learners demand agreement.

Within a specific execution, we assume some (unknown) set of
pre-determined acceptors are actually \textit{safe}. We call this set
\reallysafe.
We do not require it, but systems often assume that
 $\reallysafe \subseteq \reallylive$, since a \byzantine
acceptor~\cite{Lamport82} may choose not to send messages.
\ifreport
Flexible BFT~\cite{Malkhi2019} explores ``alive-but-corrupt''
failures, which we would characterize as 
\fi

\subsubsection{Subset of Tolerable Failures}
We generally assume that a subset of tolerable failures is always tolerated:
\begin{assumption}\label{defn:subsetfailures}
Subset of failures properties:
\ifreport\[\else$\fi
\begin{array}{r c l c l}
  \forall \orange \ell
  &\hspace{-2mm}.&  \hspace{-2mm}
  \red{q_a} \in \red{Q_a}
  &\hspace{-2mm}\Rightarrow&\hspace{-2mm}
  \orange \ell \cup \red{q_a} \in \red{Q_a}
\\
  \forall \green x
  &\hspace{-2mm}.&\hspace{-2mm}
  s \in \edge{\red a}{\blue b}
  &\hspace{-2mm}\Rightarrow&\hspace{-2mm}
  \green x \cup s \in \edge{\red a}{\blue b}
\end{array}
\ifreport\]\else$\fi
\end{assumption}
One might imagine, for example, two learners who demand agreement if
 two acceptors fail, but not if only one acceptor fails.
However, we have no guarantee on time: if two acceptors are indeed
 faulty, one might act normally for an indefinite time, so the system
 would act as though only one has failed, and we will have to
 guarantee agreement.

\subsubsection{Generalized Learner Graph Labels}
It is possible to generalize the labels of learners and learner graph
 edges, and characterize quorums
 (conditions under which a learner must terminate)
 and safe sets
 (conditions under which pairs of learners must agree)
 as more detailed formal models
 (e.g., modeling network synchrony failures).
All consensus failure models of which we
 are aware can be formalized using learner graphs with generalized labels.
Heterogeneous Paxos works with any model of
 labels, so long as each label can be mapped
 (not necessarily uniquely) to a set of quorums for each learner, and
 a set of safe sets for each edge.
For simplicity, in this work, we define labels as a set of quorums for
 each learner, and a set of safe sets for each edge.

\subsection{Example}
\label{sed:clgexample}
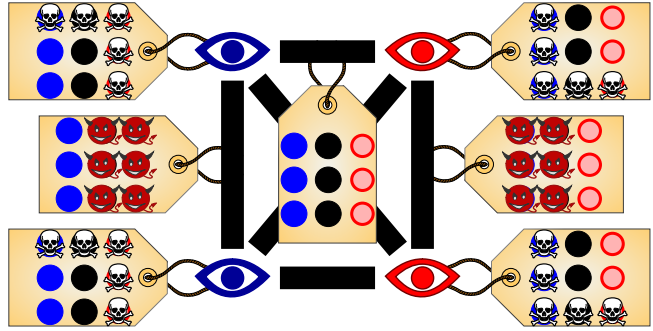
\begin{figure}
\centering
\vspace{-35mm}
\begin{tikzpicture}[scale=1]
\pgfdeclareimage[width=10mm]{eye-blue}{figures/icons/eye-icon-custom-larger-darkblue.pdf}
\pgfdeclareimage[width=10mm]{eye-red}{figures/icons/eye-icon-custom-outlinered.pdf}
\pgfdeclareimage[width=4mm]{skull}{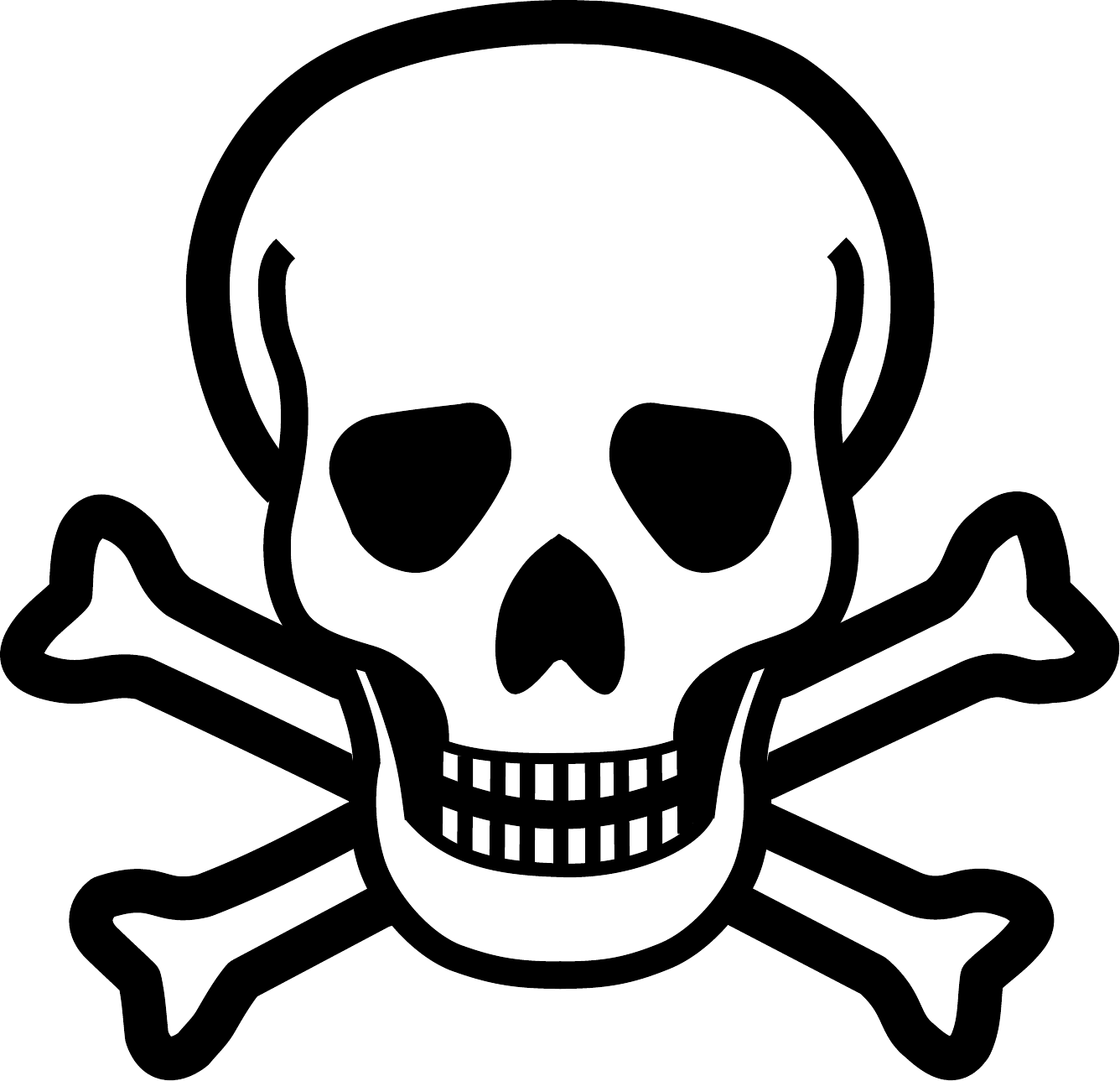}
\pgfdeclareimage[width=5mm]{devil}{figures/icons/devil.pdf}
\pgfdeclareimage[width=13mm]{label}{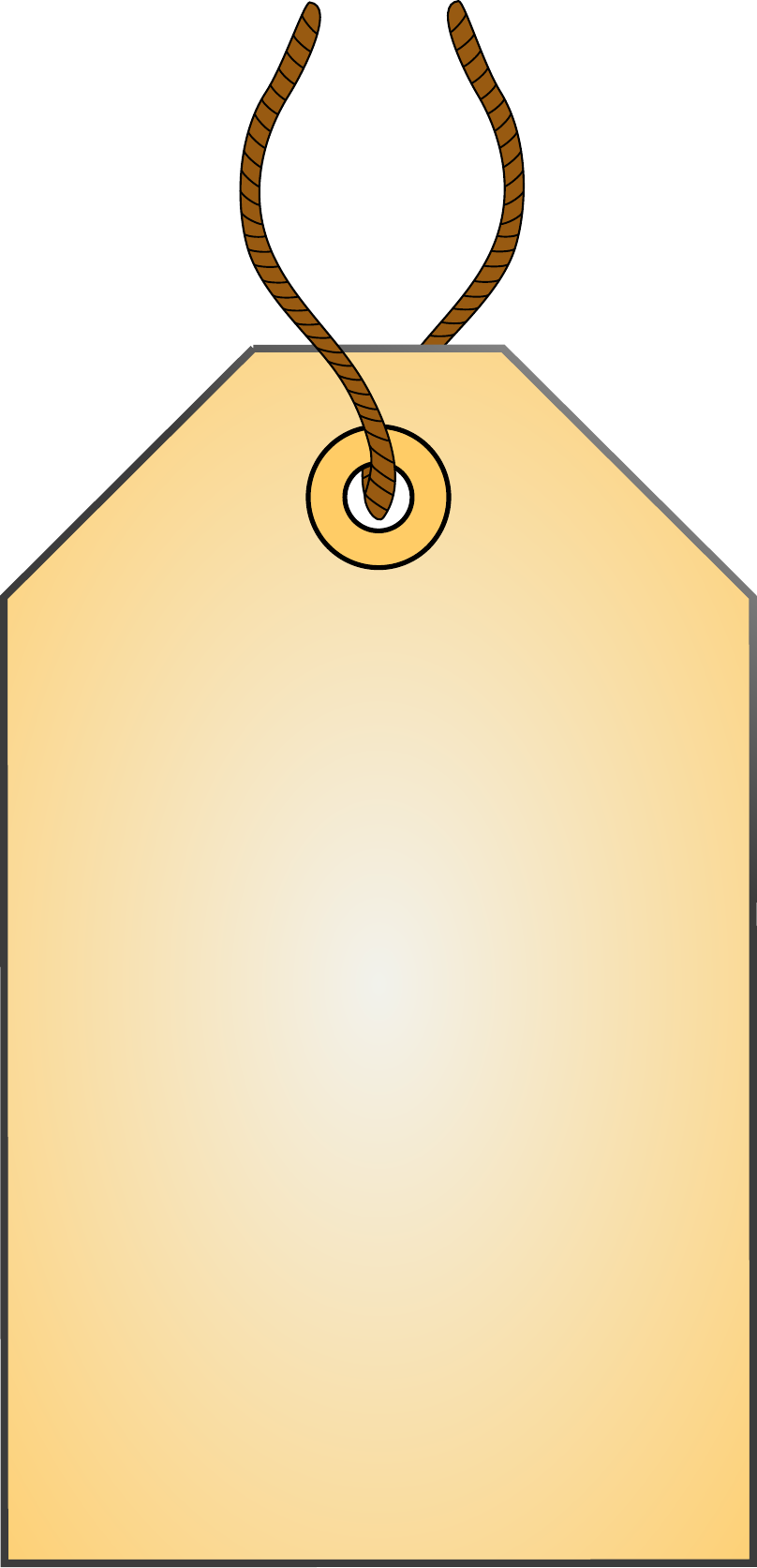}
\pgfdeclareimage[width=16mm]{label-sideways}{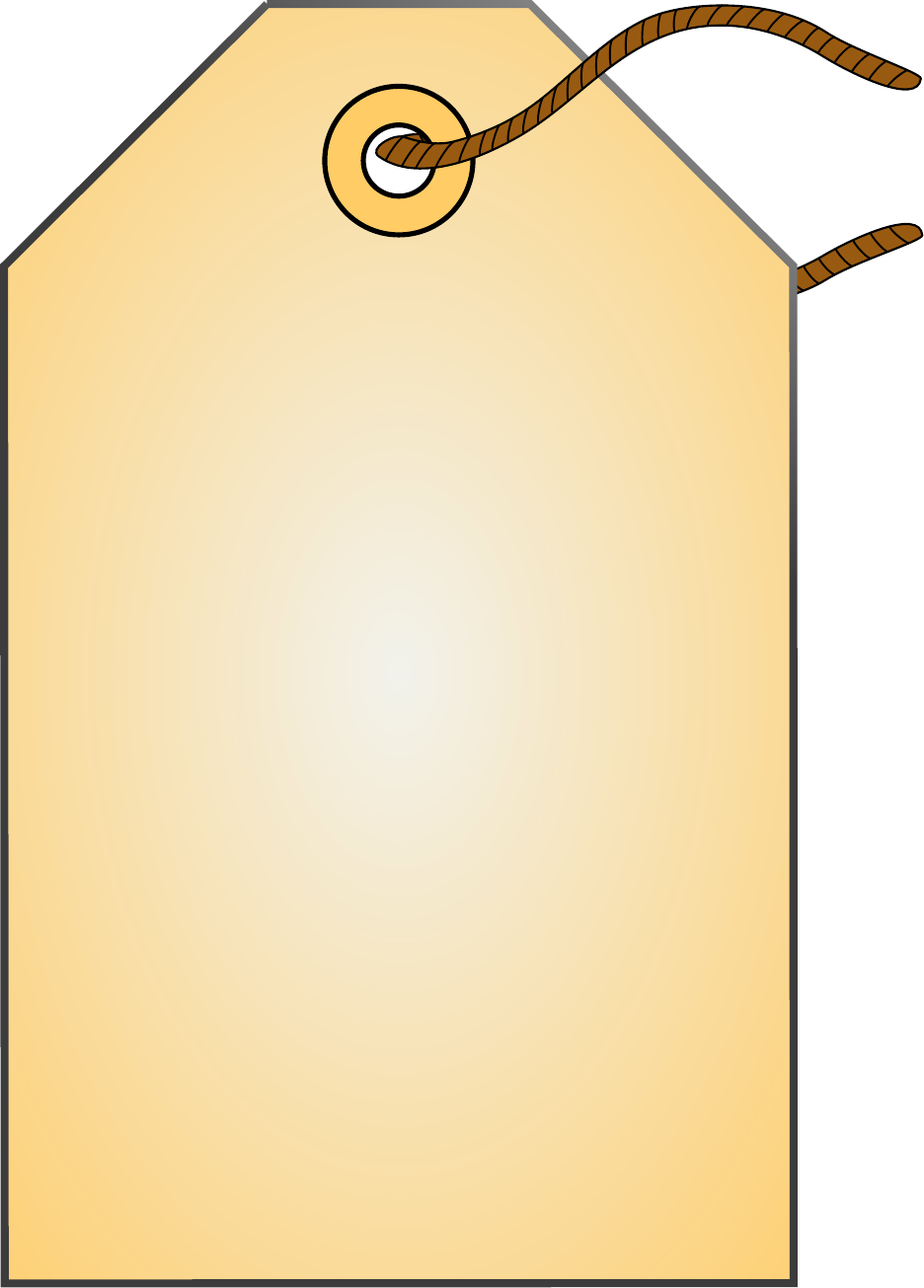}

\node[](label0) at (-2.3,5.0) {
  \begin{pgfrotateby}{\pgfdegree{270}}\pgfuseimage{label}\end{pgfrotateby}};
\introexamplenodes{-1.5}{2.55}
\introexamplefault{-1.5}{2.55}{1}{skull}
\introexamplefault{-1.5}{2.55}{2}{skull}
\introexamplefault{-1.5}{2.55}{3}{skull}
\introexamplefault{-1.5}{2.55}{6}{skull}
\introexamplefault{-1.5}{2.55}{9}{skull}

\node[](label0) at (-2.3,2.0) {
  \begin{pgfrotateby}{\pgfdegree{270}}\pgfuseimage{label}\end{pgfrotateby}};
\introexamplenodes{-1.5}{-0.45}
\introexamplefault{-1.5}{-0.45}{1}{skull}
\introexamplefault{-1.5}{-0.45}{2}{skull}
\introexamplefault{-1.5}{-0.45}{3}{skull}
\introexamplefault{-1.5}{-0.45}{6}{skull}
\introexamplefault{-1.5}{-0.45}{9}{skull}

\node[]() at (6.15,3.7) {
  \begin{pgfrotateby}{\pgfdegree{90}}\pgfuseimage{label}\end{pgfrotateby}};
\introexamplenodes{5}{2.55}
\introexamplefault{5}{2.55}{1}{skull}
\introexamplefault{5}{2.55}{4}{skull}
\introexamplefault{5}{2.55}{8}{skull}
\introexamplefault{5}{2.55}{7}{skull}
\introexamplefault{5}{2.55}{9}{skull}

\node[]() at (6.15,0.7) {
  \begin{pgfrotateby}{\pgfdegree{90}}\pgfuseimage{label}\end{pgfrotateby}};
\introexamplenodes{5}{-0.45}
\introexamplefault{5}{-0.45}{1}{skull}
\introexamplefault{5}{-0.45}{4}{skull}
\introexamplefault{5}{-0.45}{8}{skull}
\introexamplefault{5}{-0.45}{7}{skull}
\introexamplefault{5}{-0.45}{9}{skull}

\node[](blue0) at (0,0) {\pgfuseimage{eye-blue}};
\node[](blue1) at (0,3) {\pgfuseimage{eye-blue}};
\node[](label0) at (-1.9,3.5) {
  \begin{pgfrotateby}{\pgfdegree{270}}\pgfuseimage{label}\end{pgfrotateby}};
\draw[line width=3mm] (blue0) -- (blue1);

\introexamplenodes{-1.25}{1.05}

\node[](LDemonBlack0) at (-1.7,1.95) {\pgfuseimage{devil}};
\node[](LDemonBlack0) at (-1.7,1.5) {\pgfuseimage{devil}};
\node[](LDemonBlack0) at (-1.7,1.05) {\pgfuseimage{devil}};
\node[](LDemonRed0) at (-1.25,1.95) {\pgfuseimage{devil}};
\node[](LDemonRed1) at (-1.25,1.5) {\pgfuseimage{devil}};
\node[](LDemonRed2) at (-1.25,1.05) {\pgfuseimage{devil}};

\node[](red0) at (2.5,0) {\pgfuseimage{eye-red}};
\node[](red1) at (2.5,3) {\pgfuseimage{eye-red}};
\draw[line width=3mm] (blue0) -- (red1);
\draw[line width=3mm] (blue1) -- (red0);
\draw[line width=3mm] (blue0) -- (red0);

\node[xscale=-1](label1) at (1.25,1.8) {\pgfuseimage{label}};

\introexamplenodes{1.71}{0.85}


\draw[line width=3mm] (blue1) -- (red1);

\node[](label2) at (5.8,2.2) {
  \begin{pgfrotateby}{\pgfdegree{90}}\pgfuseimage{label}\end{pgfrotateby}};
\introexamplenodes{4.7}{1.05}

\node[](RDemonBlack2) at (4.25,1.05) {\pgfuseimage{devil}};
\node[](RDemonBlack2) at (4.25,1.5) {\pgfuseimage{devil}};
\node[](RDemonBlack2) at (4.25,1.95) {\pgfuseimage{devil}};
\node[](RDemonBlue0) at (3.8,1.95) {\pgfuseimage{devil}};
\node[](RDemonBlue1) at (3.8,1.5) {\pgfuseimage{devil}};
\node[](RDemonBlue2) at (3.8,1.05) {\pgfuseimage{devil}};
\draw[line width=3mm] (red0) -- (red1);
\end{tikzpicture}
\caption{
 Learner Graph from~\cref{sed:clgexample}:
Learners are eyes, with darker \blue{blue learners} on the
  left, and outlined \red{red learners} on the right.
Edge labels display \textit{one} safe set for which the learners want to agree (unsafe acceptors are marked with a devil).
The center label represents all edges between \red{red} and \blue{blue} learners.
Learner labels display \textit{one} quorum for which the learner wants to terminate (crashed acceptors are marked with a skull).
In each label, \blue{blue acceptors} are blue circles, \red{red acceptors} are red, outlined circles, and third-party acceptors are black circles.
}
\label{fig:examplecog}
\end{figure}

Consider our example from~\cref{sec:introexample} and~\cref{fig:example}.
All learners want to agree when all acceptors are safe.
However, each learner demands termination
 (it must eventually decide on a value) even when one of its own
  acceptors has crashed, and one third part as well as all the other
  organization's acceptors have failed as well.
Furthermore, learners  within the same organization expect \textit{never}
 to disagree, so long as none of their own organization's
 acceptors are \byzantine: neither organization tolerates the
 other, or third-party acceptors, creating internal disagreement.
In~\cref{fig:examplecog}, we diagram the learner graph.
For space reasons, we draw each label with only \textit{one} quorum or
 \textit{one} safe set.

\subsection{Agreement is Transitive and Symmetric}
\label{sec:undirected}
\label{sec:transitivity}
Agreement~(formally defined in~\cref{sec:agreement}) is symmetric, so learner graphs are
 undirected ($\edge{\red a}{\blue b} = \edge{\blue b}{\red a}$).
Agreement is also transitive:
 if $\red a$ agrees with $\blue b$ and $\blue b$ agrees with
 $\green c$, then $\red a$ agrees with $\green c$.
As a result, $\red a$ and $\green c$ must agree whenever both the
 conditions $\edge{\red a}{\blue b}$ and $\edge{\blue b}{\green c}$
 are met.
When learners' requirements reflect this assumption, we call the
 resulting learner graph \textit{condensed}.
\ifreport\else
We describe how to condense a learner graph
 in \cref{sec:condensed}.
\fi
\begin{definition}[Condensed Learner Graph $\p\clg$]
  \label{defn:condensed}
A learner graph $G$ is \textit{condensed} iff:
$
 \forall \red a, \blue b, \green c .\ 
 \p{\edge{\red a}{\blue b} \cap \edge{\blue b}{\green c}}
 \subseteq
 \edge{\red a}{\green c}
$
\end{definition}

\ifreport
\subsection{Condensed Learner Graph \p\clg}
\label{sec:condensed}
A \textit{Condensed Learner Graph $\p\clg$} represents all the
 conditions under which each pair of learners must actually
 have agreement, after transitivity is taken into account.
As a result of the transitive property of agreement, for any learner
 graph  $G$ featuring principals $\red a$, $\blue b$, and $\green c$,
 we can create a new learner graph $G^\prime$, in which
 $\edge{\red a}{\green c}$ is replaced by:
\[
\edge{\red a}{\green c} \cup
\p{\edge{\red a}{\blue b}
   \cap
   \edge{\blue b}{\green c}}
\]

Starting with learner graph $G$, applying this operation
 repeatedly results in a Condensed Learner Graph \p\clg:
 a learner graph in which this operation doesn't change anything
 when applied to any pair of adjacent edges.\footnote{
  With the Floyd-Warshall algorithm~\cite{Floyd1962}, this can be done
   in $O\p{\abs{G}^3}$ time.}

Given \cref{defn:subsetfailures}, there is an equivalent
 expression in terms of unions of safe sets:

\begin{lemma}\label{lemma:equivcondensed}
\[
  \p{
     \p{\edge{\red a}{\blue b} \cap \edge{\blue b}{\green c}}
     \subseteq
     \edge{\red a}{\green c} }
  \Leftrightarrow
  \p{
      \purple s \in \edge{\red a}{\blue b}
    \land
      \orange{s^\prime} \in \edge{\blue b}{\green c}
    \Rightarrow
      \purple s \cup \orange{s^\prime} \in \edge{\red a}{\green c}
  }
\]
\end{lemma}
\begin{proof}
First, we show:
\[
  \p{
     \p{\edge{\red a}{\blue b} \cap \edge{\blue b}{\green c}}
     \subseteq
     \edge{\red a}{\green c} }
  \Rightarrow
  \p{
      \purple s \in \edge{\red a}{\blue b}
    \land
      \orange{s^\prime} \in \edge{\blue b}{\green c}
    \Rightarrow
      \purple s \cup \orange{s^\prime} \in \edge{\red a}{\green c}
  }
\]

By \cref{defn:subsetfailures}:
\[
  \andlinesTwo{
    \purple s \in \edge{\red a}{\blue b} \Rightarrow
    \purple s \cup \orange{s^\prime}
    \in
    \edge{\red a}{\blue b}
  }{
    \orange{s^\prime} \in \edge{\blue b}{\green c} \Rightarrow
    \purple s \cup \orange{s^\prime}
    \in
    \edge{\blue b}{\green c}
  }
\]

Therefore: $
    \purple s \in \edge{\red a}{\blue b}
  \land
    \orange{s^\prime} \in \edge{\blue b}{\green c} \Rightarrow
    \purple s \cup \orange{s^\prime}
    \in
    \edge{\red a}{\blue b} \cap \edge{\blue b}{\green c}
  $.
Given $
     \p{\edge{\red a}{\blue b} \cap \edge{\blue b}{\green c}}
     \subseteq
     \edge{\red a}{\green c}
     $:
\[
      \purple s \in \edge{\red a}{\blue b}
    \land
      \orange{s^\prime} \in \edge{\blue b}{\green c}
    \Rightarrow
      \purple s \cup \orange{s^\prime} \in \edge{\red a}{\green c}
\]

\noindent Second, we show
\[
  \p{
      \purple s \in \edge{\red a}{\blue b}
    \land
      \orange{s^\prime} \in \edge{\blue b}{\green c}
    \Rightarrow
      \purple s \cup \orange{s^\prime} \in \edge{\red a}{\green c}
  }
  \Rightarrow
  \p{
     \p{\edge{\red a}{\blue b} \cap \edge{\blue b}{\green c}}
     \subseteq
     \edge{\red a}{\green c} }
\]
For any $\purple s
         \in
         \edge{\red a}{\blue b}\cap\edge{\blue b}{\green c}$,
it certainly holds that
\[
\andlinesTwo
{\purple s \in \edge{\red a}{\blue b}}
{\purple s \in \edge{\blue b}{\green c}}
\]
Given $
  \purple s \in \edge{\red a}{\blue b}
\land
  \orange{s^\prime}
  \in
  \edge{\blue b}{\green c}
\Rightarrow
  \purple s \cup \orange{s^\prime}
  \in
  \edge{\red a}{\green c}$, 
it follows that
$\edge{\red a}{\blue b}\cap\edge{\blue b}{\green c}
 \supseteq
 \edge{\red a}{\green c}$.
\end{proof}

\else
\fi

\ifreport
\subsection{Self-Edges}
\else
\noindent \textbf{Self-Edges:}
\fi
A \clg describes when a learner $\red a$
 agrees with itself (i.e., if it decides twice, both decisions must
 have the same value):
 $\edge{\red a}{\red a}$.
\begin{lemma}[Self-agreement]
  \label{lemma:selfedge}
   A learner must agree with itself in order to agree with anyone:
  $
    \edge{\red a}{\blue b} \subseteq \edge{\red a}{\red a}
  $
\end{lemma}
\begin{proof}
Follows from \cref{defn:condensed},
 and the fact that the \clg is undirected~(\cref{sec:undirected})
\end{proof}

\subsection{Liveness Bounds from Safety}
Given the conditions under which learners want to agree, we can derive a (sufficient) bound on the quorums they require to terminate. 
In other words, given labels for the edges in the learners graph, we can bound the labels for the vertices.

As we will cover in more detail in~\cref{sec:valid}, each of a learner's quorums must intersect its neighbors quorums at a safe acceptor.
As a result, we can construct a sufficient set of quorums for each learner in a \clg as follows: for each edge of the learner, each quorum includes a majority of acceptors from a each of the safety sets.

\subsection{Safety Bounds from Liveness}
Given the conditions under which learners want to terminate, we can derive a (necessary) bound on the safe sets they can require on each of their edges.
As we will cover in more detail in~\cref{sec:valid}, each of a learner's quorums must intersect its neighbors quorums at a safe acceptor.
As a result, safe sets can be assembled for each edge in a \clg as follows: each set includes one acceptor from the intersection of each pair of quorums (one from each learner).

\section{Heterogeneous Consensus}
\label{sec:heterogeneousconsensus}
We now define our novel heterogeneous generalization of traditional
 consensus properties.

\subsection{Validity}
\label{sec:validity}
Intuitively, a consensus protocol shouldn't allow learners to 
 always decide some predetermined value.
Validity is the same in heterogeneous and homogeneous settings.

\begin{definition}[Heterogeneous Validity]
  \label{defn:validity}
  \ 
  \begin{itemize}
    \item A consensus execution is \textit{valid} if all
           values learners decide were proposed in that
           execution.
    \item A consensus protocol is \textit{valid} if all possible
           executions are valid.
  \end{itemize}
\end{definition}

\subsection{Agreement}
\label{sec:agreement}
Our generalization of Agreement from the homogeneous setting to a
 heterogeneous one is the key insight that makes our conception of
 heterogeneous consensus possible.
It generalizes not only the traditional homogeneous approach, but also the
 ``intact nodes'' concept from Stellar~\cite{mazieresstellar}, and
 ``linked nodes'' from Cobalt~\cite{cobalt}.

\begin{definition}[Entangled]
  \label{defn:entangled}
  In an execution, two learners are \textit{entangled} if their
   failure assumptions matched the 
   failures that actually happen:
  $
    \entangled{\red a}{\blue b} \triangleq
    \reallysafe \in \edge{\red a}{\blue b}
  $
\end{definition}
\noindent
In the example~(\cref{sec:introexample}), if one third-party acceptor
 were \byzantine, the \blue{blue learners} would be entangled
 with each other, and similarly with the \red{red learners},
 but no \blue{blue learners} would be entangled with \red{red learners}.
It is possible for failures to divide the learners into separate
 groups, which may then decide different values
 even if they agree among themselves.  

\begin{definition}[Heterogeneous Agreement]
  \label{defn:agreement}
  \ 
  \begin{itemize}
    \item Within an execution, two learners have \textit{agreement}
           if all decisions for
           either learner have the same value.
    \item A heterogeneous consensus protocol has \textit{agreement}
           if, for all possible executions of that protocol, all
           entangled pairs of learners have agreement.
  \end{itemize}
\end{definition}

In Heterogeneous Paxos, as in many other protocols, learners decide on
 a value whenever certain conditions are met for that value: learners
 can even decide multiple times.
If there aren't too many failures, a learner is guaranteed to decide
 the same value every time.
 Because learners send no messages, they cannot  \textit{fail}, but they can
 make incorrect assumptions.
Within the context of an execution, entanglement neatly defines when a
 learner is \textit{accurate}, meaning it cannot decide different
 values.
\begin{definition}[Accurate Learner]
  \label{defn:accurate}
  \ifreport An accurate learner \fi is entangled with itself:
  \ifreport\[\else$\fi
    \accurate{\red a} \triangleq \entangled{\red a}{\red a}
  \ifreport\]\else$\fi
\end{definition}
In the example~(\cref{sec:introexample}), if one third-party acceptor
 were \byzantine, then the \blue{blue} and \red{red} learners would
 be accurate, but if a \blue{blue acceptor} were also \byzantine,
 the \blue{blue learners} would not be accurate (although the
 \red{red learners} would still be accurate).

\subsection{Termination}
\label{sec:termination}
Termination has no well agreed-upon definition for the heterogeneous
 setting, as it does not generalize easily from the homogeneous one.
A heterogeneous consensus protocol is specified in terms of the
 (possibly differing) conditions under which each learner is
 guaranteed termination~(\cref{sec:learnergraph}).
\ifblinded
For example,
Sheff et al.\@~\cite{hetconstechreport} distinguish between
 ``gurus,'' learners with accurate failure assumptions, and
 ``chumps,'' who hold inaccurate assumptions;
\else
For example, in our prior work on Heterogeneous Fast Consensus, we
 distinguish between ``gurus,'' learners with accurate failure assumptions,
 and ``chumps,'' who hold inaccurate
 assumptions~\cite{hetconstechreport};
\fi
Stellar calls them ``intact'' and ``befouled''~\cite{mazieresstellar}.
When discussing termination properties, we use the following
 terminology:

\begin{definition}[Termination]
  \label{defn:termination}
  \ 
  \begin{itemize}
    \item Within an execution, a learner has termination if it
           eventually decides.
    \item A heterogeneous consensus protocol has \textit{termination}
           if, for all possible executions of that protocol, all
           learners with a safe and live quorum have termination.
  \end{itemize}
\end{definition}

\noindent
Protocols can only guarantee termination under specific network
 assumptions, and varying notions of
 ``eventually''~\cite{Fischer82b,paxos-made-simple,Miller2016}.
Following in the footsteps of
 Dwork et al.~\cite{Dwork1988}, Heterogeneous Paxos
 guarantees Validity and Agreement in a fully asynchronous
 network, and termination in a partially synchronous
 network (\cref{defn:networkassumption}).
Furthermore, as in all other consensus protocols, if there are too
many acceptor failures, some learners may not terminate.
Specifically, a learner will decide (terminate) if at least one of its quorums is live.
\begin{definition}[Terminating Learner]
  \label{defn:terminating}
  \ifreport A terminating learner \fi has a live, safe quorum:
  \ifreport\[\else$\fi
  \terminating{\red a} \triangleq \reallylive\cup\reallysafe \in \red{Q_a}
  \ifreport\]\else$\fi
\end{definition}

\begin{figure}
\newsavebox\kmsub
\savebox\kmsub{\scriptsize\ttfamily known\char95messages}
\newcommand\Decision{\mathit{Decision}}
\newcommand\WellFormed{\mathit{WellFormed}}

\begin{minipage}[t]{0.63\columnwidth}
\begin{lstlisting}[style=pseudocode]
@\textit{\large acceptor\_initial\_state}@:
 known_messages = @$\cb{}$@
 recently_received = @$\cb{}$@

@\textit{\large acceptor\_on\_receipt}@(@$\red m$@):
 for @$\green r \in \red{m.refs}$@:
  while @$\green r ∉\,$@known_messages:
   wait()
 atomic:
  if @$\red m∉\,$@known_messages:
   forward @$\red m$@ to all acceptors and learners
   recently_received @$\cup$@= @$\cb{\red m}$@
   known_messages @$\cup$@= @$\cb{\red m}$@
   if @$\red m$@ has type 1a:
     @$\green z$@ = new 1b(refs = recently_received)
     recently_received = @$\cb{}$@
     on_receipt(@$\green z$@)
   if @$\red m$@ has type 1b and @$\ba{\red m}$@ == max@$_{x∈\usebox\kmsub}\,\ba{\green x}$@
     @\hetdifftext{\blue{for} learner $\in$ learners:}@
       @$\green z$@ = new 2a(refs = recently_received@\hetdifftext{, lrn = learner}@)
       if @$\WellFormed(\green z)$@:
         recently_received = @$\cb{}$@
         on_receipt(@$\green z$@)
\end{lstlisting}
\end{minipage}
\begin{minipage}[t]{0.37\columnwidth}
\begin{lstlisting}[style=pseudocode]
@\textit{\large learner\_initial\_state}@:
  known_messages = @$\cb{}$@

@\textit{\large learner\_on\_receipt}@(@$\red m$@):
 for @$\green r \in \red{m.refs}$@:
  while @$\green r$ ∉\,@known_messages:
   wait()
 known_messages @$\cup$@= @$\cb{\red m}$@
 for S @$\subseteq$@ known_messages:
  if @${\decision{\hetdifftext{\blue{\texttt{self}}}}{{\texttt{S}} \cup \cb{\red m}}}$@:
    decide(@$\va{\red m}$@)
\end{lstlisting}
\end{minipage}
\vspace{-5mm}
\caption{Pseudo-code for Acceptor (left) and Learner (right).
\cref{sec:hetcons} defines 
message structure~(\cref{sec:messaging}),
$WellFormed$~(\cref{defn:wellformedness}),
$\ba{}$~(\cref{defn:b}), 
$\va{}$~(\cref{defn:va}),
and $\decision{}{}$~(\cref{defn:decision}).  
}
\label{fig:pseudocode}
\end{figure}

\section{Heterogeneous Paxos} 
\label{sec:hetcons}
Heterogeneous Paxos is a consensus protocol~(\cref{sec:consensus})
 based on Byzantine Paxos, Lamport's
 \byzantine-fault-tolerant~\cite{Lamport82} variant of
 Paxos~\cite{paxos, paxos-made-simple} using a simulated
 leader~\cite{byzantizing-paxos}.
This protocol is conceptually simpler than 
 \textit{Practical Byzantine Fault Tolerance}~\cite{osdi99}.
When all learners have the same failure assumptions, Heterogeneous Paxos
 is \textit{exactly} Byzantine Paxos.

Byzantine Paxos was originally written as a sequence of changes from
 crash-tolerant Paxos~\cite{byzantizing-paxos,paxos}.
We were able to construct a complete version of Byzantine Paxos in
 such a way that we could describe Heterogeneous Paxos with only a few
 additions, \hetdifftext{highlighted in pale blue}.
To our knowledge, without the portions
 \hetdifftext{highlighted in pale blue} this is also the most direct
 description of the Byzantine Paxos via Simulated Leader protocol in
 the literature. 
\Cref{fig:pseudocode} presents pseudocode for Heterogeneous Paxos
 acceptors and learners.

\ifreport
\subsection{Overview}
\fi
\label{sec:overview}
Informally, Heterogeneous Paxos proceeds as a series of
 (possibly overlapping) \textit{phases} corresponding to three types
 of messages, traditionally called
 \textit{1a}, \textit{1b}, and \textit{2a}:
\begin{itemize}
  \item Proposers send \textit{1a} messages, each carrying a value and
         unique \emph{ballot number} (stage identifier), to
         acceptors. 
  \item Acceptors send \textit{1b} messages to each other to
         communicate that they've received a \textit{1a}
         (line 15 of~\cref{fig:pseudocode}).
  \item When an acceptor receives a \textit{1b} message for the
         highest ballot number it has seen from a
         learner \red a's \textit{quorum} of acceptors,  it sends a
         \textit{2a} message labeled with \red a and that ballot
         number (line 20 of~\cref{fig:pseudocode}).
        There is one exception ($WellFormed$ in~\cref{fig:pseudocode}): 
         once a safe acceptor sends a \textit{2a} message
         \green m for a learner \red a, it never sends a
         \textit{2a} message with a different value for a learner
         \blue b, unless:
        \begin{itemize}
          \item It knows that a quorum of acceptors has seen
                 \textit{2a} messages with learner \red a and 
                 ballot number higher than \green m.
          \item Or it has seen \byzantine behavior that proves
                \red a and \blue b do not have to agree.
        \end{itemize}
  \item A learner \red a \textit{decides} when it receives
         \textit{2a} messages with the same ballot number from one of its
         quorums of acceptors (line 11 on the right of~\cref{fig:pseudocode}).
\end{itemize}
Proposers can restart the protocol at any time, with a new ballot
 number.
Acceptor and Learner behavior in Heterogeneous Paxos is described
 in~\cref{fig:pseudocode}.
We now describe their sub-functions, including message
 construction~(\cref{sec:messaging}),
 $WellFormed$~(\cref{defn:wellformedness}),
 $\ba{}$~(\cref{defn:b}), 
 $\va{}$~(\cref{defn:va}),
 and $\decision{}{}$~(\cref{defn:decision}).  

\textbf{Key Insight:}
Intuitively, Heterogeneous Paxos operates much like Byzantine Paxos,
 except that all acceptors execute the final phase separately for each
 learner.
The shared phases allow learners to agree when possible, while the
 replicated final phase allows different learners to decide under
 different conditions.
\ifreport
\subsection{Assumptions and Definitions}
\label{sec:heterogeneousassumptions}
Additional assumptions and definitions are needed to precisely define
 Heterogeneous Paxos:
\begin{itemize}
\item \textbf{Acceptors} are servers which both send and receive
               messages as part of the
               protocol~\cite{byzantizing-paxos}.
\item \textbf{Proposers} are machines that propose
              potential values for consensus~\cite{byzantizing-paxos}.
\item \textbf{Learners} receive messages from the acceptors.
              Learners \textit{decide} on
               values~\cite{byzantizing-paxos}.
\item Acceptors can exchange \textbf{messages} over the
       network, and \textbf{digital signatures} make messages
       unforgeable.
\item Messages can \textbf{reference} other messages
       \textit{by collision-resistant hash}: if one message contains a
       hash of another, it uniquely identifies the message it is
       referencing~\cite{collision-resistance}. 
\item \textbf{Safe} acceptors will not send messages unless required
        by the protocol, and will not send messages in any order other
        than the order specified by the protocol.
\item \textbf{Live} acceptors eventually send any message required by
        the protocol.
\item Any message from one live acceptor to another live
       acceptor eventually arrives (\cref{sec:network}).
\item The learner graph (including its labels) as \textit{valid}.
\end{itemize}
\fi
\Cref{sec:examples} describes several heterogeneous consensus
 scenarios, as well as quorums for each learner.

\subsection{Valid Learner Graph}
\label{sec:valid}
Naturally, there are bounds on the learner graphs for which
 Heterogeneous Paxos can provide guarantees.
Unlike traditional consensus, in a Heterogeneous Consensus learner
 graph, each learner $\red a$ has its own set of quorums $\red{Q_a}$. 
These describe the learner's termination constraints: it may not
 terminate if all of its quorums contain a non-live
 acceptor~(\cref{defn:terminating}).
The notion of a \textit{valid} learner graph generalizes the
 homogeneous assumption that every pair of quorums have a safe
 acceptor in their intersection.

Homogeneous Byzantine Paxos guarantees
 agreement~(\cref{sec:agreement}) when all pairs of quorums have
 \ifreport at least one \else $\geq 1$ \fi safe acceptor in their
 intersection.
The heterogeneous case has a similar requirement:

\begin{definition}[Valid Learner Graph]
\label{defn:valid}
A learner graph is valid iff for each pair of learners
 $\red a$ and $\blue b$, whenever they must agree, all of their quorums feature at least one safe acceptor in their intersection:
$
\purple s \in \edge{\red a}{\blue b}\land
\red{q_a}\in\red{Q_a}\land
\blue{q_b}\in\blue{Q_b}
\ \Rightarrow\ 
\red{q_a}\cap\blue{q_b}\cap\purple s\ne \emptyset
$
\end{definition}

\subsection{Messaging}
\label{sec:messaging}
Acceptors send messages to each other. 
Live acceptors echo all messages sent and received to all other
 acceptors and learners, so if one live acceptor
 receives a message, all acceptors eventually receive it.
When safe acceptors receive a message, they process and send
 resulting messages specified by the protocol atomically:
 they do not receive messages between sending results to
 other acceptors.
Safe acceptors also receive any messages they send to themselves
 im\-med\-iate\-ly: they receive no other messages between sending and
 receiving.

Each message $\green x$ contains a cryptographic signature
 allowing anyone to identify the signer:
\begin{definition}[Message Signer]
\label{defn:sig}
$
  \sig{\green x\!:\!message}\!\triangleq\!
  \textrm{the acceptor or proposer that signed }\green x
$
\end{definition}
We can define $\sig{}$ over sets of messages, to mean the set of signers of those
 messages:
\begin{definition}[Message Set Signers]
  \label{defn:sigs}
$
  \sig{\green x : set}
  \triangleq \cb{\tallpipe{\sig{\blue m}}{\blue m \in \green x}}
$
\end{definition}

Furthermore, each message $\green x$ carries references to 0 or more
 other messages, \textit{\green{x.refs}}.
These references are by hash, ensuring both the absence of cycles in the reference
 graph and that it is possible to know exactly when one message
 references another~\cite{collision-resistance}.
In each message, safe acceptors reference each
 message they received since the last message they sent.
Since all messages sent are sent to all acceptors, and safe
 acceptors receive messages sent to themselves immediately, each
 message a safe acceptor sends \textit{transitively} references
 all messages it has ever sent or received.
Safe acceptors delay receipt of any message until they have
 received all messages it references.
This ensures they receive, for example, a $1a$ for a given ballot
 before receiving any $1b$s for that ballot.
\ifreport
It also ensures that safe acceptors' messages are always
 received in the order they were sent.
\fi

Each message has a unique ID and an identifiable type:
 $1a, 1b, $ or $2a$.
A \textit{2a} message $\green x$ has one type-specific field:
 $\green{x.lrn}$ specifies a learner.
A $1a$ message $\blue y$ has two type-specific fields:
  $\blue{y.value}$ is a proposed value, and
  $\blue{y.ballot}$ is a natural number specific to this
    proposal.

We assume that each $1a$ has a unique ballot number,
which could be accomplished by including signature information
 in the least significant bits of the ballot number:
\begin{assumption}[Unique ballot assumption]
\label{defn:ballotassumption} 
$\red z\!:\!1a \land \blue y\!:\!1a \land
 \red{z.ballot} = \blue{y.ballot} \Rightarrow \red z = \blue y$
\end{assumption}

\subsection{Machinery}
\label{sec:machinery}
To describe Heterogeneous Paxos, we require
 some mathematical machinery.

\ifreport\subsubsection{Transitive References}
\else\textbf{Transitive References:}\fi
\label{sec:tran}
We define $\tran{\green x}$ to be the transitive closure of message
 $\green x$'s references. 
Intuitively, these are all the messages in the ``causal past'' of
 $\green x$.
\begin{definition}
\label{defn:tran}
$
  \tran{\green x}
  \triangleq
    \cb{\green x}
  \cup
    \bigcup_{\blue m \in \green{x.\mathit{refs}}} \tran{\blue m}
$
\end{definition}

\ifreport\subsubsection{Get1a}
\else\textbf{Get1a:}
\fi
\label{sec:get1a}
It is useful to refer to the $1a$ that started the ballot of a
 message:
 the highest ballot number $1a$ in its transitive references.
\begin{definition}
\label{defn:geta}
$\displaystyle
\geta{\green x} \triangleq
\argmax_{\blue m:\textit{1a}\in\tran{\green x}}{\blue{m.ballot}}
$
\end{definition}

\ifreport\subsubsection{Ballot Numbers}
\else\textbf{Ballot Numbers:}
\fi
\label{sec:b}
The ballot number of a $1a$ is part of the message, and
 the ballot number of anything else is the highest ballot number among
 the \textit{1a}s it (transitively) references.
\begin{definition}
\label{defn:b}
$
  \ba{\green x} \triangleq \geta{\green x}.ballot
$
\end{definition}

\ifreport\subsubsection{Value}
\else\textbf{Value:}
\fi
\label{sec:va}
The value of a $1a$ is part of the message, and
 the value of anything else is the value of the highest ballot
 $1a$  among the messages it (transitively) references.
\begin{definition}
\label{defn:va}
$
  \va{\green x} \triangleq \geta{\green x}.value
$
\end{definition}

\ifreport\subsubsection{Decisions}
\else\textbf{Decisions:}
\fi
\label{sec:decision}
A learner decides when it has observed a set of \textit{2a}
 messages with the same ballot, sent by a quorum of acceptors. 
We call such a set a \textit{decision}:
\begin{definition}
  \label{defn:decision}
\ifreport
$
  \decision{\red a}{\red{q_a}} \triangleq
    \sig{\red{q_a}} \in \red{Q_a}
  \ \land\ 
  \forall \cb{\green x,\blue y} \subseteq \red{q_a} . \ \p{\andlinesThree
    {\ba{\green x}=\ba{\blue y}}
    {\green{x.lrn}=\red a}
    {\green x:\textit{2a}}}
$
\else
$
  \decision{\red a}{\red{q_a}} \triangleq
    \sig{\red{q_a}} \in \red{Q_a}
  \land
  \forall \cb{\green x,\!\blue y} \subseteq \red{q_a} . \ 
        \ba{\green x}\!=\!\ba{\blue y}
      \land
        \green{x.lrn}\!=\!\red a
      \land
        \green x\!:\!\textit{2a}
$
\fi
\end{definition}

Messages in a decision share a ballot
 (and therefore a value), so we extend our value
 function to include decisions:
$
  \decision{\red a}{\red{q_a}} 
  \Rightarrow
  \va{\red{q_a}} =
  \tallpipe{\va{\blue m}}{\blue m \in \red{q_a}}
$

Although decisions are not messages, applications might send decisions
 in other messages as a kind of ``proof of consensus.''
This is how the Heterogeneous Paxos \integrityattestations
 work in our prototype blockchains~(\cref{sec:charlotte}).

\ifreport\subsubsection{\hetdifftext{Caught}}
\else
\hetdifftext{
\textbf{Caught:}
}
\fi
Some behavior can create proof that an acceptor is \byzantine.
Unlike Byzantine Paxos, our acceptors and learners must
 adapt to \byzantine behavior. 
We say that an acceptor $\purple p$ is \textit{Caught} in a message
 $\green x$ if the transitive references of the messages include
 evidence such as two messages, $\red m$ and $\blue{m^\prime}$, both
 signed by $\purple p$, in which neither is featured in the other's
 transitive references (safe acceptors transitively reference
 all prior messages).

\begin{definition}
\label{defn:caught}
\ifreport
$
  \caught{\green x} \triangleq
  \cb{\tallpipe{\sig{\red m}}{\andlinesFour
        { \cb{\red m, \blue{ m^\prime}} \subseteq \tran{\green x} }
        {  \sig{\red m} = \sig{\blue{ m^\prime}}}
        { \red m \not\in \tran{\blue{m^\prime}}}
        { \blue{m^\prime} \not\in \tran{\red m}}
   }}
$
\else
$
  \caught{\green x} \triangleq
  \cb{\tallpipe{\sig{\red m}}{\andlinesTwo
        { \cb{\red m, \blue{ m^\prime}} \subseteq \tran{\green x}
          \ \land\ 
          \sig{\red m} = \sig{\blue{ m^\prime}}\hspace{-3mm}}
        { \red m \not\in \tran{\blue{m^\prime}}
          \ \land\ 
          \blue{m^\prime} \not\in \tran{\red m}}
   }}
$
\fi
\end{definition}
\ifreport\subsubsection{\hetdifftext{Connected}}
\else

  \hetdifftext{\textbf{Connected:}}
\fi
When some acceptors are proved \byzantine, clearly
 some learners need not agree, meaning that {\reallysafe}
 isn't in the edge between them in the \clg: at least one acceptor in
 each safe set in the edge is proven \byzantine.
Homogeneous learners are always connected unless there
 are so many failures no consensus is required.
\begin{definition}
\label{defn:con}
\ifreport
\[
  \con{\red a}{\green x} \triangleq
  \cb{\tallpipe{\blue b}{\andlinesTwo
        {\purple s \in \edge{\red a }{ \blue b} \in \clg}
        {\purple s \cap \caught{\green x} = \emptyset }
      }
     }
   \]
\else
$
  \con{\red a}{\green x} \triangleq
  \cb{\tallpipe{\blue b}{
        {\purple s \in \edge{\red a }{ \blue b} \in \clg}
    \ \land\ 
        {\purple s \cap \caught{\green x} = \emptyset }
      }
     }
     $
\fi
\end{definition}
It is clear that disconnected learners may not agree, and so
 each \textit{2a} message $\green x$ will have some implications only
 for learners still connected to its specified learner:
 $\con{\green{x.lrn}}{\green x}$.

\ifreport\subsubsection{Quorums in Messages}
\else\textbf{Quorums in Messages:}
\fi
\textit{2a} messages reference \textit{quorums of messages} with the same
 value and ballot.
A \textit{2a}'s quorums are formed from fresh \textit{1b} messages with the
 same ballot and value (we define \textit{fresh} in \cref{defn:fresh}).
\begin{definition}
\label{defn:qa}
\ifreport
\[
  \qa{\green x : \textit{2a}} \triangleq \cb{\tallpipe{\red m}{\andlinesFour
  {\red m : \textit{1b}}
  {\fresh{\hetdiff{\green{x.lrn}}}{\red m}}
  {\red m \in \tran{\green x}}
  {\ba{\red m} = \ba{\green x}}
}}
\]
\else
$
  \qa{\green x\!:\!\textit{2a}} \triangleq \cb{\hspace{-2mm}\tallpipe{\red m\!}{\!
  {\red m\!:\!\textit{1b}}\ \land\ 
  {\fresh{\hetdiff{\green{x.lrn}}}{\red m}}\ \land\ 
  {\red m \in \tran{\green x}}\ \land\ 
  {\ba{\red m} = \ba{\green x}}
}\hspace{-2mm}}
$
\fi
\end{definition}

\ifreport\subsubsection{Buried messages}
\else\textbf{Buried messages:}
\fi
A \textit{2a} message can become irrelevant if, after a time, an entire quorum of acceptors has
 seen \textit{2a}s with different values, \hetdifftext{the same learner}, and
 higher ballot numbers.
We call such a \textit{2a} \textit{buried} (in the context of some later message \purple y):
\begin{definition}
\label{defn:buried}
\ifreport
$
  \buried{\green x : \textit{2a}}{ \purple y} \triangleq
  \cb{\tallpipe{\sig{\red m}}{\andlinesSix
      {\red m \in \tran{\purple y}}
      {\blue z : \textit{2a}}
      {\cb{\green x, \blue z} \subseteq \tran{\red m}}
      {\va{\blue z} \ne \va{\green x}}
      {\ba{\blue z} > \ba{\green x}}
      {\hetdiff{\blue{z.lrn} = \green{x.lrn}}}
  }}
  \in \green{Q_{\hetdiff{x.lrn}}}
$
\else
\[
\begin{array}{l}
  \vspace{-8mm}
\\
  \buried{\green x : \textit{2a}}{ \purple y} \triangleq
  \cb{\tallpipe{\hspace{-2mm}\sig{\red m}\hspace{-1mm}}{\hspace{-2mm}\andlinesTwo
      {\hspace{-3mm}
       \red m \in \tran{\purple y}
       \ \land\ 
       \blue z\!:\!\textit{2a}
       \ \land\ 
       \cb{\green x, \blue z} \subseteq \tran{\red m}
       }
      {\hspace{-3mm}
       \va{\blue z} \ne \va{\green x}
       \ \land\ 
       \ba{\blue z} > \ba{\green x}
       \ \hetdiff{\land\ \blue{z.lrn} = \green{x.lrn}}
       \hspace{-3mm}}
  }}
  \in \green{Q_{\hetdiff{x.lrn}}}
\\
  \vspace{-3mm}
\end{array}
\]
\fi
\end{definition}

\ifreport\subsubsection{Well-Formedness}
\else\textbf{Well-Formedness:}
\fi
In addition to the basic message layout, \textit{2a} and \textit{1b}
 messages must be \textit{well-formed}.
No \textit{2a} should have an invalid quorum upon creation,
 and no acceptor should create a \textit{2a} unless it sent one of
 the \textit{1b} messages in the \textit{2a}.
Similarly, no \textit{1b} should reference any message with the same
 ballot number besides a \textit{1a} (safe acceptors make \textit{1b}s as soon as
 they receive a \textit{1a}).
Acceptors and learners should ignore messages that are not well-formed.
\begin{assumption}[Well-Formedness Assumption]
  \label{defn:wellformedness}
\[
  \begin{array}{l}
\green x : \textit{1b}\ \land
  \ \blue y \in \tran{\green x}
  \ \land\ \green x \ne \blue y 
  \ \land\ \blue y \ne \geta{\green x}
  \ \Rightarrow\ \ba{\blue y} \ne \ba{\green x}
\\
\red z : \textit{2a}\Rightarrow 
  \ \qa{\red z} \in \red{Q_{\hetdiff{z.lrn}}}
  \ \land\ \sig{\red z} \in \sig{\qa{\red z}}
\end{array}
\]
\end{assumption}

\ifreport\subsubsection{Connected \textit{2a} messages}
\else\textbf{Connected \textit{2a} messages:}
\fi
Entangled learners must agree, but learners that are not connected
 are not entangled, so they need not agree.
Intuitively, a \textit{1b} message references a \textit{2a} message
 to demonstrate that some learner may have decided some value.
For learner $\red a$, it can be useful to find the set of
 \textit{2a} messages from the same sender as a message ${\green x}$
 (and sent earlier) which are still unburied,
 and for learners connected to $\red a$.
The \textit{1b} cannot be used to make any new \textit{2a}
 messages for learner $\red a$ that have values
 different from these \textit{2a} messages.
\begin{definition}
  \label{defn:cona}
    \ifreport
    $
      \cona{\hetdiff{\red a}}{\green x} \triangleq
      \cb{\tallpipe{\blue m}{\andlinesFive
          {\blue m : \textit{2a}}
          {\blue m \in \tran{\green{x}}}
          {\sig{\blue m} = \sig{\green x}}
          {\lnot \buried{\blue m}{\green x}}
          {\hetdiff{\blue{m.lrn} \in \con{\red a}{\green x}}}
      }}
     $
    \else
    $
      \cona{\hetdiff{\red a}}{\green x} \triangleq
      \cb{\tallpipe{\blue m}{\andlinesTwo
          {\blue m\!:\!\textit{2a}
           \ \land\ 
           \blue m \in \tran{\green{x}}
           \ \land\ 
           \sig{\blue m} = \sig{\green x}
           \hspace{-3mm}}
          {\lnot \buried{\blue m}{\green x}
           \ \hetdiff{\land\ \blue{m.lrn} \in \con{\red a}{\green x}}}
      }}
     $
    \fi
\end{definition}

\ifreport\subsubsection{Fresh \textit{1b} messages}
\else\textbf{Fresh \textit{1b} messages:}
\fi
Acceptors send a \textit{1b} message whenever they receive a
 \textit{1a} message with a ballot number higher than they have yet
 seen.
However, this does not mean that the \textit{1b}'s value
 (which is the same as the \textit{1a}'s) agrees with that of
 \textit{2a} messages the acceptor has already sent.
We call a \textit{1b} message \textit{fresh}
 (with respect to a learner)
 when its value agrees with that of unburied \textit{2a} messages the
 acceptor has sent.
\begin{definition}
  \label{defn:fresh}
  $
    \fresh{\hetdiff{\red a}}{\green x : \textit{1b}} \triangleq
    \forall \blue m \in \cona{\hetdiff{\red a}}{\green x} . \ 
      \va{\green x} = \va{\blue m}
  $
\end{definition}

\ifreport
\subsection{The Heterogeneous Paxos Protocol}
\else
\subsection{Ballots}
\fi
\label{sec:protocol}
Heterogeneous Paxos can be thought of as taking place in
 \textit{stages} identified by natural numbers called ballots.
\Cref{sec:ballotnumbers} describes one way to construct unique ballot numbers.
\ifreport
Acceptor and Learner behavior is described
 in~\cref{fig:pseudocode}.

Heterogeneous Paxos reduces to a Byzantine Paxos variant in the
 case of homogeneous learners.
For a precise definition of this Byzantine Paxos variant, ignore the
 text \hetdifftext{highlighted in pale blue}.

It may be possible to optimize an implementation with additional
 messages pertaining to outdated ballots, but these are not strictly
 necessary for correctness.
For instance, an acceptor might inform a proposer that their
 proposal has been ignored because of its ballot number, and the
 proposer might then know to try again later.

\subsubsection{For Ballot $b$}
\label{sec:heterogeneousballot}
The actions in each ballot are broken down into three \textit{phases},
 traditionally called \textit{1a}, \textit{1b}, and
 \textit{2a}~\cite{paxos-made-simple}, corresponding to the subtypes
 of message.

\noindent\textbf{1a}:
A proposer proposes a value $v$ by creating a new ballot number $b$, and
 sending a \textit{1a} message containing $v$ and $b$ to all
 acceptors.

\noindent\textbf{1b}:
Upon receiving a \textit{1a}, an acceptor sends a \textit{1b}, which
 references all messages it has ever received (transitively).

\noindent\textbf{2a}:
If an acceptor receives a \textit{1b} message $\blue m$,
 and has never received a message with a higher ballot number,
 it creates a new \textit{2a} message $\blue{m^\prime}$
 \hetdifftext{for each learner}, and if
 $\blue{m^\prime}$ is well-formed~(\cref{defn:wellformedness}), and
 if $m\in\qa{m^\prime}$, it sends $\blue{m^\prime}$ to all
 acceptors and learners.

If a learner $\hetdiff{\blue{a}}$ receives \textit{2a} messages 
 \hetdifftext{for learner $\blue a$}
 all with ballot $b$ and value $v^\prime$, signed by
 \hetdifftext{one of learner $\blue a$'s}
 quorums of acceptors, then
 \hetdifftext{learner $\blue{a}$}
 decides value $v^\prime$.

\subsubsection{Multiple Ballots}
\else

\textbf{Multiple Ballots:}
\fi
\label{sec:multiple}
Proposers construct new $1a$ messages (with a value and a unique
 ballot number), and send them to all acceptors.
Just like in Homogeneous Byzantine Consensus, it is possible for a
 ballot to \textit{fail}: after some number of
 ballots, it may be the case that all messages have arrived, the
 protocol in~\cref{fig:pseudocode} doesn't require any acceptor to
 send any further messages, and yet no learner has decided. 
For this reason, it is necessary to start a new ballot when 
 an old one is failing.

One way to handle this is to leave the responsibility at the
 proposers:
 if a proposer proposes a ballot, and learners don't decide for a
 while, then the proposer should propose again.
Randomized exponential backoff can be used to allow clients to adapt
 to the unknown delay in a partially synchronous~\cite{Dwork1988}
 network without flooding the system.

Another way is to have acceptors propose after a ballot has
 failed: when sufficiently many \textit{1b} messages for a given
 ballot are collected,  but none are fresh, an acceptor could send a
 new \textit{1a}.
There are subtleties to ensuring liveness, which we discuss
 in~\cref{sec:terminationproof}.

\ifreport\subsubsection{Ballot Numbers}
\label{sec:ballotnumbers}
In our implementation, ballot numbers are constructed from
 lexicographically ordered
 pairs featuring the current time on the proposer's clock and the
 proposer signature of the hash of the value for that ballot.
The latter ensures no two \textit{1a}s with the same ballot have
 different values.
The former ensures that ballot numbers generally increase with time,
 which improves expected decision latency.
To prevent proposers from entering artificially high ballot numbers,
 each acceptor delays the receipt of any received message until its
 own clock exceeds the message's ballot number's time.
If all clocks are synchronized, the ``best'' a proposer can do is to
 use the correct time.
\fi

\newcommand{\thmvalidity}{
\ifreport
  \begin{restatable}[Validity]{theorem}{validity}
\else
  \begin{theorem}[Validity]
\fi
  \label{thm:validity}
Heterogeneous Paxos is Valid (\cref{defn:validity}):\\
$
  \decision{\red a}{\red{q_a}} \Rightarrow
  \exists \green x : 1a . \va{\green x} = \va{\red{q_a}}
$
\ifreport
  \end{restatable}
\else
  \end{theorem}
\fi
}

\newcommand{\thmagreement}{
\ifreport
  \begin{restatable}[Agreement]{theorem}{agreement}
\else
  \begin{theorem}[Agreement]
\fi
  \label{thm:agreement}
Heterogeneous Paxos has Agreement~(\cref{defn:agreement}):\\
$
  \entangled{\red a}{\blue b} \land 
  \decision{\red a}{\red{q_a}} \land 
  \decision{\blue b}{\blue{q_b}} \Rightarrow
  \va{\red{q_a}} = \va{\blue{q_b}}
$
\ifreport
  \end{restatable}
\else
  \end{theorem}
\fi
}

\ifreport
\section{Correctness}
\label{sec:correctness}

\subsection{Useful Lemmas}
\label{sec:lemmas}
First, we build up some useful facts about Heterogeneous Paxos.

\subsubsection{\con{}{} and \caught{}}
\begin{lemma}{Safe Acceptors are Uncaught:}
  \label{lemma:safeacceptoruncaught}
No safe acceptor can be caught in any message.
\[
  \purple p \textrm{ is safe } \Rightarrow
  \purple p \not\in \caught{\green x}
\]
\end{lemma}
\begin{proof}
By the definitions of correct behavior, the unforgeability of
 signatures, and the definition of Caught (\cref{defn:caught}).
\end{proof}

\begin{lemma}{Quorum Intersection:}
\label{lemma:quorumintersection}
If a message $\green x$ references two quorums' messages, and the
 learners for those quorums are still connected as of $\green x$,
 then there must be an uncaught acceptor who signed at least one
 message in each.
  \[
    \p{\andlinesFour
      {\sig{\red{q_a}}\in \red{Q_a} }
      {\sig{\blue{q_b}} \in \blue{Q_b} }
      {\red{q_a}\cup\blue{q_b}\subseteq \tran{\green x} }
      {\blue b \in \con{\red a}{\green x}}
  }\Rightarrow 
    \exists \red{m_a}\in\red{q_a},
    \blue{m_b}\in\blue{q_b},
    \purple p \not\in \caught{\green x}.\ 
      \sig{\red{m_a}} = \sig{\blue{m_b}} = \purple p
  \]
\end{lemma}
\begin{proof}
  By the definition of Con (\cref{defn:con}):
  \[
    \exists \purple s\in \edge{\red a}{\blue b} \in \clg 
    \land  \purple s \cap \caught{\green x} = \emptyset
  \]
Therefore, by the definition of \clg and quorum properties:
\[
  \sig{\red{q_a}} \cap \sig{\blue{q_b}} \cap \purple s \ne \emptyset
\]
\[
\therefore \exists \purple p \in \sig{\red{q_a}} \cap
                                 \sig{\blue{q_b}}\cap
                                 \purple s 
\]
\[
  \purple p \in \sig{\red{q_a}} \Rightarrow
  \exists \red{m_a}\in\red{q_a}. \sig{\red{m_a}} = \purple p
\]
\[
  \purple p \in \sig{\blue{q_b}} \Rightarrow
  \exists \blue{m_v}\in\blue{q_b}. \sig{\blue{m_b}} = \purple p
\]
\[
  \purple p \in \purple s \Rightarrow
  \purple p \not\in \caught{\green x}
\]
\end{proof}

\subsubsection{Entangled Learners}
\label{sec:entangledlearners}

\begin{lemma}{Entangled Learners are Connected:}
\label{lemma:entangledconnected}
If two learners are entangled, then for all messages that are
 actually sent, they are connected.
  \[
    \entangled{\red a}{\blue b} \Rightarrow
    \red a \in \con{\blue b}{\green x}
  \]
\end{lemma}
\begin{proof}
  By the definition of Entangled (\cref{defn:entangled}):
  \[
  \reallysafe \in \edge{\red a}{\blue b}
  \]
  Therefore, by \cref{lemma:safeacceptoruncaught}:
  \[
    \reallysafe \cap \caught{\green x} = \emptyset
  \]
  So by the definition of \textit{Con} (\cref{defn:con}):
  \[
    \red a \in \con{\blue b}{\green x}
  \]
\end{proof}

\begin{lemma}{Entanglement is not ordered:}
\label{lemma:entanglementunordered}
\[
\entangled{\red a}{\blue b} \Rightarrow \entangled{\blue b}{\red a}
\]
\end{lemma}
\begin{proof}
  By the definition of entangled (\cref{defn:entangled}):
  \[
    \reallysafe \in \edge{\red a}{\blue b}
  \]
  \clg is undirected, so
   $\edge{\red a}{\blue b} = \edge{\blue b}{\red a}$.
  Therefore:
  \[
    \reallysafe \in \edge{\blue b}{\red a}
  \]
  And so:
  \[
    \entangled{\blue b}{\red a}
  \]
\end{proof}

\begin{lemma}{Entanglement is Transitive:}
\label{lemma:entanglementtrans}
If $\red a$ and $\blue b$ are entangled, and $\blue b$ and $\green c$
 are entangled, then so are $\red a$ and $\green c$.
\[
    \entangled{\red a}{\blue b}
  \land
    \entangled{\blue b}{\green c}
\Rightarrow
  \entangled{\red a}{\green c}
\]
\end{lemma}
\begin{proof}
  By the definition of entangled (\cref{defn:entangled}):
  \[
    \andlinesTwo
    {\reallysafe \in \edge{\red a}{\blue b}}
    {\reallysafe \in \edge{\blue b}{\green c}}
  \]
  By the definition of condensed (\cref{defn:condensed}):
  \[\reallysafe \cup \reallysafe
    \in \edge{\red a}{\green c}\]
  \[\therefore \reallysafe
               \in \edge{\red a}{\green c}\]
  And so by the definition of entangled (\cref{defn:entangled}):
\[
\entangled{\red a}{\green c}
\]
\end{proof}

\begin{lemma}{Self-Entanglement:}
  \label{lemma:entanglementrefl}
If a learner is entangled to anyone, it must also be
 entangled with itself.
  \[
    \entangled{\red a}{\blue b} \Rightarrow \entangled{\red a}{\red a}
  \]
\end{lemma}
\begin{proof}
  By \cref{lemma:entanglementunordered},
  $
    \entangled{\blue b}{\red a}
  $.
  And so by \cref{lemma:entanglementtrans}:
  $
    \entangled{\red a}{\red a}
  $.
\end{proof}

\begin{lemma}{After Observing a Decision:}\label{lemma:afterquorum}
If a learner observes a quorum of 2as with the same ballot, then any
 2as for an entangled learner with a higher ballot must have the same
 value.
  \[
      \p{\andlinesFive
        {\entangled{\red a}{\blue b} }
        { \decision{\red a}{\red{q_a}}}
        {\orange w : \textit{2a}}
        {\orange{w.lrn} =\blue b}
        {\ba{\orange w} > \ba{\red {q_a}}}
      } \Rightarrow \va{\red {q_a}} = \va{\orange w}
  \]
(Using our definition of $\va{}$ from~\cref{sec:decision})
\end{lemma}
\begin{proof}
  As ballot numbers are natural numbers, we prove this by
   \tb{induction} on $\ba{\orange w}$.

  \noindent\tb{Base Case:} for
    $\ba{\orange w} \le \ba{\red {q_a}}$,
    \cref{lemma:afterquorum} trivially holds.

  \noindent\tb{Induction:}
  Assume \cref{lemma:afterquorum} holds for all values of
   $\orange w < \orange{w^\prime}$.
  We now show that it holds for $\orange{w^\prime}$.

  By well-formedness (\cref{defn:wellformedness}), 
  \[
    \qa{\orange{w^\prime}} \in \blue{Q_b}
  \]
  By \cref{lemma:entangledconnected}:
  \[
    \blue b \in \con{\red a}{\orange{w^\prime}}
  \]
  By \cref{lemma:quorumintersection}, 
  \[
    \exists \red{m_a} \in \red{q_a},
            \orange{m_w} \in \qa{\orange{w^\prime}} ,
            \purple p \not\in \caught{\orange{w^\prime}}\ .\ 
        \sig{\red{m_a}} = \sig{\orange{m_b}} = \purple p
  \]
  By the definition of Well-formed (\cref{defn:wellformedness}):
  \[
    \forall \p{\purple r : \textit{2a}}\in\tran{\orange{m_w}}\ .\ 
      \ba{\purple r} < \ba{\orange{m_w}}
  \]
  By the definition of $q$ (\cref{defn:qa}), $\orange{m_w}$ is fresh.
  By \cref{lemma:entangledconnected},
   $\blue b \in \con{\red a}{\orange{w^\prime}}$,
   and so by the definition of fresh (\cref{defn:fresh}),
   either $\red{m_a} \in \cona{\blue b}{\orange{m_w}}$,
   in which case:
   \[
     \va{\red{m_a}} =
     \va{\red {q_a}} =
     \va{\orange{w^\prime}} =
     \va{\orange{m_w}}
   \]
   or $\buried{\red{m_a}}{\orange{m_w}}$.
  However, by the definition of buried (\cref{defn:buried}):
\[
  \exists \p{\red{r} : \textit{2a}} \in \tran{\orange{m_w}}\ .\ 
      \red{r.lrn = \red a}
    \land
      \p{ \ba{\red r} > \ba{\red {m_a}}}
    \land
      \p{\va{\red r} \ne \va{\red{m_a}}}
\]
By \cref{lemma:entanglementrefl}, $\entangled{\red a}{\red a}$,
 so by our induction hypothesis, no such $\red{r}$ exists.
Therefore, $\red{m_a}$ is not buried, so:
\[
  \va{\red {q_a}} = \va{\orange{w^\prime}} 
\]
\end{proof}
 
\subsection{Validity}
\label{sec:correctnessvalidity}
If a learner decides a value, that value was proposed.
\ifreport
\thmvalidity
\else
Recall~\cref{thm:validity}:
\validity*
\fi
\begin{proof}
  By the definition of Get1A for Decisions (\cref{defn:geta}), 
  \[
    \geta{\red{q_a}} : 1a \land 
    \va{\geta{\red{q_a}}} = \va{\red{q_a}}
  \]
\end{proof}

\subsection{Agreement}
\label{sec:correctnessagreement}
If two entangled learners ($\red a$ and $\blue b$) each decide,
 they decide the same value~(\cref{sec:agreement}).
\ifreport
\thmagreement
\else
Recall~\cref{thm:agreement}:
\agreement*
\fi
\begin{proof}
  If $\ba{\red{q_a}} = \ba{\blue{q_b}}$, then by the ballot
   assumption~(\cref{defn:ballotassumption}), and definition of
   $V$~(\cref{sec:decision}): 
   \[
     \va{\red{q_a}} = \va{\blue{q_b}}
   \]
  Otherwise, without loss of generality, assume
   $\ba{\red{q_a}} < \ba{\blue{q_b}}$.
  By the definition of decision (\cref{defn:decision}):
  \[
    \exists \blue m : \textit{2a} \in \blue{q_b}\ .\ 
    \va{\blue m} = \va{\blue{q_b}}
  \]
  By \cref{lemma:afterquorum}:
  \[
    \va{\blue m} = \va{\blue{q_b}} = \va{\red{q_a}}
  \]
\end{proof}
\else
\subsection{Safety}
\label{sec:safety}
Under our assumptions~(\cref{sec:heterogeneousassumptions}),
 Heterogeneous Paxos has the safety properties of Validity and
 Agreement (proofs in~\cref{sec:correctnessvalidity}
 and~\cref{sec:correctnessagreement}):
\thmvalidity
\thmagreement
\fi

\subsection{Liveness}
\label{sec:liveness}
\ifreport
\subsubsection{Network Assumption}
\fi
Heterogeneous Paxos, and indeed Byzantine Paxos,
 rely on a weak network assumption to guarantee
 termination.
The assumption is complex precisely because it is weak;
 a simpler but stronger assumption, such as a partially synchronous
 network, would suffice.

\ifreport
  \begin{restatable}[Network Assumption]{assumption}{networkassumption}
\else
  \begin{assumption}[Network Assumption]
\fi
\label{defn:networkassumption}
To guarantee that a learner $\red a$ decides, we assume that for
 some quorum $\red{q_a} \in \red{Q_a}$ of safe and live acceptors:
\begin{itemize}
  \item Eventually, there will be 13 consecutive periods of any
         duration, with no time in between, numbered 0 through 12,
         such that any message sent to $\red a$ or an acceptor in
         $\red{q_a}$ before one period begins is
         delivered before it ends.
  \item If an acceptor in $\red{q_a}$ sends a message in between
         receiving two messages $\red m$ and $\orange{ m^\prime}$
         (and it receives no other messages in between),
         and $\red m$ is delivered in some period $n$, then the
         message is sent in period $n$.
  \item No 1a message except $\green x$, $\blue y$, and $\purple z$
         is delivered to any acceptor in $\red{q_a}$ during any
         \ifreport of the 13 periods\else period\fi.
  \item $\green x$ is delivered to an acceptor in $\red{q_a}$ in period 0,
        $\blue y$ is delivered to an acceptor in $\red{q_a}$ in period 4, and
        $\purple z$ is delivered to an acceptor in $\red{q_a}$ in period 9.
  \item $\va{\blue y} = \va{\purple z}$ is the value of the highest
         ballot 2a known to any acceptor in $\red{q_a}$ at the
         end of period 3.
  \item $\ba{\green x}$ is greater than any ballot number of any
         message delivered to any acceptor in $\red{q_a}$ before period 0, and
         $\ba{\green x} < \ba{\blue y} < \ba{\purple z}$.
\end{itemize}
\ifreport
  \end{restatable}
\else
  \end{assumption}
\fi
\noindent This assumption is \textit{only necessary} for termination, not any
 safety property.
\ifreport
Furthermore, it is possible to use artificial delays to reduce other
 network assumptions to this
 one~(\cref{sec:semisynchronous}).
\else
We prove our termination theorem in~\cref{sec:terminationproof}.
\fi
\ifreport
  \begin{restatable}[Termination]{theorem}{termination}
\else
  \begin{theorem}[Termination]
\fi
\label{thm:termination}
  If \cref{defn:networkassumption} holds for learner $\red a$,
   then $\red a$ has Termination~(\cref{defn:termination}).
  Specifically, after period 12:
  $
    \terminating{\red a} \Rightarrow
    \exists\red{q_a}. \decision{\red a}{\red{q_a}}
  $
  If \cref{defn:networkassumption} holds for all terminating learners,
   then Heterogeneous Paxos has Termination.
\ifreport
  \end{restatable}
\else
  \end{theorem}
\fi
\ifreport
\label{sec:terminationproof}
\begin{proof}
Periods 0-3 are sufficient for all acceptors in $\red{q_a}$ to send \textit{1b}s
 (and \textit{2a}s) in response to any \textit{1a} delivered prior to period 0. 

By the end of period 1, $\green x$ will be delivered to all
 acceptors in $\red{q_a}$.
They will therefore cease generating \textit{2a}s for any ballot number
 $< \ba{\green x}$.

By the end of period 3, all acceptors in $\red{q_a}$ will have received
 all \textit{1b}s and \textit{2a}s generated by acceptors in $\red{q_a}$ with
 $\le \ba{\green x}$.
This means that any prior \textit{2a} signed by an acceptor in
 $\red{q_a}$  is received by all acceptors in $\red{q_a}$ .
If there are no prior \textit{2a}s, there will now be \textit{2a}s with value
 $\va{\green x}$.

By the end of period 5, $\blue y$ will be delivered to all acceptors
 in $\red{q_a}$ .
They will then send \textit{1b}s to each other.

By the end of period 8, all \textit{2a}s generated by acceptors in
 $\red{q_a}$  with values $\ne \va{\blue y}$ will be buried. 

By the end of period 10, $\purple z$ is delivered to all
 acceptors in $\red{q_a}$.
All acceptors in $\red{q_a}$ must respond with fresh \textit{1b}s.

By the end of period 11, all acceptors in $\red{q_a}$ receive a quorum
  of fresh \textit{1b}s, and so produce \textit{2a}s.

By the end of period 12, a quorum of \textit{2a}s with ballot $\ba{\purple z}$
 have been delivered.
These form $\red{q_a}$.
\end{proof}

\subsubsection{Introducing Artificial Delays in a Partially Synchronous Network}
\label{sec:semisynchronous}
\fi
A \textit{partially synchronous} network is one in which, after some
 point in time,  there exists some (possibly unknown) constant
 latency $\Delta$ such that all sent messages arrive within
 $\Delta$~\cite{Dwork1988}. 
\ifreport
We can introduce artificial delays to ensure that in any partially
 synchronous network, Heterogeneous Paxos has Termination.

If the set of proposers is finite (it could be limited to the set of
 acceptors in any learner's quorum), then each can be assigned
 specific times they are allowed to propose.
For example, if proposals include a timestamp, and all safe
 acceptors delay receipt of a \textit{1a} until after that time,
 there could  be a limited set of timestamps each proposer is allowed
 to use.

Allowed times can be allocated in turns: each proposer gets a
 span of time during which only they can propose, and
 these turns are allocated in a round-robin fashion. 
If turns grow exponentially longer with time (from some predetermined
 start time), then for any maximum message latency $\Delta$ and
 any maximum clock skew $\Delta^\prime$, the network assumption
 will eventually be met during the turn of a correct proposer,
 with a turn of dura\-tion exceeding $13\p{\Delta + \Delta^\prime}$.
The correct proposer need only propose once, wait a third of its
 turn, and then propose twice more, using the value of the highest
 known \textit{2a} to any safe acceptor. 
\else
We explain elsewhere how to add artificial message
 receipt delays to Heterogeneous Paxos in order to
 guarantee~\cref{defn:networkassumption} in a partially synchronous
 network (\Cref{sec:semisynchronous}).
\fi

\ifreport
\section{Repeated Consensus}
\label{sec:repeatedconsensus}
In maintaining, for instance, an ordered log, it is useful for
 learners to \textit{decide} on the value which goes in each
 \textit{slot}, traditionally starting at slot 0, then 1, etc.
In general, one might want to prohibit filling a slot before a
 previous slot has been filled.
With homogeneous learners, one might say that \textit{1a} messages
 should be ignored unless they can show that consensus for the
 previous slot was reached.
For instance, an acceptor might ignore a \textit{1a} for slot $n$
 unless it references \textit{2a} messages signed by a quorum of
 acceptors which share a ballot and a value identified as belonging
 in slot $n-1$. 
In other words, an acceptor can demand \textit{proof of consensus}
 for slot $n-1$ before filling slot $n$ (for $n \ne 0$).

Heterogeneous learners makes this concept more difficult.
What if the previous slot has been filled for one learner, but not
 for another?
What if they are filled with different values for different learners?
We describe a few possible solutions:

\subsection{Allow slots to be filled in any order.}
\label{sec:fillanyorder}
Each consensus protocol for each slot is fully independent.
This is easier to implement, technically correct, and probably
 acceptable for some applications.

\subsection{A \textit{1a} for slot $s$ is a \textit{1a} for slot $s-1$.}
\label{sec:reference1a}
Suppose each \textit{1a} for slot $s > 0$ is required to reference a
 \textit{1a} for slot $s-1$. 
Furthermore, each safe acceptor delays receipt of each
 \textit{1a} until it has received and acted on the \textit{1a}s
 referenced.
Other than that, we treat all slots independently.

No slot could be filled for any learner without consensus at least
 having \textit{begun} for all prior slots. 
This does not guarantee, however, that consensus has yet
 \textit{finished} for all prior slots.
However, given termination~(\cref{sec:termination}), it
 \textit{will finish}, and so all prior slots \textit{will be filled}.

This does not guarantee that the value decided in slot $s$ references
 the value decided in slot $s-1$.

\subsection{Keep track of proof of consensus per learner.}
\label{sec:trackperlearner}
A proposer could put all known proofs for slot $s-1$ in
 each \textit{1a} for slot $s$, and acceptors only consider those
 learners in the \textit{2a} phase for the ballot that \textit{1a}
 begins.
This pushes the duty of tracking proofs of consensus onto the proposer.

Later \textit{1a} messages might feature more learners, thus enabling
 more learners to achieve consensus for slot $s$.

This is compatible with the \textit{1a} for slot $s$
 references a \textit{1a} for slot $s-1$
 solution~(\cref{sec:reference1a}).
\fi

\ifreport
\section{Examples}
\label{sec:examples}
In \cref{sec:intro}, we contrast Heterogeneous Paxos
 with traditional consensus in three ways:
\begin{itemize}
  \item Heterogeneous failures
  \item Heterogeneous acceptors
  \item Heterogeneous learners
\end{itemize}
We illustrate the advantages of heterogeneity through examples with
 all eight combinations of heterogeneous or homogeneous failures,
 acceptors, and learners.

\subsection{Fully Homogeneous}
\label{sec:examplehomogeneous}
Consider a traditional \byzantine-tolerant consensus protocol
 with 4 acceptors, tolerating any 1 \byzantine failure.
Quorums (for any learner) consist of any 3 acceptors.
Safe sets (for any edge) likewise consist of any 3 acceptors.
We illustrate this traditional scenario in \cref{fig:fullhom}.

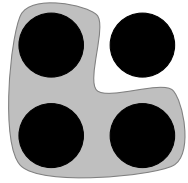
\begin{figure}
  \centering
\begin{tikzpicture}[scale=2]
\draw[gray,fill=gray!50] plot [smooth cycle] coordinates {
  (0,0) (0,1) (.5,1) (.5,.5) (1,.5) (1,0)};
\node[circle, fill=black, inner sep=3mm, draw](A) at (.2,.2) {\ };
\node[circle, fill=black, inner sep=3mm, draw](B) at (.2,.8) {\ };
\node[circle, fill=black, inner sep=3mm, draw](C) at (.8,.2) {\ };
\node[circle, fill=black, inner sep=3mm, draw](D) at (.8,.8) {\ };
\end{tikzpicture}
\caption[Fully Homogeneous Example]{
Fully Homogeneous Example: acceptors are shown as black circles,
 and a quorum is shown as a shaded region.}
\label{fig:fullhom}
\end{figure}

\subsection{Heterogeneous Failures}
\label{sec:exampleheterogeneousfailures}
Heterogeneous Paxos can express protocols wherein learners and
 acceptors are homogeneous, but mixed failures~\cite{Siu1998} are
 allowed.
For instance, consider a protocol with 6 acceptors, tolerating at
 most 1 \byzantine failure, and 1 additional \crash failure.
Quorums (for any learner) consist of any 4 acceptors.
Safe sets (for any edge) consist of any 5 acceptors: the crash failure
 is safe, but not live.
We illustrate this scenario in 
\cref{fig:hetfail}.

\begin{figure}
  \centering
\begin{tikzpicture}[scale=1.2]
\draw[gray,fill=gray!50] plot [smooth cycle] coordinates {
    ( .6* 1.5,-.5  *1.5)
    (-.05*1.5, .65 *1.5)
    (-.5* 1.5,.866 *1.5)
    (-1*  1.5, 0)
    (-.5* 1.5,-.866*1.5)
    (.5 * 1.5,-.866*1.5)
  };
\node[circle, fill=black, inner sep=3mm, draw](A) at (-1,0) {\ };
\node[circle, fill=black, inner sep=3mm, draw](B) at (-.5,-.866) {\ };
\node[circle, fill=black, inner sep=3mm, draw](C) at (.5,-.866) {\ };
\node[circle, fill=black, inner sep=3mm, draw](D) at (1,0) {\ };
\node[circle, fill=black, inner sep=3mm, draw](E) at (.5,.866) {\ };
\node[circle, fill=black, inner sep=3mm, draw](F) at (-.5,.866) {\ };
\end{tikzpicture}
\caption[Heterogeneous Failures Example]{
  Heterogeneous Failures Example: acceptor are shown as circles,
   and a quorum is shown as a shaded region.}
\label{fig:hetfail}
\end{figure}
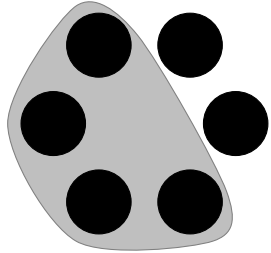

Note that in order to tolerate 2 total failures, a homogeneous
 \byzantine-fault-tolerant consensus protocol would need at least 7
 acceptors.
Heterogeneity spares the expense (in latency and resources) of an
 unnecessary additional acceptor.

\subsection{Heterogeneous Acceptors}
\label{sec:exampleheterogeneousacceptors}
Heterogeneous Paxos can express protocols wherein learners and
 failures are homogeneous, but not all acceptors are the same.
For instance, consider a protocol with 6 acceptors, divided into
 two groups of 3.
We tolerate up to 2 \byzantine failures, but only if
 \textit{all failures are in one group}. 
This makes the acceptors heterogeneous: it matters \textit{which}
 two acceptors fail.
Quorums (for any learner) consist of 3 acceptors from one group,
 and one acceptor from the other.
Safe sets (for any edge) likewise consist of 3 acceptors from one
 group, and one from the other.
This scenario is illustrated in \cref{fig:heta}.

 \begin{figure}
   \centering
\begin{tikzpicture}[scale=1]
\draw[gray,fill=gray!50] plot [smooth cycle] coordinates {
    (-0.4, -0.4)
    (-0.4,  2.3)
    (1,     2.3)
    (1,     0.5)
    (2.3,   0.5)
    (2.3,  -0.5)
  };
\node[circle, fill=blue, inner sep=3mm, draw=blue, line width=0mm](A) at (0,0) {\ };
\node[circle, fill=blue, inner sep=3mm, draw=blue, line width=0mm](B) at (0,1) {\ };
\node[circle, fill=blue, inner sep=3mm, draw=blue, line width=0mm](C) at (0,2) {\ };
\node[circle, fill=red!30, inner sep=2.57mm, draw=red, line width=1.3mm](D) at (2,0) {\ };
\node[circle, fill=red!30, inner sep=2.57mm, draw=red, line width=1.3mm](E) at (2,1) {\ };
\node[circle, fill=red!30, inner sep=2.57mm, draw=red, line width=1.3mm](F) at (2,2) {\ };
\end{tikzpicture}
\caption[Heterogeneous Acceptors Example]{
  Heterogeneous Acceptors Example: 
  Here, we draw one group of acceptors as solid blue circles,
   and the other as hollow red circles.
A quorum is shown as a shaded region.
}
\label{fig:heta}
\end{figure}
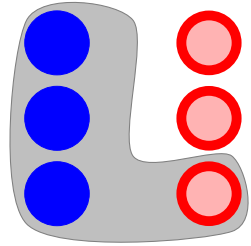

Note that in order to tolerate 2 total failures, a homogeneous
 \byzantine-fault-tolerant consensus protocol would need at least 7
 acceptors.
Heterogeneity spares the expense (in latency and resources) of an
 unnecessary additional acceptor
It is possible to express this kind of heterogeneity with quorums in
 Byzantine Paxos~\cite{byzantizing-paxos}.

\subsection{Heterogeneous Failures and Acceptors}
\label{sec:exampleheterogeneousfailuresacceptors}
Heterogeneous Paxos can express protocols wherein learners
 are homogeneous, but with mixed failures, and heterogeneous
 acceptors.
For instance, consider a protocol with 8 acceptors, divided into
 two groups of 4.
We tolerate up to 2 \byzantine failures, and 1 additional crash
 failure, but only if \textit{all \byzantine failures are in
 one group, and at most 2 failures occur in the same group}.
This makes the acceptors heterogeneous: it matters \textit{which}
 acceptors fail.
Quorums (for any learner) consist of 3 acceptors from one group,
 and 2 acceptors from the other.
Safe sets (for any edge) consist of 4 acceptors from one group, and 2
 acceptors from the other.
We illustrate this example in 
\cref{fig:hetfaila}.
 \begin{figure}
   \centering
\begin{tikzpicture}[scale=1]
\draw[gray,fill=gray!50] plot [smooth cycle] coordinates {
    (-0.4, -0.4)
    (-0.4,  2.3)
    (.8,    2.3)
    (1,     1.5)
    (2.4,   1.4)
    (2.3,  -0.5)
  };
\node[circle, fill=blue, inner sep=3mm, draw=blue, line width=0mm](A) at (0,0) {\ };
\node[circle, fill=blue, inner sep=3mm, draw=blue, line width=0mm](B) at (0,1) {\ };
\node[circle, fill=blue, inner sep=3mm, draw=blue, line width=0mm](C) at (0,2) {\ };
\node[circle, fill=blue, inner sep=3mm, draw=blue, line width=0mm](D) at (0,3) {\ };
\node[circle, fill=red!30, inner sep=2.57mm, draw=red, line width=1.3mm](E) at (2,0) {\ };
\node[circle, fill=red!30, inner sep=2.57mm, draw=red, line width=1.3mm](F) at (2,1) {\ };
\node[circle, fill=red!30, inner sep=2.57mm, draw=red, line width=1.3mm](G) at (2,2) {\ };
\node[circle, fill=red!30, inner sep=2.57mm, draw=red, line width=1.3mm](H) at (2,3) {\ };
\end{tikzpicture}
\caption[Heterogeneous Failures and Acceptors Example]{
  Heterogeneous Failures and Acceptors Example:
  Here, we draw one group of acceptors as solid blue circles,
   and the other as hollow red circles.
  A quorum is shown as a shaded region.
}
\label{fig:hetfaila}
\end{figure}
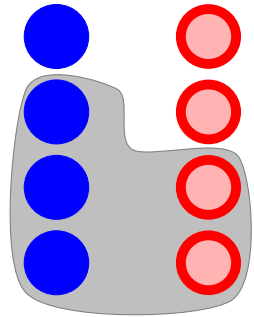

Note that in order to tolerate 3 total failures, a homogeneous
 \byzantine-fault-tolerant consensus protocol would need at least 10
 acceptors.
Heterogeneity spares the expense (in latency and resources) of 2
 unnecessary additional acceptors.

To address a similar situation with homogeneous failures, we'd
 still need 10 acceptors (to tolerate 2 \byzantine failures in one
 group, and 1 in the other).
To tolerate 2 \byzantine failures and one crash failure
 (so heterogeneous failures) with homogeneous acceptors, we'd need
 9 total acceptors.
The additional detail of heterogeneous acceptors spares the
 expense (in latency and resources) of an unnecessary additional
 acceptor, as opposed to just heterogeneous failures.

\subsection{Heterogeneous Learners}
\label{sec:exampleheterogeneouslearners}
Heterogeneous Paxos can express protocols wherein learners have
 different failure assumptions, but each assumes acceptors and
 failures are homogeneous. 
This is the sort of scenario the original Ripple consensus
 protocol~\cite{Schwartz2014} tried to address, although it lacks
 liveness~\cite{Chase2018}.

\subsubsection{Acceptor Disagreement}
\label{sec:acceptordisagreement}
We discussed an acceptor disagreement example in
 \cref{sed:clgexample}, and we expand upon it here.

Suppose learners agree that they want to tolerate 1
 \byzantine failure out of 4 acceptors, but disagree about who the
 4 relevant acceptors are.
Suppose there are 4 learners: 2 red and 2 blue, as well as 5
 acceptors: 1 red, 1 blue, and 3 black.
The acceptors are illustrated in 
\cref{fig:membershipdisagreement}.

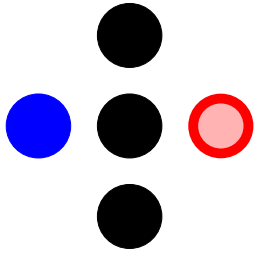
\begin{figure}
\centering
\begin{tikzpicture}[scale=1.2]
\membershipdisagreement{0}{0}{Final}{1}{1}{10mm}
\end{tikzpicture}
\caption[Membership Disagreement Example]{
  Membership Disagreement Example:
  Here, we draw one the blue acceptor as a solid blue circle,
   the red acceptor as a hollow red circle, and the black
   acceptors as black circles.
}
\label{fig:membershipdisagreement}
\end{figure}

The red learners want to agree if there is at most 1 \byzantine
 failure among the red and black acceptors, even if the blue
 acceptor has failed.
Likewise, the blue learners want to agree if there is at most 1
 \byzantine failure among the blue and black acceptors, even if
 the red acceptor has failed.
The red learners and blue learners acknowledge that they may
 disagree with each other if a black acceptor fails, but otherwise
 they want to agree.
Thus we draw the learner graph~(\cref{sec:learnergraph}) in,
\cref{fig:membershipdisagreementclg}.

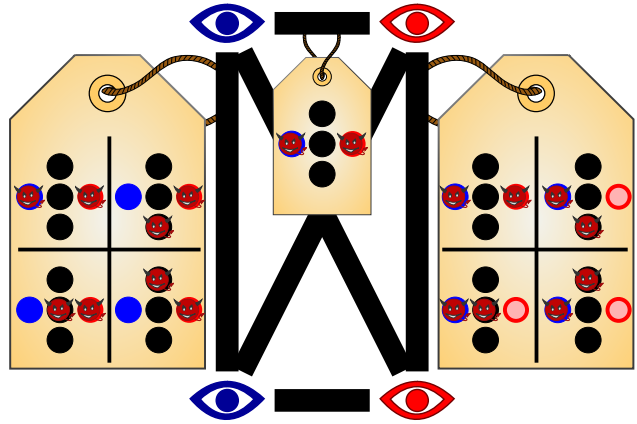
\begin{figure}
  \centering
\begin{tikzpicture}[scale=1]
  \pgfdeclareimage[width=10mm]{eye-blue}{figures/icons/eye-icon-custom-larger-darkblue.pdf}
  \pgfdeclareimage[width=10mm]{eye-red}{figures/icons/eye-icon-custom-outlinered.pdf}
  \pgfdeclareimage[width=13mm]{label}{figures/icons/label.pdf}
  \pgfdeclareimage[width=30mm]{label-sideways}{figures/icons/label-sideways.pdf}

  \node[](blue0) at (0,0) {\pgfuseimage{eye-blue}};
  \node[](blue1) at (0,5) {\pgfuseimage{eye-blue}};
  \node[](label0) at (-1.37,2.5) {\pgfuseimage{label-sideways}};
  \draw[line width=3mm] (blue0) -- (blue1);

  \membershipdisagreement{-2.6}{2.3}{label0universe0}{0.4}{0.4}
  \membershipdisagreementfailure{-2.6}{2.3}{label0universe0}{0.4}{0.4}{4mm}{A}{byz}
  \membershipdisagreementfailure{-2.6}{2.3}{label0universe0}{0.4}{0.4}{4mm}{E}{byz}

  \draw[line width=0.5mm, black] (-1.55, .5) -- (-1.55, 3.5);

  \membershipdisagreement{-1.3}{2.3}{label0universe1}{0.4}{0.4}
  \membershipdisagreementfailure{-1.3}{2.3}{label0universe1}{0.4}{0.4}{4mm}{B}{byz}
  \membershipdisagreementfailure{-1.3}{2.3}{label0universe1}{0.4}{0.4}{4mm}{E}{byz}

  \draw[line width=0.5mm, black] (-2.75, 2) -- (-0.35, 2);

  \membershipdisagreement{-2.6}{.8}{label0universe2}{0.4}{0.4}
  \membershipdisagreementfailure{-2.6}{.8}{label0universe2}{0.4}{0.4}{4mm}{C}{byz}
  \membershipdisagreementfailure{-2.6}{.8}{label0universe2}{0.4}{0.4}{4mm}{E}{byz}

  \membershipdisagreement{-1.3}{.8}{label0universe3}{0.4}{0.4}
  \membershipdisagreementfailure{-1.3}{.8}{label0universe3}{0.4}{0.4}{4mm}{D}{byz}
  \membershipdisagreementfailure{-1.3}{.8}{label0universe3}{0.4}{0.4}{4mm}{E}{byz}

  \node[](red0) at (2.5,0) {\pgfuseimage{eye-red}};
  \node[](red1) at (2.5,5) {\pgfuseimage{eye-red}};
  \draw[line width=3mm] (blue0) -- (red1);
  \draw[line width=3mm] (blue1) -- (red0);
  \draw[line width=3mm] (blue0) -- (red0);
  \node[xscale=-1](label1) at (1.25,3.8) {\pgfuseimage{label}};
  \draw[line width=3mm] (blue1) -- (red1);

  \membershipdisagreement{.85}{3}{label1universe0}{0.4}{0.4}
  \membershipdisagreementfailure{.85}{3}{label1universe0}{0.4}{0.4}{4mm}{E}{byz}
  \membershipdisagreementfailure{.85}{3}{label1universe0}{0.4}{0.4}{4mm}{A}{byz}

  \node[xscale=-1](label2) at (3.86,2.5) {\pgfuseimage{label-sideways}};
  \draw[line width=3mm] (red0) -- (red1);

  \membershipdisagreement{3}{2.3}{label2universe0}{0.4}{0.4}
  \membershipdisagreementfailure{3}{2.3}{label2universe0}{0.4}{0.4}{4mm}{E}{byz}
  \membershipdisagreementfailure{3}{2.3}{label2universe0}{0.4}{0.4}{4mm}{A}{byz}

  \draw[line width=0.5mm, black] (4.07, .5) -- (4.07, 3.5);

  \membershipdisagreement{4.35}{2.3}{label2universe1}{0.4}{0.4}
  \membershipdisagreementfailure{4.35}{2.3}{label2universe1}{0.4}{0.4}{4mm}{B}{byz}
  \membershipdisagreementfailure{4.35}{2.3}{label2universe1}{0.4}{0.4}{4mm}{A}{byz}

  \draw[line width=0.5mm, black] (2.83, 2) -- (5.27, 2);

  \membershipdisagreement{3}{.8}{label2universe2}{0.4}{0.4}
  \membershipdisagreementfailure{3}{.8}{label2universe2}{0.4}{0.4}{4mm}{C}{byz}
  \membershipdisagreementfailure{3}{.8}{label2universe2}{0.4}{0.4}{4mm}{A}{byz}

  \membershipdisagreement{4.35}{.8}{label2universe3}{0.4}{0.4}
  \membershipdisagreementfailure{4.35}{.8}{label2universe3}{0.4}{0.4}{4mm}{D}{byz}
  \membershipdisagreementfailure{4.35}{.8}{label2universe3}{0.4}{0.4}{4mm}{A}{byz}

\end{tikzpicture}
\caption[Membership Disagreement Learner Graph]{
Membership Disagreement Learner Graph:
Only edge labels are shown.
Learners are drawn as eyes, with
 darker blue learners on the left, and lighter, outlined red
 learners on the right.
Note that the edges are labeled with safe sets for which
 the pair of learners want to agree.
In each safe set, \byzantine failures are marked with a demonic
 face, and the unmarked acceptors are \textit{safe}.
All the edges except the rightmost and leftmost share the same label,
 in the middle.
}
\label{fig:membershipdisagreementclg}
\end{figure}

For the red learners, quorums are any 3 red or black acceptors,
 while for the blue learners, quorums are any 3 blue or black
 acceptors.
Quorums are shown in 
\cref{fig:membershipdisagreementquorums}.

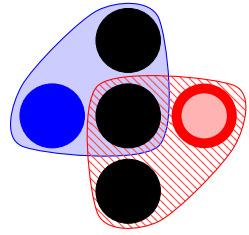
\begin{figure}
  \centering
\begin{tikzpicture}[scale=1]
\draw[blue,fill=blue!20] plot [smooth cycle] coordinates {
    (-0.4, 0.6)
    (-0.4, 1.4)
    (.6,   2.4)
    (1.4,  2.3)
    (1.4,  .6)
  };
\draw[red, pattern=north west lines, pattern color=red!70] plot [smooth cycle] coordinates {
    (2 + .4,    2 + -.6)
    (2 + .4,    2 + -1.4)
    (2 + -.6,   2 + -2.4)
    (2 + -1.4,  2 + -2.3)
    (2 + -1.4,  2 + -.6)
  };
\membershipdisagreement{0}{0}{Final}{1}{1}{10mm}
\end{tikzpicture}
\caption[Membership Disagreement Quorums]{
  Membership Disagreement Quorums:
  A quorum for the blue learners is shown in the light solid blue
   region, and a quorum for the red learners is shown as a striped
   red region.
}
\label{fig:membershipdisagreementquorums}
\end{figure}
Note that in order to tolerate 2 total failures, a homogeneous
 \byzantine-fault-tolerant consensus protocol would need at least 7
 acceptors.
Heterogeneity spares the expense of 2
 unnecessary additional acceptors.

\subsubsection{Failure Disagreement}
\label{sec:failuredisagreement}
Alternatively, learners might disagree about the types of failures
 they expect. 
Even if each learner expects homogeneous failures, learners'
 expectations may differ.

For example, consider a protocol with 5 acceptors. 
Two learners, called \textit{blue}, want termination and agreement so long as there
 isn't more than one failure, even if that failure is \byzantine.
Another two learners, called \textit{red}, want termination and agreement so long as
 there are no more than 2  crash failures, but accept that they may
 disagree if there is a \byzantine failure.
The red learners and the blue learners accept that they may disagree iff
 there is at least one \byzantine failure.
Thus we draw the learner graph~(\cref{sec:learnergraph}), in
\cref{fig:failuredisagreementclg}.

\begin{figure}
 \centering
\begin{tikzpicture}[scale=1]
  \pgfdeclareimage[width=10mm]{eye-blue}{figures/icons/eye-icon-custom-larger-darkblue.pdf}
  \pgfdeclareimage[width=10mm]{eye-red}{figures/icons/eye-icon-custom-outlinered.pdf}
  \pgfdeclareimage[width=8mm]{skull}{figures/icons/Skull-Crossbones.pdf}
  \pgfdeclareimage[width=10mm]{devil}{figures/icons/devil.pdf}
  \pgfdeclareimage[width=13mm]{label}{figures/icons/label.pdf}
  \pgfdeclareimage[width=16mm]{label-sideways}{figures/icons/label-sideways.pdf}

  \node[]() at (-2,3.5) {
    \begin{pgfrotateby}{\pgfdegree{270}}\pgfuseimage{label}\end{pgfrotateby}};
  \node[]() at (-1.5,1.5) {\pgfuseimage{devil}};

  \node[]() at (-2.3,5) {
    \begin{pgfrotateby}{\pgfdegree{270}}\pgfuseimage{label}\end{pgfrotateby}};
  \node[]() at (-2,3) {\pgfuseimage{skull}};

  \node[]() at (-2.3,2) {
    \begin{pgfrotateby}{\pgfdegree{270}}\pgfuseimage{label}\end{pgfrotateby}};
  \node[]() at (-2,0) {\pgfuseimage{skull}};

  \node[]() at (6.1,3.7) {
    \begin{pgfrotateby}{\pgfdegree{90}}\pgfuseimage{label}\end{pgfrotateby}};
  \node[]() at (4.15,3) {\pgfuseimage{skull}};
  \node[]() at (5.05,3) {\pgfuseimage{skull}};

  \node[]() at (6.1,0.7) {
    \begin{pgfrotateby}{\pgfdegree{90}}\pgfuseimage{label}\end{pgfrotateby}};
  \node[]() at (4.15,0) {\pgfuseimage{skull}};
  \node[]() at (5.05,0) {\pgfuseimage{skull}};

  \node[](blue0) at (0,0) {\pgfuseimage{eye-blue}};
  \node[](blue1) at (0,3) {\pgfuseimage{eye-blue}};
  \draw[line width=3mm] (blue0) -- (blue1);

  \node[](red0) at (2.5,0) {\pgfuseimage{eye-red}};
  \node[](red1) at (2.5,3) {\pgfuseimage{eye-red}};
  \draw[line width=3mm] (blue0) -- (red1);
  \draw[line width=3mm] (blue1) -- (red0);
  \draw[line width=3mm] (blue0) -- (red0);

  \draw[line width=3mm] (blue1) -- (red1);

  \draw[line width=3mm] (red0) -- (red1);
\end{tikzpicture}

\caption[Failure Disagreement Learner Graph]{
  Failure Disagreement Learner Graph:
  Learners are shown as eyes, with darker blue learners on the left,
   and lighter, outlined red learners on the right.
The edge between the blue learners is labeled with one \byzantine
 failure: they want to agree even if one acceptor is \byzantine.
All the other learners may disagree if there is a \byzantine failure.
The learners are each labeled with a number of crash failures: they
   want to terminate even if one or two acceptors crash. 
}
\label{fig:failuredisagreementclg}
\end{figure}
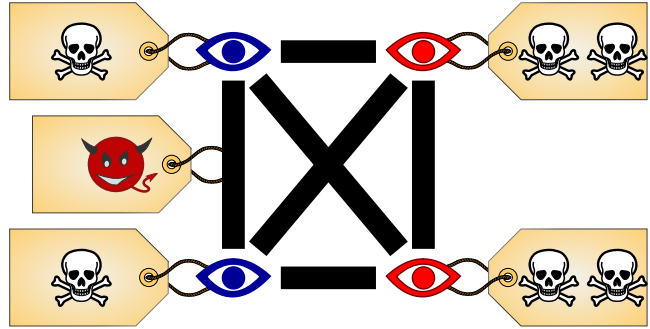

For the red learners, quorums are any 3 acceptors,
 while for the blue learners, quorums are any 4 acceptors.
The safe sets on the edge between the blue learners consist of any 4
 acceptors, and on all the other edges consist of all the acceptors.
Quorums are illustrated in~\cref{fig:failuredisagreement}.

 \begin{figure}
   \centering
\begin{tikzpicture}[scale=1]

\draw[blue,fill=blue!20] plot [smooth cycle] coordinates {
    (-.4,-1)
    (-1.3,-1)
    (-1.3,1.1)
    (.5,1.5)
    (1.45,.3)
    (1.45,-.2)
    (1,-.6)
    (0,0)
  };
\draw[red, pattern=north west lines, pattern color=red!70] plot [smooth cycle] coordinates {
    ( 0.8, 1.3)
    (-0.1, 1.2)
    (-0.1,-1.2)
    ( 0.8,-1.3)
    ( 1.6,0)
  };
\def \n {5}
\def \radius {1}
\foreach \a in {1,2,...,\n}{
  \draw(\a*360/\n: 10mm) node[circle, fill=black, inner sep=3mm, draw=black, line width=0mm]{\ };
}
\end{tikzpicture}
\caption[Failure Disagreement Example]{Failure Disagreement Example:
A quorum for the blue learners is shown in the light solid blue
 region, and a quorum for the red learners is shown as a striped red
 region.
}
\label{fig:failuredisagreement}
\end{figure}
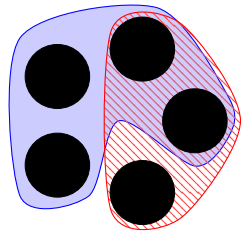

Note that in order to tolerate 2 total failures, a homogeneous
 \byzantine-fault-tolerant consensus protocol would need at least 7
 acceptors.
Heterogeneity spares the expense (in latency and resources) of 2
 unnecessary additional acceptors.

\subsection{Heterogeneous Learners and Failures}
\label{sec:exampleheterogeneouslearnersfailures}
Heterogeneous Paxos can express protocols wherein learners have
 different failure assumptions, and failures are heterogeneous, but
 acceptors are homogeneous.
Consider a protocol with 12 acceptors, and 4 learners.
Two learners, called \textit{blue}, want agreement and termination so long as there
 aren't aren't more than 3 \byzantine failures, and 1 additional \crash
 failure.
The edges between blue learners have safe sets consisting of any 9 acceptors.
Another two learners, called \textit{red}, want agreement and termination so long as
 there is no more than 1 \byzantine failure, and 4 additional \crash
 failures.
The edges between red learners have safe sets consisting of any 11 acceptors.
The red learners and the blue learners want to agree if there
 is no more than 1 \byzantine failure.
They accept that they may disagree otherwise.
The edges between red and blue learners likewise have safe sets consisting of any 11 acceptors.
Thus we draw the learner graph~(\cref{sec:learnergraph}) in 
      \cref{fig:hetobfailclg}.

 \begin{figure}
 \centering
\begin{tikzpicture}[scale=1]
  \pgfdeclareimage[width=10mm]{eye-blue}{figures/icons/eye-icon-custom-larger-darkblue.pdf}
  \pgfdeclareimage[width=10mm]{eye-red}{figures/icons/eye-icon-custom-outlinered.pdf}
  \pgfdeclareimage[width=5mm]{skull}{figures/icons/Skull-Crossbones.pdf}
  \pgfdeclareimage[width=6mm]{devil}{figures/icons/devil.pdf}
  \pgfdeclareimage[width=13mm]{label}{figures/icons/label.pdf}
  \pgfdeclareimage[width=16mm]{label-sideways}{figures/icons/label-sideways.pdf}

  \node[]() at (-2,3.5) {
    \begin{pgfrotateby}{\pgfdegree{270}}\pgfuseimage{label}\end{pgfrotateby}};
  \node[]() at (-1.5,1.2) {\pgfuseimage{devil}};
  \node[]() at (-2.1,1.2) {\pgfuseimage{devil}};
  \node[]() at (-1.8,1.8) {\pgfuseimage{devil}};

  \node[]() at (-2.3,5) {
    \begin{pgfrotateby}{\pgfdegree{270}}\pgfuseimage{label}\end{pgfrotateby}};
  \node[]() at (-2.5,3.3) {\pgfuseimage{skull}};
  \node[]() at (-2.5,2.7) {\pgfuseimage{skull}};
  \node[]() at (-1.5,3.3) {\pgfuseimage{skull}};
  \node[]() at (-1.5,2.7) {\pgfuseimage{skull}};

  \node[]() at (-2.3,2) {
    \begin{pgfrotateby}{\pgfdegree{270}}\pgfuseimage{label}\end{pgfrotateby}};
  \node[]() at (-2.5,0.3) {\pgfuseimage{skull}};
  \node[]() at (-2.5,-0.3) {\pgfuseimage{skull}};
  \node[]() at (-1.5,0.3) {\pgfuseimage{skull}};
  \node[]() at (-1.5,-0.3) {\pgfuseimage{skull}};

  \node[]() at (6.1,3.7) {
    \begin{pgfrotateby}{\pgfdegree{90}}\pgfuseimage{label}\end{pgfrotateby}};
  \node[]() at (4.1,3.3) {\pgfuseimage{skull}};
  \node[]() at (4.1,2.7) {\pgfuseimage{skull}};
  \node[]() at (5.1,3.3) {\pgfuseimage{skull}};
  \node[]() at (5.1,2.7) {\pgfuseimage{skull}};
  \node[]() at (4.6,3.0) {\pgfuseimage{skull}};

  \node[]() at (6.1,0.7) {
    \begin{pgfrotateby}{\pgfdegree{90}}\pgfuseimage{label}\end{pgfrotateby}};
  \node[]() at (4.1,0.3) {\pgfuseimage{skull}};
  \node[]() at (4.1,-0.3) {\pgfuseimage{skull}};
  \node[]() at (5.1,0.3) {\pgfuseimage{skull}};
  \node[]() at (5.1,-0.3) {\pgfuseimage{skull}};
  \node[]() at (4.6,0.0) {\pgfuseimage{skull}};

  \node[]() at (5.8,2.2) {
    \begin{pgfrotateby}{\pgfdegree{90}}\pgfuseimage{label}\end{pgfrotateby}};
  \node[]() at (4.3,1.5) {\pgfuseimage{devil}};

  \node[](blue0) at (0,0) {\pgfuseimage{eye-blue}};
  \node[](blue1) at (0,3) {\pgfuseimage{eye-blue}};
  \draw[line width=3mm] (blue0) -- (blue1);

  \node[](red0) at (2.5,0) {\pgfuseimage{eye-red}};
  \node[](red1) at (2.5,3) {\pgfuseimage{eye-red}};
  \draw[line width=3mm] (blue0) -- (red1);
  \draw[line width=3mm] (blue1) -- (red0);
  \draw[line width=3mm] (blue0) -- (red0);

  \node[xscale=-1](label1) at (1.25,1.8) {\pgfuseimage{label}};
  \node[](skull0) at (1.27,1.3) {\pgfuseimage{devil}};
  \draw[line width=3mm] (blue1) -- (red1);

  \draw[line width=3mm] (red0) -- (red1);
\end{tikzpicture}
      \caption[Heterogeneous Learners and Failures Learner Graph]{
Heterogeneous Learners and Failures Learner Graph:
Learners are drawn as eyes, with
 darker blue learners on the left, and lighter, outlined red
 learners on the right.
Note that the edges between each pair of learners are labeled with
 the \byzantine failures under which they require agreement.
All edges except the rightmost and leftmost share a central label.
Each learner is labeled with the crash failures under which they require termination.
}
      \label{fig:hetobfailclg}
    \end{figure}
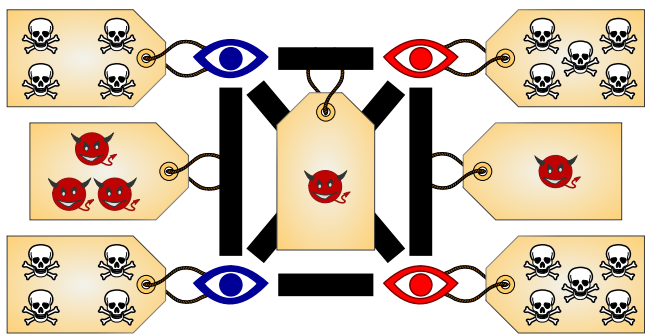

For the red learners, quorums are any 7 acceptors,
 while for the blue learners, quorums are any 8 acceptors.
Example quorums are shown in 
\cref{fig:hetlearnerfail}.

 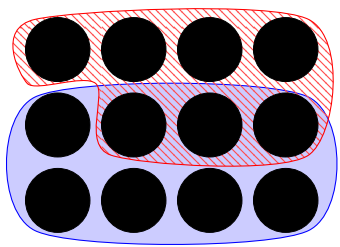
\begin{figure}
   \centering
\begin{tikzpicture}[scale=1]
\draw[blue,fill=blue!20] plot [smooth cycle] coordinates {
    (-1.3,-1.4)
    (-1.3,0.37)
    (2.3,0.37)
    (2.3,-1.4)
  };
\draw[red, pattern=north west lines, pattern color=red!70] plot [smooth cycle] coordinates {
    (-1.3, 1.4)
    ( 2.3, 1.4)
    ( 2.3,-0.4)
    (-0.3,-0.4)
    (-0.5, 0.55)
    (-1.4, 0.55)
  };
\node[circle, fill=black, inner sep=3mm, draw=black, line width=0mm](A) at (-1,-1) {\ };
\node[circle, fill=black, inner sep=3mm, draw=black, line width=0mm](B) at ( 0,-1) {\ };
\node[circle, fill=black, inner sep=3mm, draw=black, line width=0mm](C) at ( 1,-1) {\ };
\node[circle, fill=black, inner sep=3mm, draw=black, line width=0mm](D) at ( 2,-1) {\ };
\node[circle, fill=black, inner sep=3mm, draw=black, line width=0mm](E) at (-1,0) {\ };
\node[circle, fill=black, inner sep=3mm, draw=black, line width=0mm](F) at ( 0,0) {\ };
\node[circle, fill=black, inner sep=3mm, draw=black, line width=0mm](G) at ( 1,0) {\ };
\node[circle, fill=black, inner sep=3mm, draw=black, line width=0mm](H) at ( 2,0) {\ };
\node[circle, fill=black, inner sep=3mm, draw=black, line width=0mm](I) at (-1,1) {\ };
\node[circle, fill=black, inner sep=3mm, draw=black, line width=0mm](J) at ( 0,1) {\ };
\node[circle, fill=black, inner sep=3mm, draw=black, line width=0mm](K) at ( 1,1) {\ };
\node[circle, fill=black, inner sep=3mm, draw=black, line width=0mm](L) at ( 2,1) {\ };
\end{tikzpicture}
\caption[Heterogeneous Learners and Failures Example]{
  Heterogeneous Learners and Failures Example:
  A quorum for the blue learners is shown in the light solid blue
   region, and a quorum for the red learners is shown as a striped
   red region.
}
\label{fig:hetlearnerfail}
\end{figure}

Note that in order to tolerate 5 total failures, a homogeneous
 \byzantine-fault-tolerant consensus protocol would need at least 16
 acceptors.
Heterogeneity spares the expense (in latency and resources) of 4
 unnecessary additional acceptors.

To simultaneously tolerate all learners' worst fears, (3 \byzantine
 and 2 additional crash failures), we'd need 14 acceptors.
The additional detail of heterogeneous learners spares the
 expense (in latency and resources) of 2 unnecessary additional
 acceptors, as opposed to just heterogeneous failures.

\subsection{Heterogeneous Learners and Acceptors}
\label{sec:exampleheterogeneouslearnersacceptors}
Heterogeneous Paxos can express protocols wherein learners have
 different failure assumptions, and acceptors are heterogeneous,
 but failures are homogeneous.
Consider a protocol with 8 acceptors, and 4 learners.
The acceptors are divided into two groups of 4: \textit{blue} and
 \textit{red}.

All learners want agreement whenever at most 1 acceptor is
 \byzantine.
All edges therefore include at safe sets consisting of any 7 acceptors.
Two learners, called \textit{blue}, also want agreement with each
 other, as well as termination, whenever at most 1 blue and 2 red acceptors are \byzantine.
The edge between them therefore includes safe sets with any 3 blue and 2 red acceptors, and their quorums likewise consist of any 3 blue and 2 red acceptors.
Two learners, called \textit{red}, also want agreement with each
 other, as well as termination, whenever at most 1 red and 2 blue acceptors are \byzantine.
The edge between them therefore includes safe sets with any 2 blue and 3 red acceptors, and their quorums likewise consist of any 2 blue and 3 red acceptors.
Example quorums are illustrated in 
\cref{fig:hetlearnersacceptors}.

 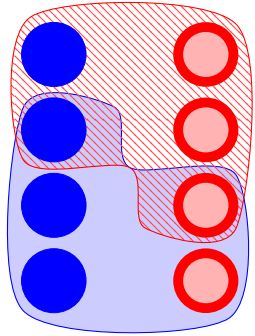
\begin{figure}
   \centering
\begin{tikzpicture}[scale=1]
\draw[blue,fill=blue!20] plot [smooth cycle] coordinates {
    (-0.4, -0.4)
    (-0.4,  2.3)
    (.8,    2.3)
    (1,     1.5)
    (2.4,   1.4)
    (2.3,  -0.5)
  };
\draw[red, pattern=north west lines, pattern color=red!70] plot [smooth cycle] coordinates {
    (2.4,   3.4)
    (2.4,    .7)
    (1.2,    .7)
    (1,     1.5)
    (-.4,   1.6)
    (-.3,   3.5)
  };
\node[circle, fill=blue, inner sep=3mm, draw=blue, line width=0mm](A) at (0,0) {\ };
\node[circle, fill=blue, inner sep=3mm, draw=blue, line width=0mm](B) at (0,1) {\ };
\node[circle, fill=blue, inner sep=3mm, draw=blue, line width=0mm](C) at (0,2) {\ };
\node[circle, fill=blue, inner sep=3mm, draw=blue, line width=0mm](D) at (0,3) {\ };
\node[circle, fill=red!30, inner sep=2.57mm, draw=red, line width=1.3mm](E) at (2,0) {\ };
\node[circle, fill=red!30, inner sep=2.57mm, draw=red, line width=1.3mm](F) at (2,1) {\ };
\node[circle, fill=red!30, inner sep=2.57mm, draw=red, line width=1.3mm](G) at (2,2) {\ };
\node[circle, fill=red!30, inner sep=2.57mm, draw=red, line width=1.3mm](H) at (2,3) {\ };
\end{tikzpicture}
\caption[Heterogeneous Learners and Acceptors Example]{Heterogeneous Learners and Acceptors Example:
Here, we draw one group as solid blue circles, and the other as hollow
 red circles.
A quorum for the blue learners is shown in the light solid blue
 region, and a quorum for the red learners is shown as a striped red
 region.
}
\label{fig:hetlearnersacceptors}
\end{figure}

To tolerate all the failures any learner believes possible,
 a consensus with homogeneous learners would need at least one more
 acceptor.
Taking heterogeneous learners into account spares the expense of that
 unnecessary acceptor

Since a single learner can tolerate 3 total failures, a protocol with
 homogeneous acceptors would require at least 10 acceptors. 
Heterogeneity spares the expense of 2 unnecessary additional
 acceptors.

\subsection{Heterogeneous Learners, Failures, and Acceptors}
\label{sec:exampleheterogeneouslearnersfailuresacceptors}
Heterogeneous Paxos can express protocols wherein learners have
 different failure assumptions, failures are heterogeneous,
 and so are acceptors.
Consider a protocol with 9 acceptors, and 4 learners.
The acceptors are divided into 3 groups of 3: \textit{blue},
 \textit{black}, and \textit{red}.

All learners want agreement when no \byzantine failures occur: all edges in the learner graph include safe sets with all 9 acceptors.
Two learners, called \textit{blue}, want agreement whenever at most 1 black acceptor, and all the red acceptors, are \byzantine: the edge between them includes safe sets with all the blue acceptors, and any 2 black acceptors.
Furthermore, blue learners want termination even when one blue acceptor has crashed, one black acceptor is \byzantine, and all red acceptors are \byzantine.
The blue learners' quorums therefore consist of any 2 blue acceptors, as well as any 2 black acceptors.

Another two learners, called \textit{red}, want agreement whenever at most 1 black acceptor, and all the blue acceptors, are \byzantine: the edge between them includes safe sets with all the red acceptors, and any 2 black acceptors.
Furthermore, red learners want termination even when one red acceptor has crashed, one black acceptor is \byzantine, and all blue acceptors are \byzantine.
The red learners' quorums therefore consist of any 2 red acceptors, as well as any 2 black acceptors.
Example quorums are illustrated in 
\cref{fig:hetx}.

 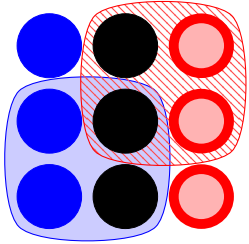
\begin{figure}
 \centering
\begin{tikzpicture}[scale=1]
\draw[blue,fill=blue!20] plot [smooth cycle] coordinates {
    (-.4, -.4)
    (-.4, 1.4)
    (1.4, 1.4)
    (1.4,  -.4)
  };
\draw[red, pattern=north west lines, pattern color=red!70] plot [smooth cycle] coordinates {
    (2.4, 2.4)
    ( .6, 2.4)
    ( .6,  .6)
    (2.4,  .6)
  };
\node[circle, fill=blue, inner sep=3mm, draw=blue, line width=0mm](A) at (0,0) {\ };
\node[circle, fill=blue, inner sep=3mm, draw=blue, line width=0mm](B) at (0,1) {\ };
\node[circle, fill=blue, inner sep=3mm, draw=blue, line width=0mm](C) at (0,2) {\ };
\node[circle, fill=black, inner sep=3mm, draw=black, line width=0mm](D) at (1,0) {\ };
\node[circle, fill=black, inner sep=3mm, draw=black, line width=0mm](E) at (1,1) {\ };
\node[circle, fill=black, inner sep=3mm, draw=black, line width=0mm](F) at (1,2) {\ };
\node[circle, fill=red!30, inner sep=2.57mm, draw=red, line width=1.3mm](G) at (2,0) {\ };
\node[circle, fill=red!30, inner sep=2.57mm, draw=red, line width=1.3mm](H) at (2,1) {\ };
\node[circle, fill=red!30, inner sep=2.57mm, draw=red, line width=1.3mm](I) at (2,2) {\ };
\end{tikzpicture}
\label{fig:hetx}
\caption[Heterogeneous Learners, Acceptors, and Failures Example]{
Heterogeneous Learners, Acceptors, and Failures Example:
Here, we draw one group as solid blue circles, and the other as hollow
 red circles.
A quorum for the blue learners is shown in the light solid blue
 region, and a quorum for the red learners is shown as a striped red
 region.}
\end{figure}

For a fully homogeneous consensus to tolerate 4 \byzantine failures
 would require 13 acceptors, so heterogeneity spares the cost of 4
 additional unnecessary acceptors.
\fi

\section{Implementation}
\label{sec:charlotte}
Since Heterogeneous Paxos is designed for cross-domain
 applications where different parties have different trust
 assumptions, it is well-suited for blockchains.
We constructed a variety of example blockchains using the Charlotte
 framework~\cite{CharlotteTR}, which allows for pluggable
 integrity (consensus) mechanisms.
\ifreport
In particular, we added a new ``proof of consensus'' subtype 
 featuring decision sets~(\cref{defn:decision}), for a given
 learner.
Each block on a chain features such a proof, to demonstrate that the
 learners for that chain have decided on that block at that height.

\subsection{Cross-Domain Commits}
\label{sec:meet}
The Charlotte framework~\cite{CharlotteTR} allows a single
 proof to atomically commit one block into two distinct data
 structures.
This is similar to atomically committing a transaction on two separate
 databases.
Such a proof must simultaneously satisfy the learners (who dictate
 the failure tolerance requirements) relevant to both data structures.
Heterogeneous Paxos provides an expressive language for learners
 to specify precise trust assumptions: the learner
 graph~(\cref{sec:learnergraph}).
If one group of learners care about one data structure, they
 presumably have edges among themselves dictating the conditions under
 which they agree on the makeup of that data structure.
If another group cares about another data structure, then when both
 groups want to commit a block to both, they need to specify a new
 learner graph, including edges between the two groups: the conditions
 under which they want the commit to be atomic.

To preserve maximum \textit{safety} (but not necessarily maximum
 liveness), the new learner graph will feature new quorums for each learner:
 each quorum necessary for a dual-data-structure protocol
 is the union of one quorum from each component type.
With this construction, we can atomically commit a single \block onto
 multiple Heterogeneous Paxos chains.

To preserve maximum \textit{liveness}, learners should simply not
 demand to agree between the two groups.
We can posit some hypothetical acceptor which everyone knows to be
 \byzantine, and include that acceptor in all safe sets in the edges
 between the two groups.
Each will commit blocks independently, and they will likely not agree
 (it's not very safe).
With this construction, Heterogeneous Paxos neatly describes
 independent consensus protocols as a single protocol.

\subsection{Charlotte Representation}
\label{sec:charlotteimpl}
We implemented a prototype of Heterogeneous Paxos as a ``Fern''
 Integrity service within the Charlotte
 framework~\cite{CharlotteTR}.
The proofs it produces are specific to each learner's assumptions.
We also use Charlotte's blocks as messages in the consensus protocol,
 taking advantage of its reference-by-hash format, as well as
 marshaling and message-passing functions.
\fi
Our servers are implemented in 1,704 lines of
 \ifreport\else open-source \fi Java\ifreport,\footnote{
    excluding import statements and comments}
    available at \url{\charlotteURL}\fi.
Charlotte uses 256-bit SHA3 hashes, P256 elliptic curve signatures,
 protobufs~\cite{protobufs} for marshaling\ifreport and unmarshaling\fi, and
 gRPC~\cite{grpc} for transmitting messages over TLS 1.3 channels.

\ifreport

Our implementation has a separate \textit{light client} that
 does not participate in consensus;
 it merely requests that a block be added to a chain, and an
 acceptor server acts as the proposer in the consensus protocol.
Including communication with the light client, the process has a
 minimum latency of 5 messages.

\subsection{Evaluation}
\label{sec:evaluation}
We ran two types of experiments: heterogeneous configuration
 experiments demonstrate the advantages of tailoring consensus to a
 heterogeneous environment, and contention experiments measure our
 implementation's ability to handle multiple simultaneous proposals.
In each experiment, we assigned quorums to preserve maximum
 safety~\cref{sec:meet}, and a single light client appended 2,000
 blocks.
To avoid any ``warm-up'' or ``cool-down'' effects, the measurements
 ignore the first and last 500 blocks.
Our servers were VMs with 1 physical core of an Intel E5-2690 2.9\,GHz
 CPU, and 8 GB of RAM each.
To simulate geodistribution, we added 100 ms of artificial latency to
 each network connection, so the theoretical optimum latency from the
 light client, through the consensus protocol, and back, is 500 ms. 

\subsubsection{Heterogeneous Configuration Experiments}
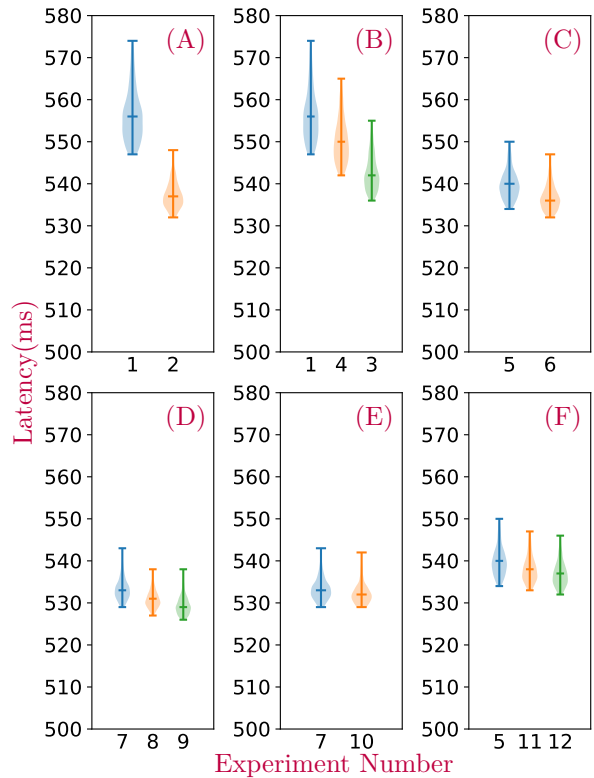
\begin{figure} \centering
\begin{tikzpicture}[scale=1]
\def\xscale{1.24} 
\def\yscale{5} 
\pgfdeclareimage[height=\yscale  cm]{clusterA}{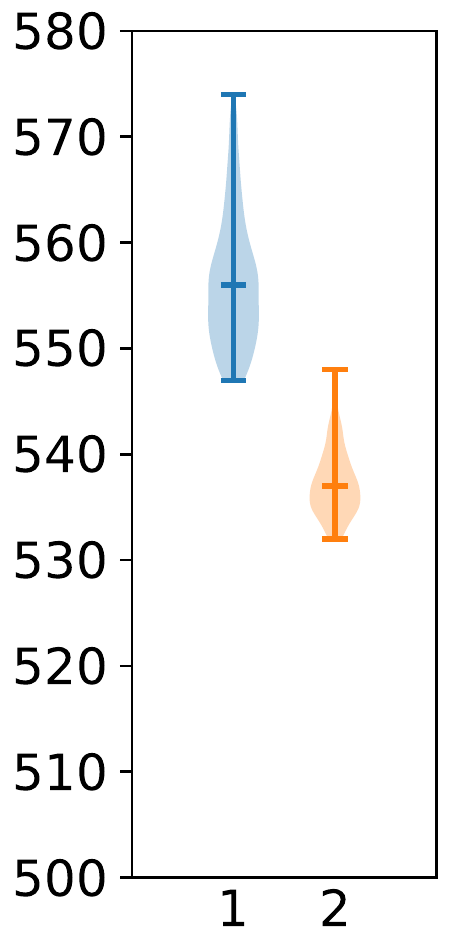}
\pgfdeclareimage[height=\yscale  cm]{clusterB}{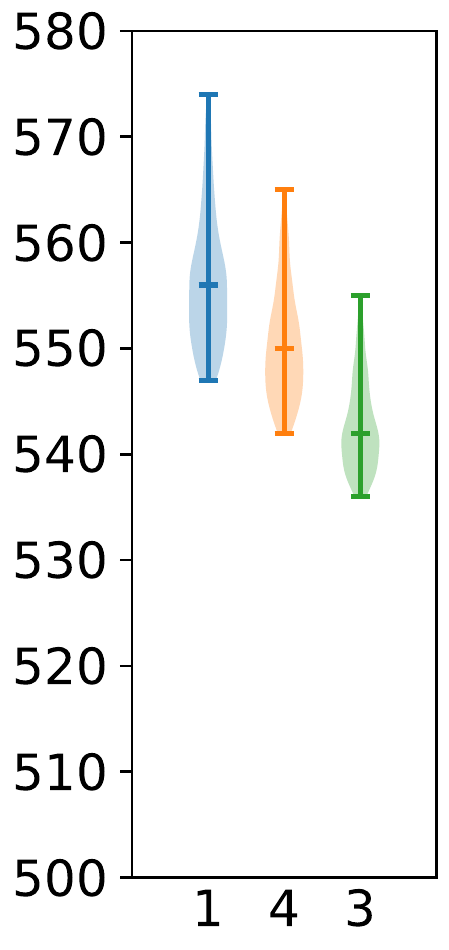}
\pgfdeclareimage[height=\yscale  cm]{clusterC}{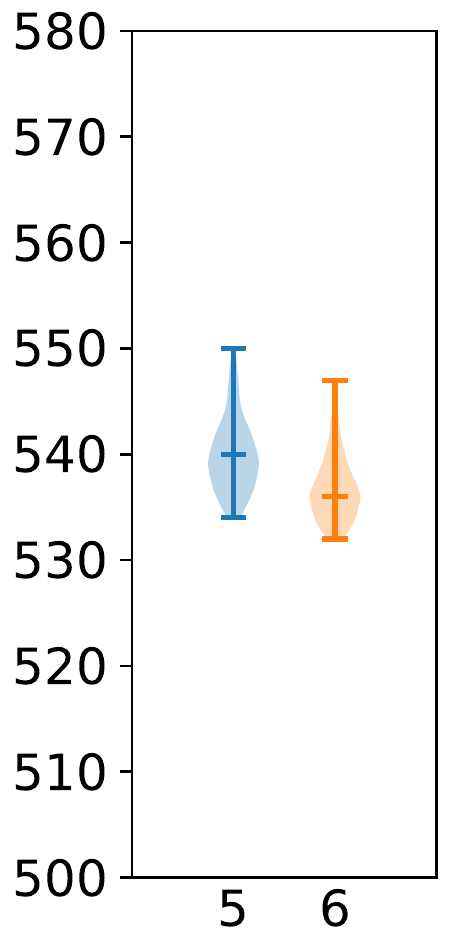}
\pgfdeclareimage[height=\yscale  cm]{clusterD}{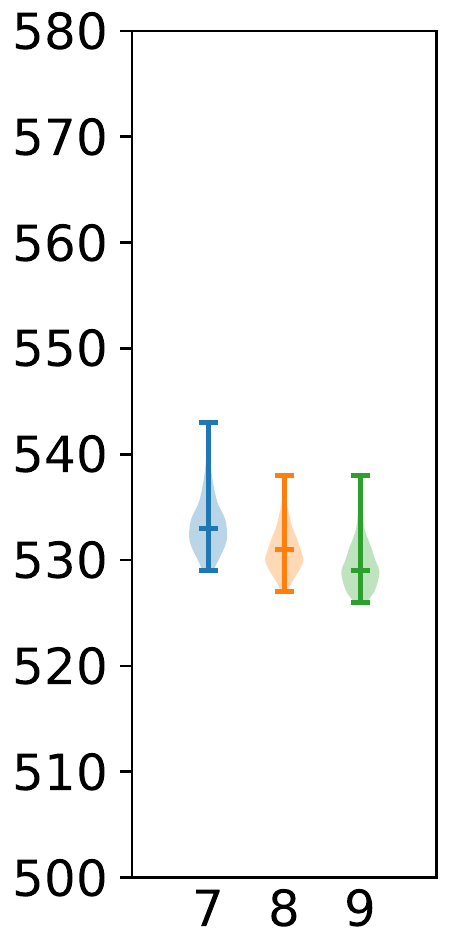}
\pgfdeclareimage[height=\yscale  cm]{clusterE}{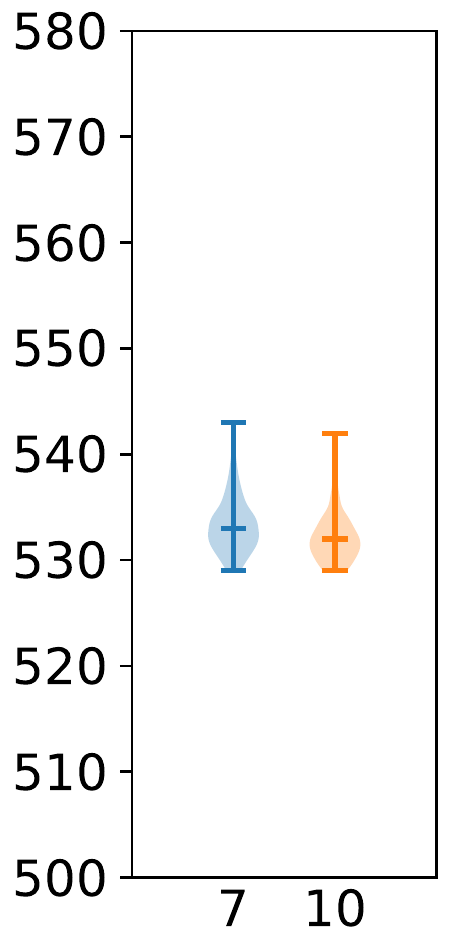}
\pgfdeclareimage[height=\yscale  cm]{clusterF}{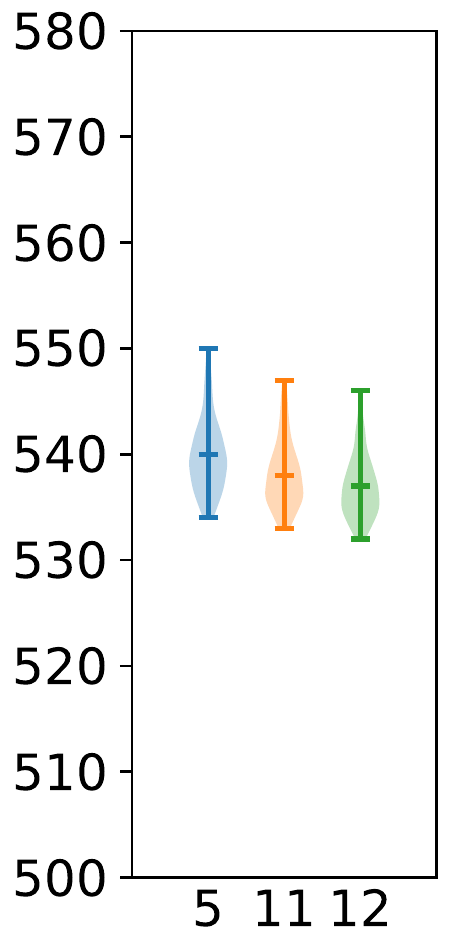}

\node[] at (-5*\xscale,\yscale){\pgfuseimage{clusterA}};
\node[] at (-3*\xscale,\yscale){\pgfuseimage{clusterB}};
\node[] at (-1*\xscale,\yscale){\pgfuseimage{clusterC}};
\node[] at ( -5*\xscale,0){\pgfuseimage{clusterD}};
\node[] at ( -3*\xscale,0){\pgfuseimage{clusterE}};
\node[] at ( -1*\xscale,0){\pgfuseimage{clusterF}};

\def\labelwidthoffset{.15*\yscale}
\def\labelheight{.4*\yscale}
\node[] at (-5*\xscale + \labelwidthoffset, \yscale+\labelheight){\purple{(A)}};
\node[] at (-3*\xscale + \labelwidthoffset, \yscale+\labelheight){\purple{(B)}};
\node[] at (-1*\xscale + \labelwidthoffset, \yscale+\labelheight){\purple{(C)}};
\node[] at ( -5*\xscale + \labelwidthoffset, \labelheight){\purple{(D)}};
\node[] at ( -3*\xscale + \labelwidthoffset, \labelheight){\purple{(E)}};
\node[] at ( -1*\xscale + \labelwidthoffset, \labelheight){\purple{(F)}};

\node[] at (-6.1*\xscale, 2.5){\rotatebox[origin=l]{90}{
                              \purple{Latency(ms)}}};
\node[] at (-3.5, -.52*\yscale){\purple{Experiment Number}};
\end{tikzpicture}
\\\vspace{-2mm}
\caption{
Round-trip latency in various heterogeneous configurations.
Bars represent the 5th percentile, median, and 95th percentile for
 each experiment, within the distribution shown as a violin plot.
Experiments are clustered according to their ability to tolerate
 similar failures with varying degrees of heterogeneity (some are
 shown in multiple clusters).
Optimal latency is 500 ms.
}
\label{fig:hetlatency}
\end{figure} 
\begin{table}
\centering
\begin{tabular}{| r | r | r | r | l | l | l | l | l |  l |}
  \hline
  \textbf{\#}
  & \multicolumn{3}{c|}{\textbf{acceptors}}
  & \multicolumn{3}{c|}{\textbf{agreement}}
  & \multicolumn{2}{c|}{\textbf{termination}}
  & \textbf{\S}
\\\hline
  & \textbf{\blue{b}}
  & \textbf{\red{r}}
  & \textbf{other}
  & \textbf{all}
  & \textbf{\blue{blue}}
  & \textbf{\red{red}}
  & \textbf{\blue{blue}}
  & \textbf{\red{red}}
  &
\\\hline
  1 & \multicolumn{3}{r|}{16} & \multicolumn{3}{l|}{5} & \multicolumn{2}{l|}{5} &
\\\hline
  2 & \blue 3 & \red 3            & 3                  &
      0       & 1t $ \land $ \red{3r} & \blue{3b} $ \land $ 1t &
      \blue{1b}$\land$1t$\land$\red{3r} & \blue{3b}$\land$1t$\land$\red{1r}
      &\ref{sec:exampleheterogeneouslearnersfailuresacceptors}
\\\hline
  3 & \multicolumn{3}{r|}{12} & 1 & 3 & 1 &  4 & 5
      &\ref{sec:exampleheterogeneouslearnersfailures}
\\\hline
  4 & \multicolumn{3}{r|}{14} & \multicolumn{3}{l|}{3} & \multicolumn{2}{l|}{5} &
\\\hline
  5 & \multicolumn{3}{r|}{10} & \multicolumn{3}{l|}{3} & \multicolumn{2}{l|}{3} &
\\\hline
  6 & \blue 4 & \red 4 & 0 &
    \red{1r}$\lor$\blue{1b} & \red{2r}$\land$\blue{1b} & \red{1r}$\land$\blue{2b} & 
    \red{2r} $ \land $ \blue{1b} & \red{1r} $ \land $ \blue{2b} 
    &\ref{sec:exampleheterogeneouslearnersfailuresacceptors}
\\\hline
  7 & \multicolumn{3}{r|}{7} & \multicolumn{3}{l|}{2} & \multicolumn{2}{l|}{2} &
\\\hline
  8 & \multicolumn{3}{r|}{6} & \multicolumn{3}{l|}{1} & \multicolumn{2}{l|}{2} 
    &\ref{sec:exampleheterogeneousfailures}
\\\hline
  9 & \multicolumn{3}{r|}{5} & 0 & 1 & 0 & 1 & 2 
    &\ref{sec:failuredisagreement}\\\hline
  10&\blue 3 & \red 3 & 0 & \multicolumn{3}{l|}{\blue{2b} $ \lor $ \red{2r}} 
    & \multicolumn{2}{l|}{\blue{2b} $ \lor $ \red{2r}} 
    &\ref{sec:exampleheterogeneousacceptors}
\\\hline
  11& \multicolumn{3}{r|}{9} & \multicolumn{3}{l|}{2} & \multicolumn{2}{l|}{3} & 
\\\hline
  12&\blue 4 & \red 4 & 0 & \multicolumn{3}{l|}{\blue{2b} $ \lor $ \red{2r}} 
    & \multicolumn{2}{l|}{(\blue{2b} $ \land $ \red{1r}) $ \lor $ (\blue{1b} $ \land $ \red{2r})} 
    &\ref{sec:exampleheterogeneousfailuresacceptors}
\\\hline
\end{tabular}
\caption{
Each experiment is numbered (\#).
For each, we divide acceptors into groups of \red{red} (\red r),
\blue{blue} (\blue b), and other (t), as necessary.
Each experiment has multiple \red{red} and  \blue{blue} learners.
For learners of the same or different colors, the agreement column
 lists the number of \byzantine failures of each acceptor group for
 which the learners require agreement.
For each color of learner, the termination column lists the number of
 crash failures of each acceptor group for which the learners require
 termination.
}
\label{table:experiments}
\end{table}
 
To show the advantages of tailoring consensus for
 heterogeneous environments, we ran several experiments, each with a
 single blockchain and a different \clg~(\cref{sec:learnergraph}).
In general, heterogeneous configurations save resources and latency
 compared with homogeneous configurations tolerating the same
 failures.
\Cref{table:experiments} describes the acceptors and learner
 graph for each experiment.
\Cref{fig:hetlatency} shows the round-trip latency distributions for
 committing blocks in each.

Experiment 2 uses the setup from~\cref{sec:introexample}.
Since the learners can tolerate 5 failures, some of them
 \byzantine, a fully homogeneous system would need 16
 acceptors, becoming experiment 1. 
Homogeneity costs 7 additional unnecessary acceptors, and
 as~\cref{fig:hetlatency}.A shows, 51\% median latency overhead
 (over the 500 ms optimal latency).

Similarly, experiment 3 uses heterogeneous learners and failures.
To tolerate all learners' worst fears, two more
 acceptors would be necessary, as represented in experiment 4.
Heterogeneous learners save the unnecessary expense of 2
 acceptors, and, as~\cref{fig:hetlatency}.B shows, 14\% median
 latency overhead.
Since learners can tolerate 5 failures, some of them
 \byzantine, a fully homogeneous system would need 16
 acceptors, becoming experiment 1. 
Heterogeneous failures alone save the unnecessary expense of a
 further 2 acceptors, and, as~\cref{fig:hetlatency}.B shows, 12\%
 median latency overhead.
For experiment 3, then, heterogeneity saves the unnecessary expense of
 4 acceptors, and, as~\cref{fig:hetlatency}.B shows, 33\% median
 latency overhead.

Without the ability to distinguish heterogeneous acceptors,
 experiment 6 becomes experiment 5: learners need to tolerate 3
 \byzantine failures.
Heterogeneous acceptors save the unnecessary expense of 2 
 acceptors, and, as~\cref{fig:hetlatency}.C shows, 11\%
 median latency overhead.

Since learners in experiment 8 can tolerate 2 failures, 1 of them
 \byzantine, a homogeneous system would need 7 acceptors,
 becoming experiment 7.
Heterogeneity saves the unnecessary expense of an additional
 acceptor, and, as~\cref{fig:hetlatency}.D shows, 6\%
 median latency overhead.
Without accounting for heterogeneous learners, experiment 9
 would need to tolerate 1 \byzantine and 1 additional crash
 failure, becoming experiment 8. 
Heterogeneous learners save the expense of an unnecessary
 acceptor and,  as~\cref{fig:hetlatency}.D shows, 7\%
 median latency overhead.
A fully homogeneous consensus would need to tolerate 2
 \byzantine failures, becoming experiment 7.
Heterogeneity saves the unnecessary expense of 2 additional
 acceptors, and, as~\cref{fig:hetlatency}.D shows, 14\%
 median latency overhead.

Without the ability to distinguish heterogeneous acceptors,
 learners in experiment 10 would reduce to a homogeneous consensus
 tolerating 2 \byzantine failures.
They would need 7 acceptors, becoming experiment 7.
Heterogeneity saves the unnecessary expense of an additional
 acceptor, and, as~\cref{fig:hetlatency}.E shows, 3\%
 median latency overhead.

Without heterogeneous failures, experiment 11 becomes a homogeneous
 system tolerating 3 \byzantine failures.
Learners would need 10 acceptors, becoming experiment 5.
Heterogeneity saves the expense of an additional acceptor
 and, as~\cref{fig:hetlatency}.F shows, 5\% median latency
 overhead.
Without accounting for heterogeneous acceptors, experiment 12
 would need to tolerate 2 \byzantine and 1 crash
 failure, becoming experiment 11.
Heterogeneous acceptors save the expense of 1 unnecessary
 acceptor, and, as~\cref{fig:hetlatency}.F shows, 3\%
 median latency overhead. 
A fully homogeneous system tolerating 3 \byzantine failures
 would need 10 acceptors, becoming experiment 5.
Heterogeneity saves the expense of 2 unnecessary
 acceptors, and, as~\cref{fig:hetlatency}.F shows, 8\%
 median latency overhead.

\subsubsection{Multichain shared \blocks}
\label{sed:multichain}
Shared (joint) \blocks facilitate cross-domain interaction.
In these experiments, the client appends blocks to
 2--4 chains, simultaneously, so a single execution of consensus
 decides atomically whether the block is appended to all chains, or
 none of them.
Each chain uses with 4 or 7 ``Fern'' Consensus acceptors,
 and tolerates 1 or 2 \byzantine failures.
These cross-chain commits preserve the \textit{safety} of both
 chains~(\cref{sec:meet}): quorums for the cross-chain commits each
 include a quorum from each individual chain.
As the yellow lines in \cref{fig:ManyMulti} show,
 latency scales roughly linearly with the number of chains.
Each acceptor's computational overhead
 (marshaling messages and verifying signatures) is linear in the
 number of chains.
The darker green lines in  \cref{fig:ManyMulti} serve as a control:
 running 1--4 independent chains in parallel (on separate VMs) does
 not affect individual chains' latency.

\paragraph{Single Chain}
In these experiments, a client appends 2000 successive \blocks to one
 chain.
Mean latency is 527 ms for a chain with 4 Fern servers and 538$ms$
 for 7 Fern servers.
Since the best possible latency is 500 ms, these results are
 promising.
Overheads include cryptographic signatures, verification, and garbage
 collection.

\begin{figure}

\centering
\begin{tabular}{ccc}

\begin{minipage}{.3\textwidth}

\begin{tabular}{rc}
  \rotatebox[origin=l]{90}{Latency (ms)}&
\includegraphics[width=.8\textwidth]{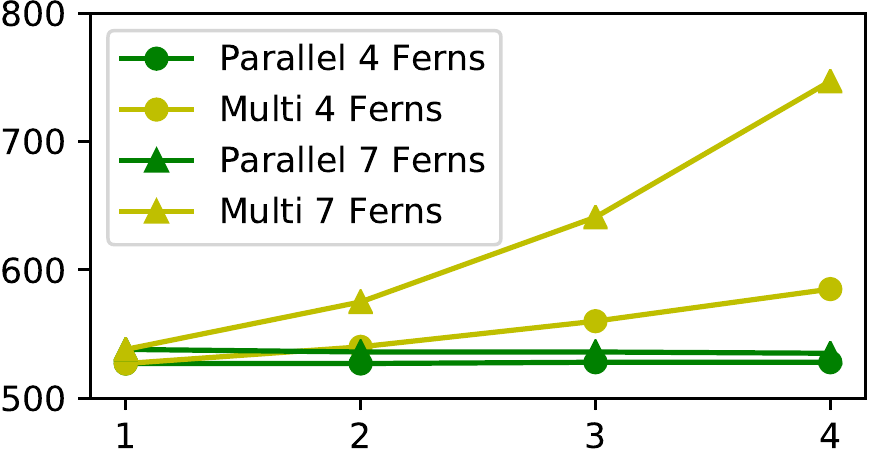}
\\ & Number of Chains
\end{tabular}

\end{minipage} & \begin{minipage}{.3\textwidth}

\begin{tabular}{rc}
\rotatebox[origin=l]{90}{blocks/sec}&
\includegraphics[width=.8\textwidth]{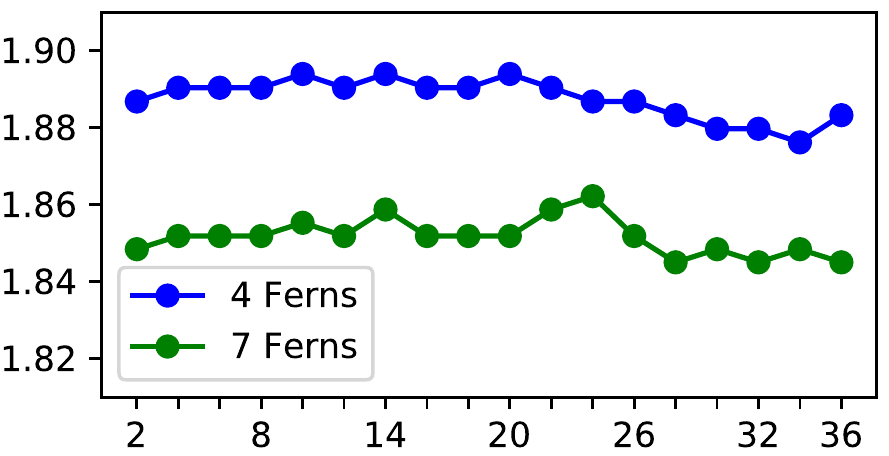}
\\ & Number of Clients
\end{tabular}

\end{minipage} & \begin{minipage}{.3\textwidth}

\begin{tabular}{rc}
\rotatebox[origin=l]{90}{blocks/sec}&
\includegraphics[width=.8\textwidth]{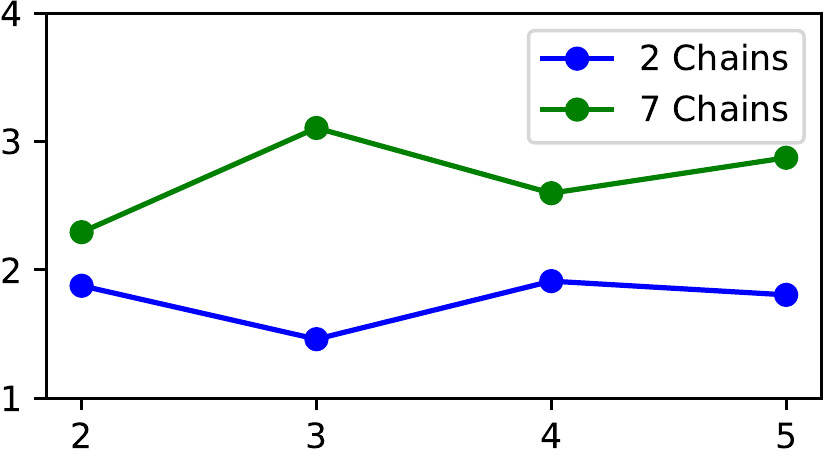}
\\ & Number of Clients
\end{tabular}

\end{minipage} \\

\begin{minipage}[t]{.3\textwidth}

\caption[Heterogeneous Paxos Multichain and Parallel experiments.]{
Heterogeneous Paxos Multichain and Parallel experiments.
In Parallel experiments, each chain operates independently
 (and has its own client).
In Multichain experiments, one client tries to append all blocks to
 all chains.
Optimal latency is 500 ms.
}
\label{fig:ManyMulti}

\end{minipage} & \begin{minipage}[t]{.3\textwidth}

\caption[Throughput of Heterogeneous Paxos under contention.]{
Throughput of Heterogeneous Paxos under contention.
  2--36 clients try to
          append 2000 blocks to just one chain. Optimal throughput is 2
          \blocks/sec.}
          \label{fig:ContentionThroughput}

\end{minipage} & \begin{minipage}[t]{.3\textwidth}

\caption{Throughput of Heterogeneous Paxos mixed-workload experiment (4 Fern
servers).}
\label{fig:MixedThroughput} 

\end{minipage}

\end{tabular}

\end{figure}

\begin{figure}
\centering
  \begin{tabular}{rl}
    \begin{tabular}{rc}
        \rotatebox[origin=l]{90}{\hspace{0.85cm}blocks/sec}
      &
        \includegraphics[width=.4\textwidth]{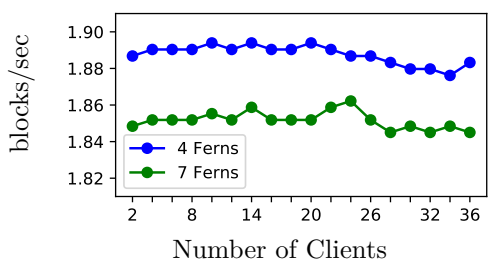}
      \\
      &
        Number of Clients
    \end{tabular}
  &
    \begin{minipage}[c]{.4\textwidth}
      \caption{Throughput of Heterogeneous Paxos under contention.
               2--36 clients try to append 2000 blocks to a single
                chain with either 4 or 7 acceptors (Fern servers),
				tolerating 1 or 2  \byzantine failures.
               Optimal throughput is 2 \blocks/sec.}
      \label{fig:ContentionThroughput}
    \end{minipage}
  \end{tabular}
\end{figure}

\subsubsection{Contention}
In these experiments, all clients simultaneously contend to append
 2000 unique \blocks to the same chain.
We measured the blocks that were actually accepted into slots
 500--1500 of the chain.
We used 2--36 clients, and chains with 4 or 7 servers,
configured with 2 GB RAM.
Like Byzantized Paxos~\cite{byzantizing-paxos},
 Heterogeneous Paxos can require a dynamic timeout to
 automatically trigger a new ballot.
Chain throughput is shown in
\cref{fig:ContentionThroughput}. 
Our chains, on average, achieved 1.88 \blocks/sec throughput for 4
 servers and 1.85 \blocks/sec for 7, not far from the
 2 \blocks/sec optimum.
Throughput does not decrease much with the number of clients.

	
\subsubsection{Mixed} 
These experiments attempt to simulate a more realistic scenario with
 multiple clients appending to multiple chains.
 including the previous 3 types of workload.
2--5 clients contend to append blocks onto either 2 or 7 chains, each
 with 4 acceptors.
On each block, a client tries to append a shared block to two random
chains with probability 10\% and otherwise tries appending to a
random single chain.
The results are in \cref{fig:MixedThroughput}.
Throughput can be over 2.0 blocks/sec because multiple clients can
 append blocks to different chains in parallel.
Mean throughput is 1.8 blocks/sec and 2.7 blocks/sec for
 2 and 7 chains respectively, which
is expected because the 2-chain configuration has more contention.

Heterogeneous Paxos scales horizontally with multiple chains running in
 parallel.
Furthermore, throughput does not decrease much with contention.

This gives us ability to make progress even with many clients
concurrently connecting to the same chain.
We also notice that the number of servers plays a major role in
 latency.
With a small group of servers, Heterogeneous Paxos almost reaches
 the theoretical optimum of 500 ms.
Since our Heterogeneous Paxos implementation is just a prototype,
 we believe that with further efforts in optimization, average latency
 performance can be improved.

\section{Future Work}
\label{sec:future}
We have generalized what it means to have
 Validity, Agreement and Termination
 in a heterogeneous setting, and designed and implemented a
 consensus protocol that meets these properties with minimum
 theoretical latency.
However, Heterogeneous Paxos is far from perfect.
For example, we have not discussed reconfiguration or changing the
 learner graph.
Here are a few more ways it can almost certainly improved, without
 altering its core concepts.

\subsection{Additional Failure Types}
Although the labels in our learner graphs have so far only taken
 crash and \byzantine failures into account, there are many other
 possible behaviors we can express.
For instance, detectable failures could be very useful to add to the
 model: rather than reconfiguring when acceptors are known to have
 failed, detection could be integrated directly.

\subsection{Network Assumption and Termination}
While Paxos and PBFT are partially synchronous consensus
 protocols~\cite{paxos,pbft}, recent advances in cryptography have
 made fully asynchronous, probabilistically terminating consensus
 protocols viable~\cite{Miller2016,Abraham2018}.
These have the advantage of never needing to insert artificial delays
 in order to guarantee termination.
It is possible that Heterogeneous Paxos could be adapted to this
 setting, using a similar ``shared coin flip'' mechanism, however it
 may require more interesting cryptography to deal with non-uniform
 quorums.

Even in the partially synchronous~\cite{Dwork1988} setting, there are
 almost certainly
 better timing strategies than those we
 present~(\cref{sec:semisynchronous}).
Optimizations like the leader/view change used in pbft~\cite{pbft} or
 the batching used in BFT-SMart~\cite{bessani2014state} could greatly
 improve performance.

\subsection{Bandwidth}
If each message bounces off of each acceptor, and each message is
 sent to each acceptor and each learner, and each phase features each
 acceptor sending a message, then the communication overhead for a
 setup with $n$ acceptors and $m$ learners is $O\p{n^3 + n^2 m}$
 messages.
While the reference-by-hash architecture allows each message to avoid
 copying the proposed value, that can still be a lot of overhead.
A fairly naive optimization would be to allow acceptors not to
 bounce all messages they receive, and instead request them only if
 the original sender didn't broadcast the message fast enough.
This could reduce communication overhead to $O\p{n^2 + nm}$, in the
 absence of \byzantine failures.

Using advanced cryptographic techniques, protocols like
 Hot-Stuff~\cite{Abraham2018} achieve consensus in $O\p n$ bandwidth
 overhead.
It may be possible to adapt Heterogeneous Paxos with some of these
 optimizations.

\subsection{Programming}
Implementations like BFT-SMart~\cite{bessani2014state}  have been
 heavily optimized, with researchers discovering better data
 structures and communication protocols along the way.
Our implementation is far from optimal.
Better memory layout, data structures, and parallelism are certainly
 possible.
\else
To explore the performance of Heterogeneous Paxos, we created
 several blockchains with different
 {\clg}s~(\cref{sec:learnergraph}).
The results~(\cref{sec:evaluation})
 show that heterogeneous configurations save resources and latency
 compared with homogeneous configurations tolerating the same
 failures.
For instance, in our example configuration~\cref{sec:introexample},
 a Homogeneous configuration tolerating similar failures would cost
 an extra 7 unnecessary acceptors, increasing latency overhead by 51\%
 relative to Heterogeneous Paxos.
%
\textit{2a} messages include a quorum of 256-bit message hashes,
 so they expand linearly with quorum size, as does the cost of
 unmarshaling and verifying the signatures of the messages referenced.
In all experiments, however, computational overhead was dominated by
 the theoretical minimum (simulated) geodistributed network latency. 
\fi

\section{Related Work}
\label{sec:related}

\ifreport\subsection{Heterogeneous Acceptors and Failures}
\else\textbf{Heterogeneous Acceptors and Failures:}
\fi
Heterogeneous Paxos is based on Leslie Lamport's
 \byzantine-fault-tolerant variant~\cite{byzantizing-paxos} of
 Paxos~\cite{paxos}.
Byzantine Paxos supports heterogeneous acceptors because it uses
 quorums: not all acceptors need be of equal worth, but
 all quorums are.
Although Lamport does not describe it explicitly,
 Byzantine Paxos can have heterogeneous, or
 \textit{mixed}~\cite{Siu1998},
 failures, so long as quorum
 intersections have a safe acceptor and at least one quorum
 is safe and live.

Many papers have investigated hybrid failure
 models~\cite{Siu1998,clement2009upright,Cachin2013,Liu2015XFTPF}
 in which different consensus protocol acceptors can
 have different failure modes, including crash failures and \byzantine
 failures (heterogeneous failures).
These papers typically investigate how many failures in each class can
 be tolerated.
Other papers have looked at system models in which different
 acceptors may be more or less likely to
 fail~\cite{Gifford79,Malkhi97a},
 or where failures are dependent (heterogeneous
 acceptors)~\cite{survivor-sets,gallet2011,guerraoui2007}.

Further generalizations are possible.
Our Learner Graph uses only \emph{safe} and \emph{live} acceptors, but
 its labels might be generalized to support other failure types such as
 rational failures~\cite{bar-sosp}.
We have only considered learners that all make
 the same (weak) synchrony assumption, but others have studied
 learners with heterogeneous network
 assumptions~\cite{Blum2019,Malkhi2019}.

\ifreport\subsection{Heterogeneous Learners}
\else\textbf{Heterogeneous Learners:}
\fi
Unlike ours, most related work conflates learners and
 acceptors.
Early related work on ``Consensus with Unknown
 Participants''~\cite{Cavin2004,Greve2007,Alchieri2008}
 defines protocols in which each participant 
 knows only a subset of other participants, inducing a
``who-knows-whom'' digraph; this work identifies
 properties of this graph that must hold to achieve
 consensus.
Not every participant knows all participants, but trust
 assumptions are homogeneous: participants have the
 same beliefs about trustworthiness of other participants.

\ifblinded Sheff et al.\@ describe\else
Our prior work describes\fi~\cite{hetconstechreport}
 a heterogeneous failure model in which different
 participants may have different failure assumptions about other
 participants.
\ifblinded They distinguish \else We distinguished \fi 
 learners whose failure assumptions are accurate from
 those whose failure assumptions are inaccurate and
 \ifblinded they specify \else we specified \fi
 a heterogeneous consensus protocol in terms of the possibly
 different conditions under which each learner is guaranteed
 agreement.
The paper constructs a heterogeneous consensus protocol that meets the
 requirements of all learners using lattice-based information
 flow to analyze and prove protocol properties.

Heterogeneous learners became of interest to blockchain
 implementations based on voting protocols where open membership was
 desirable.
Ripple (XRP)~\cite{Schwartz2014} was the earliest blockchain to
 attempt support for heterogeneous learners.
Originally, each learner had its own
 Unique Node List (UNL), the set of acceptors that it partially
 trusts and uses for making decisions.
An acceptor in more UNLs is implicitly more influential.
The protocol was updated because of correctness
 issues~\cite{Chase2018}, and support for diverse UNLs was all but
 eliminated.
Ripple has proposed a protocol called Cobalt~\cite{cobalt},
 in which each learner specifies a set of acceptors they
 partially trust, and it works if those sets intersect ``enough.''
Cobalt does not account for heterogeneous failures, and only limited
 acceptor heterogeneity.

The Stellar Consensus~\cite{mazieresstellar,stellarSOSP,StellarDISC}
 blockchain protocol supports both heterogeneous learners
 and acceptors, although it does not distinguish the two;
each learner specifies a set of ``quorum slices.''
Like Cobalt, Stellar does not account for heterogeneous failures.
Neither Stellar nor Cobalt match Heterogeneous Paxos' best-case
 latency.
Heterogeneous Paxos inherits Byzantine Paxos' 3-message-send best
 case latency, which is optimal for a consensus tolerating
 $\ceil{\frac{n}{3}} - 1$ failures in the homogeneous \byzantine case
 or $\ceil{\frac{n}{2}} - 1$ failures in the homogeneous \crash
 case~\cite{Bracha1983}.
However, both Cobalt and Stellar are designed for an ``open-world'' model, where not all
 acceptors and learners are known in advance. 
We have not yet adapted Heterogeneous Paxos to an open-world setting.

The heterogeneous learner models of Cobalt and Stellar have been
 studied in detail by Garc{\'i}a-P\'erez and
 Gotsman~\cite{GarcaPrez2018FederatedBQ}.
\ifreport
They explore what happens when \byzantine nodes lie about their trust
 choices and strengthen earlier results.
\fi
Cachin and Tackmann examine Stellar-style asymmetric trust models,
 including in shared-memory environments~\cite{Cachin2019}.
However, neither paper separates learners from acceptors,
 attempts to solve consensus, or considers heterogeneous failures;
the Learner Graph is more general.

Like our work, \textit{Flexible BFT}~\cite{Malkhi2019}
 distinguishes learners from acceptors and accounts for both
 heterogeneous learners and heterogeneous failures.
It does not allow heterogeneous acceptors: they
 are interchangeable, and quorums are specified by size.
Flexible BFT also has optimal best-case latency.
It does not support crash failures, but introduces a new
 failure type called \emph{alive-but-corrupt} for acceptors
 interested in violating safety but not liveness.

\section{Conclusion}
\label{sec:conclusion}
Heterogeneous Paxos is the first consensus protocol with
 heterogeneous acceptor, failures, and learners. 
It is based on the Learner Graph, a new and expressive way to
 capture learners' diverse failure-tolerance assumptions.
Heterogeneous consensus facilitates a more nuanced approach that can
 save time and resources, or even make previously unachievable
 consensus possible.
Heterogeneous Paxos is proven correct against our new generalization
 of consensus for heterogeneous settings.
This approach is well-suited to systems spanning heterogeneous
 trust domains; for example,
we demonstrate working \ifreport, composable \fi
 blockchains with heterogeneous trust.

Future work may expand learner graphs to represent even more types
 of failures.
Heterogeneous Paxos may be extended to allow for changing
 configurations, or improved efficiency in terms of bandwidth and
 computational overhead.
New protocols can also make use of our definition of heterogeneous
 consensus, perhaps adding new guarantees such as probabilistic
 termination in asynchronous networks.

\ifacknowledgments
\acknowledgements{}
\fi

\pagelimit{15}
\finalpage
\ifreport
\bibliography{bibtex/pm-master}
\else
\bibliography{abbreviated}

\begin{thebibliography}{1}

\bibitem{DBLP:journals/cacm/Dijkstra68a}
Edsger~W. Dijkstra.
\newblock Letters to the editor: go to statement considered harmful.
\newblock {\em Commun. {ACM}}, 11(3):147--148, 1968.
\newblock \href {https://doi.org/10.1145/362929.362947}
  {\path{doi:10.1145/362929.362947}}.

\bibitem{DBLP:books/mk/GrayR93}
Jim Gray and Andreas Reuter.
\newblock {\em Transaction Processing: Concepts and Techniques}.
\newblock Morgan Kaufmann, 1993.

\bibitem{DBLP:conf/focs/HopcroftPV75}
{John E.} Hopcroft, {Wolfgang J.} Paul, and {Leslie G.} Valiant.
\newblock On time versus space and related problems.
\newblock In {\em 16th Annual Symposium on Foundations of Computer Science,
  Berkeley, California, USA, October 13-15, 1975}, pages 57--64. {IEEE}
  Computer Society, 1975.
\newblock \href {https://doi.org/10.1109/SFCS.1975.23}
  {\path{doi:10.1109/SFCS.1975.23}}.

\bibitem{DBLP:journals/cacm/Knuth74}
Donald~E. Knuth.
\newblock {Computer Programming as an Art}.
\newblock {\em Commun. {ACM}}, 17(12):667--673, 1974.
\newblock \href {https://doi.org/10.1145/361604.361612}
  {\path{doi:10.1145/361604.361612}}.

\end{thebibliography}


\begin{thebibliography}{10}

\bibitem{quorumWP}
\href
  {https://github.com/ConsenSys/quorum/blob/master/docs/Quorum%20Whitepaper%20v0.2.pdf}
  {Quorum whitepaper}.
\newblock Technical report, ConsenSys, 2018.

\bibitem{Abraham2018}
Ittai Abraham, Guy Gueta, and Dahlia Malkhi.
\newblock \href {http://arxiv.org/abs/1803.05069} {Hot-stuff the linear,
  optimal-resilience, one-message {BFT} devil}.
\newblock {\em CoRR}, abs/1803.05069, 2018.

\bibitem{bar-sosp}
Amitanand~S. Aiyer, Lorenzo Alvisi, Allen Clement, Mike Dahlin, Jean-Philippe
  Martin, and Carl Porth.
\newblock \href {http://dx.doi.org/10.1145/1095810.1095816} {{BAR} fault
  tolerance for cooperative services}.
\newblock In {\em 20\textsuperscript{th} {ACM} Symp.~on Operating System
  Principles (SOSP)}, SOSP '05, pages 45--58, 2005.

\bibitem{Alchieri2008}
Eduardo~A. Alchieri, Alysson~Neves Bessani, Joni Silva~Fraga, and Fab\'{\i}ola
  Greve.
\newblock \href {http://dx.doi.org/10.1007/978-3-540-92221-6_4} {Byzantine
  consensus with unknown participants}.
\newblock In {\em Conference on Principles of Distributed Systems}, OPODIS
  2008, pages 22--40, 2008.

\bibitem{bessani2014state}
Alysson Bessani, Jo{\~a}o Sousa, and Eduardo~EP Alchieri.
\newblock \href
  {http://www.di.fc.ul.pt/~bessani/publications/dsn14-bftsmart.pdf} {State
  machine replication for the masses with bft-smart}.
\newblock In {\em Dependable Systems and Networks (DSN), 2014 44th Annual
  IEEE/IFIP International Conference on}, pages 355--362. IEEE, 2014.

\bibitem{Blum2019}
Erica Blum, Jonathan Katz, and Julian Loss.
\newblock Synchronous consensus with optimal asynchronous fallback guarantees.
\newblock In Dennis Hofheinz and Alon Rosen, editors, {\em Theory of
  Cryptography}, pages 131--150, 2019.

\bibitem{hyperledger}
Tamas Blummer, Sean Bohan, Mic Bowman, Christian Cachin, Nick Gaski, Nathan
  George, Gordon Graham, Daniel Hardman, Ram Jagadeesan, Travin Keith, Renat
  Khasanshyn, Murali Krishna, Tracy Kuhrt, Arnaud~Le Hors, Jonathan Levi,
  Stanislav Liberman, Esther Mendez, Dan Middleton, Hart Montgomery, Dan
  O'Prey, Drummond Reed, Stefan Teis, Dave Voell, Greg Wallace, and Baohua
  Yang.
\newblock \href
  {https://www.hyperledger.org/wp-content/uploads/2018/08/HL_Whitepaper_IntroductiontoHyperledger.pdf}
  {An introduction to hyperledger}.
\newblock Technical report, Hyperledger, 2018.

\bibitem{Bracha1983}
Gabriel Bracha and Sam Toueg.
\newblock \href {http://dx.doi.org/10.1145/800221.806706} {Resilient consensus
  protocols}.
\newblock In {\em 2\textsuperscript{nd} {ACM} Symp.~on Principles of
  Distributed Computing}, PODC '83, pages 12--26, 1983.

\bibitem{Cachin2013}
Christian Cachin and Michael Backes.
\newblock \href {http://dx.doi.org/10.1109/DSN.2003.1209914} {Reliable
  broadcast in a computational hybrid model with byzantine faults, crashes, and
  recoveries}.
\newblock In {\em Conference on Dependable Systems and Networks (DSN)},
  page~37, jun 2003.

\bibitem{Cachin2019}
Christian Cachin and Bj{\"{o}}rn Tackmann.
\newblock \href {http://dx.doi.org/10.4230/LIPIcs.OPODIS.2019.7} {Asymmetric
  distributed trust}.
\newblock In Pascal Felber, Roy Friedman, Seth Gilbert, and Avery Miller,
  editors, {\em 23rd International Conference on Principles of Distributed
  Systems, {OPODIS} 2019, December 17-19, 2019, Neuch{\^{a}}tel, Switzerland},
  volume 153 of {\em LIPIcs}, pages 7:1--7:16. 2019.

\bibitem{calder2011}
Brad Calder, Ju~Wang, Aaron Ogus, Niranjan Nilakantan, Arild Skjolsvold, Sam
  McKelvie, Yikang Xu, Shashwat Srivastav, Jiesheng Wu, Huseyin Simitci, et~al.
\newblock Windows {Azure} {Storage}: a highly available cloud storage service
  with strong consistency.
\newblock In {\em 23\textsuperscript{rd} {ACM} Symp.~on Operating System
  Principles (SOSP)}, pages 143--157. ACM, 2011.

\bibitem{osdi99}
Miguel Castro and Barbara Liskov.
\newblock Practical {B}yzantine fault tolerance.
\newblock In {\em 3\textsuperscript{rd} {USENIX} Symp.~on Operating Systems
  Design and Implementation (OSDI)}, February 1999.

\bibitem{pbft}
Miguel Castro and Barbara Liskov.
\newblock Practical {B}yzantine fault tolerance and proactive recovery.
\newblock {\em ACM Trans.\@ on Computer Systems}, 20:2002, 2002.

\bibitem{Cavin2004}
David Cavin, Yoav Sasson, and Andr{\'{e}} Schiper.
\newblock \href {http://dx.doi.org/10.1007/978-3-540-28634-9\_11} {Consensus
  with unknown participants or fundamental self-organization}.
\newblock In Ioanis Nikolaidis, Michel Barbeau, and Evangelos Kranakis,
  editors, {\em Ad-Hoc, Mobile, and Wireless Networks {(ADHOC-NOW)}}, volume
  3158 of {\em Lecture Notes in Computer Science}, pages 135--148. 2004.

\bibitem{Chase2018}
Brad Chase and Ethan MacBrough.
\newblock \href {http://arxiv.org/abs/1802.07242} {Analysis of the {XRP} ledger
  consensus protocol}.
\newblock {\em CoRR}, abs/1802.07242, 2018.

\bibitem{clement2009upright}
Allen Clement, Manos Kapritsos, Sangmin Lee, Yang Wang, Lorenzo Alvisi, Mike
  Dahlin, and Taylor Riche.
\newblock \href
  {http://www.cs.utexas.edu/users/lorenzo/papers/clement-sosp09.pdf} {Upright
  cluster services}.
\newblock In {\em 22\textsuperscript{nd} {ACM} Symp.~on Operating System
  Principles (SOSP)}, pages 277--290, 2009.

\bibitem{corbett2013}
James~C. Corbett, Jeffrey Dean, Michael Epstein, Andrew Fikes, Christopher
  Frost, Jeffrey~John Furman, Sanjay Ghemawat, Andrey Gubarev, Christopher
  Heiser, Peter Hochschild, et~al.
\newblock Spanner: Google’s globally distributed database.
\newblock {\em ACM Transactions on Computer Systems (TOCS)}, 31(3):8, 2013.

\bibitem{Correia2006}
Miguel Correia, Nuno Neves, and Paulo Ver{\'i}ssimo.
\newblock From consensus to atomic broadcast: Time-free byzantine-resistant
  protocols without signatures.
\newblock {\em Comput. J.}, 49:82--96, 01 2006.

\bibitem{ScalingDecentralizedBlockchains}
Kyle Croman, Christian Decker, Ittay Eyal, Adem~Efe Gencer, Ari Juels, Ahmed
  Kosba, Andrew Miller, Prateek Saxena, Elaine Shi, Emin G{\"u}n~Sirer, Dawn
  Song, and Roger Wattenhofer.
\newblock \href {https://fc16.ifca.ai/bitcoin/papers/CDE+16.pdf} {On scaling
  decentralized blockchains}.
\newblock In Jeremy Clark, Sarah Meiklejohn, Peter~Y.A. Ryan, Dan Wallach,
  Michael Brenner, and Kurt Rohloff, editors, {\em Financial Cryptography and
  Data Security}, pages 106--125, 2016.

\bibitem{gallet2011}
Carole Delporte-Gallet, Hugues Fauconnier, Rachid Guerraoui, and Andreas
  Tielmann.
\newblock \href {http://dx.doi.org/10.1007/s00446-010-0122-4} {The disagreement
  power of an adversary}.
\newblock {\em Distributed Computing}, 24:137--147, November 2011.

\bibitem{Dwork1988}
Cynthia Dwork, Nancy Lynch, and Larry Stockmeyer.
\newblock \href {http://dx.doi.org/10.1145/42282.42283} {Consensus in the
  presence of partial synchrony}.
\newblock {\em J. ACM}, 35(2):288–323, April 1988.

\bibitem{Fischer82b}
M.~J. Fischer, N.~A. Lynch, and M.~S. Paterson.
\newblock Impossibility of distributed consensus with one faulty process.
\newblock {\em Journal of the ACM}, 32(2):374--382, April 1985.
\newblock Also published as MIT Laboratory of Science Technical Report
  MIT/LCS/TR-282, Cambridge, MA, 1982.

\bibitem{Floyd1962}
Robert~W. Floyd.
\newblock \href {http://dx.doi.org/10.1145/367766.368168} {Algorithm 97:
  Shortest path}.
\newblock {\em Commun. ACM}, 5(6):345--, June 1962.

\bibitem{ethereum}
Ethereum Foundation.
\newblock \href {https://github.com/ethereum/wiki/wiki/White-Paper} {Ethereum
  white paper}.
\newblock Technical report, Ethereum Foundation, 2018.

\bibitem{GarcaPrez2018FederatedBQ}
{\'A}lvaro Garc{\'i}a-P{\'e}rez and Alexey Gotsman.
\newblock \href {http://dx.doi.org/10.4230/LIPIcs.OPODIS.2018.17} {{Federated
  Byzantine Quorum Systems}}.
\newblock In Jiannong Cao, Faith Ellen, Luis Rodrigues, and Bernardo Ferreira,
  editors, {\em 22nd International Conference on Principles of Distributed
  Systems (OPODIS 2018)}, volume 125 of {\em Leibniz International Proceedings
  in Informatics (LIPIcs)}, pages 17:1--17:16, 2018.

\bibitem{Gifford79}
David~K. Gifford.
\newblock Weighted voting for replicated data.
\newblock In {\em 7\textsuperscript{th} {ACM} Symp.~on Operating System
  Principles (SOSP), ACM Operating Systems Review}, pages 150--162, December
  1979. ACM SIGOPS.

\bibitem{Greve2007}
Fab{\'\i}ola Greve and Sebastien Tixeuil.
\newblock \href {http://dx.doi.org/10.1109/DSN.2007.61} {Knowledge connectivity
  vs. synchrony requirements for fault-tolerant agreement in unknown networks}.
\newblock In {\em 37th Annual IEEE/IFIP International Conference on Dependable
  Systems and Networks (DSN'07)}, pages 82--91, June 2007.

\bibitem{grpc}
\href {https://grpc.io} {grpc: A high performance, open-source universal {RPC}
  framework}.
\newblock \url{https://grpc.io}, 2018.

\bibitem{guerraoui2007}
Rachid Guerraoui and Marko {Vukoli{\'c}}.
\newblock \href {http://dx.doi.org/10.1145/1281100.1281120} {Refined quorum
  systems}.
\newblock In {\em Principles of Distributed Computing}, PODC 2007, pages
  119--128, 2007.

\bibitem{corda}
Mike Hearn and Richard~Gendal Brown.
\newblock \href {https://www.r3.com/reports/corda-technical-whitepaper/}
  {Corda: A distributed ledger}.
\newblock Technical report, r3, 2019.

\bibitem{survivor-sets}
Flavio Junqueira and Keith Marzullo.
\newblock Designing algorithms for dependent process failures.
\newblock In {\em Workshop on Future Directions in Distributed Computing},
  pages 24--28, 2003.

\bibitem{paxos}
Leslie Lamport.
\newblock \href {http://dx.doi.org/10.1145/279227.279229} {The {P}art-time
  {P}arliament}.
\newblock {\em ACM Trans.\@ on Computer Systems}, 16(2):133--169, May 1998.

\bibitem{paxos-made-simple}
Leslie Lamport.
\newblock \href
  {https://www.microsoft.com/en-us/research/publication/paxos-made-simple/}
  {Paxos made simple}.
\newblock Technical report, Microsoft Research, December 2001.

\bibitem{byzantizing-paxos}
Leslie Lamport.
\newblock \href {https://lamport.azurewebsites.net/tla/byzsimple.pdf}
  {Byzantizing {Paxos} by refinement}.
\newblock In {\em 25th Int'l Conf. on Distributed Computing {(DISC)}}, pages
  211--224, 2011.

\bibitem{Lamport82}
Leslie Lamport, Robert Shostak, and Marshall Pease.
\newblock The {B}yzantine {G}enerals {P}roblem.
\newblock {\em ACM Trans.\@ on Programming Languages and Systems},
  4(3):382--401, July 1982.

\bibitem{LAMP}
Butler~W. Lampson and Howard~E. Sturgis.
\newblock \href
  {http://research.microsoft.com/en-us/um/people/blampson/21-crashrecovery/Abstract.html}
  {Crash recovery in a distributed data storage system}.
\newblock Technical report, Xerox Palo Alto Research Center, Palo Alto, CA,
  1979.

\bibitem{Liu2015XFTPF}
Shengyun Liu, Paolo Viotti, Christian Cachin, Vivien Qu{\'e}ma, and Marko
  Vukolic.
\newblock \href
  {https://www.usenix.org/system/files/conference/osdi16/osdi16-liu.pdf} {Xft:
  Practical fault tolerance beyond crashes}.
\newblock In {\em OSDI}, 2016.

\bibitem{stellarSOSP}
Marta Lokhava, Giuliano Losa, David Mazi{\`e}res, Graydon Hoare, Nicolas P~E
  Barry, Eli Gafni, Jonathan Jov{\'e}, Rafa{\l} Malinowsky, and J~Murphy
  McCaleb.
\newblock \href
  {http://delivery.acm.org/10.1145/3360000/3359636/p80-lokhava.pdf} {Fast and
  secure global payments with {Stellar}}.
\newblock In {\em 27\textsuperscript{th} {ACM} Symp.~on Operating System
  Principles (SOSP)}, 2019.

\bibitem{StellarDISC}
Giuliano Losa, Eli Gafni, and David Mazi{\`e}res.
\newblock Stellar consensus by instantiation.
\newblock In {\em DISC}, 2019.

\bibitem{cobalt}
Ethan MacBrough.
\newblock \href {http://arxiv.org/abs/1802.07240} {Cobalt: {BFT} governance in
  open networks}.
\newblock {\em CoRR}, abs/1802.07240, 2018.

\bibitem{Malkhi2019}
Dahlia Malkhi, Kartik Nayak, and Ling Ren.
\newblock \href {http://dx.doi.org/10.1145/3319535.3354225} {Flexible byzantine
  fault tolerance}.
\newblock In {\em 26\textsuperscript{th} ACM Conf.\@~on Computer and
  Communications Security (CCS)}, CCS 2019, pages 1041--1053, 2019.

\bibitem{Malkhi97a}
Dahlia Malkhi and Michael Reiter.
\newblock {B}yzantine quorum systems.
\newblock In {\em 29th ACM Symposium on Theory of Computing}, pages 569--578,
  May 1997.

\bibitem{mazieresstellar}
David Mazi{\`e}res.
\newblock \href {https://www.stellar.org} {The {S}tellar consensus protocol: A
  federated model for internet-level consensus}.
\newblock \url{https://www.stellar.org}, April 2015.

\bibitem{Miller2016}
Andrew Miller, Yu~Xia, Kyle Croman, Elaine Shi, and Dawn Song.
\newblock \href {http://dx.doi.org/10.1145/2976749.2978399} {The honey badger
  of {BFT} protocols}.
\newblock In {\em 23\textsuperscript{rd} ACM Conf.\@~on Computer and
  Communications Security (CCS)}, pages 31--42, 2016.

\bibitem{bitcoin}
Satoshi Nakamoto.
\newblock Bitcoin: A peer-to-peer electronic cash system, 2008.

\bibitem{collision-resistance}
Bart Preneel.
\newblock \href {http://dx.doi.org/10.1007/978-1-4419-5906-5_565} {Collision
  resistance}.
\newblock In Henk C.~A. van Tilbor and Sushil Jajodia, editors, {\em
  Encyclopedia of Cryptography and Security}, pages 221--222, 2011.

\bibitem{protobufs}
\href {https://developers.google.com/protocol-buffers/} {Protocol buffers}.
\newblock \url{https://developers.google.com/protocol-buffers/}, 2018.

\bibitem{Schwartz2014}
David Schwartz, Noah Youngs, and Arthur Britto.
\newblock \href {https://ripple.com/files/ripple_consensus_whitepaper.pdf} {The
  {Ripple} protocol consensus algorithm}.
\newblock Technical report, Ripple Labs Inc, 2014.

\bibitem{CharlotteTR}
Isaac Sheff, Xinwen Wang, Haobin Ni, Robbert van Renesse, and Andrew~C. Myers.
\newblock \href {http://arxiv.org/abs/1905.03888} {Charlotte: Composable
  authenticated distributed data structures, technical report}, 2019.

\bibitem{hetconstechreport}
Isaac~C. Sheff, Robbert van Renesse, and Andrew~C. Myers.
\newblock \href {http://arxiv.org/abs/1412.3136} {Distributed protocols and
  heterogeneous trust: Technical report}.
\newblock Technical Report arXiv:1412.3136, Cornell University Computer and
  Information Science, December 2014.

\bibitem{Siu1998}
Hin-Sing Siu, Yeh-Hao Chin, and Wei-Peng Yang.
\newblock \href {http://dx.doi.org/10.1109/71.667895} {Byzantine agreement in
  the presence of mixed faults on processors and links}.
\newblock {\em Parallel and Distributed Systems, IEEE Transactions on},
  9(4):335--345, Apr 1998.

\end{thebibliography}
\fi
\checkpagelimit

\end{document}